\documentclass[11pt]{article}
\usepackage{complexity}

\usepackage[english]{babel}
\usepackage[utf8x]{inputenc}
\usepackage[T1]{fontenc}
\usepackage{enumitem}
\usepackage[margin=1.in,marginparwidth=1.75cm]{geometry}

\usepackage[algo2e, ruled, vlined]{algorithm2e}
\usepackage{amsmath}
\usepackage{amsfonts}
\usepackage{graphicx}
\usepackage{amsthm}
\usepackage{amssymb}
\usepackage{dsfont}
\usepackage[algo2e, ruled, vlined]{algorithm2e}
\usepackage{bbm}

\usepackage[colorinlistoftodos]{todonotes}
\usepackage[colorlinks=true, allcolors=blue]{hyperref}

\theoremstyle{theorem}
\newtheorem{theorem}{Theorem}
\newtheorem{claim}{Claim}
\newtheorem{lemma}{Lemma}
\newtheorem{corollary}{Corollary}

\newtheorem{remark}{Remark}
\newenvironment{reminder}[1]{\bigskip
	\noindent {\bf Reminder of #1.}\em}{\smallskip}

\newenvironment{subproof}[1][\proofname]{%
  \begin{proof}[#1]%
}{%
  \end{proof}%
}

\theoremstyle{definition}
\newtheorem{definition}{Definition}

\newenvironment{proofsketch}{\begin{proof}[Proof sketch]}{\end{proof}}

\newcommand{\COLLIDE}{\mathsf{COLLIDE}}
\newcommand{\calA}{\mathcal{A}}

\newcommand{\calE}{{\mathcal{E}}}

\newcommand{\calH}{{\mathcal{H}}}

\newcommand{\calT}{{\mathcal{T}}}

\renewcommand{\AC}{\mathsf{AC}}

\newcommand{\Otilde}{\widetilde{O}}

\newcommand{\bits}{{\{0,1\}}}

\newcommand{\Out}{\mathsf{Out}}

\newcommand{\level}{\mathsf{level}}
\newcommand{\walk}{\mathsf{walk}}
\newcommand{\Mwalk}{\mathsf{Mwalk}}
\newcommand{\stdwalk}{\mathsf{stdwalk}}
\newcommand{\extwalk}{\mathsf{extwalk}}
\newcommand{\good}{\mathsf{good}}
\newcommand{\width}{\mathsf{wd}}
\newcommand{\std}{\mathsf{std}}
\newcommand{\probe}{\mathsf{probe}}
\newcommand{\pre}{\mathsf{pre}}
\DeclareMathOperator\supp{supp}
\DeclareMathOperator*{\Ex}{\mathbb{E}}

\newenvironment{CenterBox}
    {\begin{center}
    \begin{tabular}{p{0.9\textwidth}}
    }
    { 
    \end{tabular} 
    \end{center}
    }

\title{Time-Space Tradeoffs for Element Distinctness and Set Intersection via Pseudorandomness}
\author{Xin Lyu\thanks{Department of EECS, University of California at Berkeley. Email: \texttt{xinlyu@berkeley.edu}} \and Weihao Zhu\thanks{Department of Computer Science, Shanghai Jiao Tong University. Email: \texttt{zhuweihao@sjtu.edu.cn}}}

\begin{document}

\maketitle

\pagenumbering{gobble}

\begin{abstract}
In the \textsc{Element Distinctness} problem, one is given an array $a_1,\dots, a_n$ of integers from $[\mathrm{poly}(n)]$ and is tasked to decide if $\{a_i\}$ are mutually distinct. Beame, Clifford and Machmouchi (FOCS 2013) gave a low-space algorithm for this problem that runs in space $S(n)$ and time $T(n)$ where $T(n) \le \Otilde(n^{3/2}/S(n)^{1/2})$, assuming a random oracle (i.e., random access to polynomially many random bits). A recent breakthrough by Chen, Jin, Williams and Wu (SODA 2022) showed how to remove the random oracle assumption in the regime $S(n) = \polylog(n)$ and $T(n) = \Otilde(n^{3/2})$. They designed the first truly $\polylog(n)$-space, $\Otilde(n^{3/2})$-time algorithm by constructing a small family of hash functions $\calH \subseteq \{h | h:[\mathrm{poly}(n)]\to [n]\}$ with a certain pseudorandom property.

In this paper, we give a significantly simplified analysis of the pseudorandom hash family by Chen \emph{et al}. Our analysis clearly identifies the key pseudorandom property required to fool the BCM algorithm, allowing us to explore the full potential of this construction. Based on our new analysis, we show the following.

\begin{itemize}
    \item As our main result, we give a time-space tradeoff for \textsc{Element Distinctness} without random oracle. Namely, for every $S(n),T(n)$ such that $T\approx \Otilde(n^{3/2}/S(n)^{1/2})$, our algorithm can solve the problem in space $S(n)$ and time $T(n)$. Our algorithm also works for a related problem \textsc{Set Intersection}, for which this tradeoff is tight due to a matching lower bound by Dinur (Eurocrypt 2020).
    \item As a direct application of our technique, we show a more general pseudorandom property of the hash family, which we call the ``$c$-connecting'' property. It might be of independent interest.
    % which might be of independent interest. 
    \item The construction by Chen \emph{et al.}~needs $O(\log^3 n\log\log n)$ random bits to sample the pseudorandom hash function. We slightly improve the seed length to $O(\log^3 n)$.

\end{itemize}

\end{abstract}

% The \textsc{Element Distinctness} problem asks one to decide whether a given array of $n$ integers contains pairwise distinct elements. In problem, Given two integer sets represented by two (not necessarily sorted) arrays, the \textsc{Set Intersection} and we need to print all the elements in their intersection.
%which can fool the cycle-finding procedure by Beame, Clifford, and Machmouchi \cite{DBLP:conf/focs/BeameCM13}. 

\newpage
\pagenumbering{arabic}

\section{Introduction}

\emph{Time} and \emph{space} are arguably the two most fundamental measures of computational resources in the theory of computation. Time-space tradeoff studies the (im)possibility of solving computational problems simultaneously in low space and time. Among many problems studied in this area, \textsc{Element Distinctness} is a fundamental one.

\begin{CenterBox}
\textsc{Element Distinctness}: Given (read-only random access to) a list of $n$ positive integers $a_1, \dots, a_n$, each taking value in $[1,\mathrm{poly}(n)]$, decide whether all $a_i$'s are distinct.
\end{CenterBox}
\medskip\noindent\textbf{Upper bounds.} We start our discussion with upper bounds. If we have $O(n)$ words (each word has $O(\log n)$ bits) of working space, we can solve the problem just by sorting. This costs $\Otilde(n)$\footnote{For brevity, throughout this article, we use $\Otilde(f)$ to hide $\mathrm{polylog}(f,n)$ factors.} time
but $\Omega(n)$ space. On the other hand, we can enumerate every pair of indices to check if there is a colliding pair. This costs $\Omega(n^2)$ time but only requires $O(1)$ space. More generally, by applying a low-space sorting algorithm  \cite{DBLP:journals/tcs/MunroP80, DBLP:conf/focs/PagterR98}, \textsc{Element Distinctness} can be solved in time $T(n)$ and space $S(n)$ for every $T(n), S(n)$ such that $T(n)\cdot S(n) \ge n^2$. For \emph{comparison-based} model, this is known to be optimal \cite{DBLP:journals/siamcomp/BorodinFHUW87,DBLP:journals/siamcomp/Yao94}.
%\footnote{We always use $\Otilde(\cdot)$ to hide $\polylog(n)$ factors.} 

Surprisingly, if we consider non-comparison-based algorithms, a better tradeoff can be achieved. In 2013, Beame, Clifford, and Machmouchi \cite{DBLP:conf/focs/BeameCM13} gave a non-comparison-based algorithm (which we call the BCM algorithm henceforth) for \textsc{Element Distinctness} with time-space tradeoff $T(n)\le \Otilde(n^{3/2}/S(n)^{1/2})$, \emph{assuming a random oracle}: read-only random access to polynomially many random bits that do not count towards the space complexity. This assumption is strong: it assumes that there is a large table of random bits to which the algorithm has free and random access. A more practical model is called \emph{one-way} access to random bits. That is, the algorithm can request a random bit (i.e., toss a coin) at each time, but cannot read random bits requested in the past (unless it stores the bits in its own working memory).

\medskip\noindent\textbf{The CJWW algorithm.} For the case that $S(n) = \polylog(n)$, the random oracle assumption was removed in a recent breakthrough by Chen, Jin, Williams and Wu \cite{ChenJWW22}, who gave the first truly low-space algorithm for \textsc{Element Distinctness} that beats pairwise-comparison: their algorithm runs in $O(\polylog(n))$ space and $\Otilde(n^{3/2})$ time, with typical one-way access to random bits. 

% In more detail, Chen \emph{et al.}~achieve this result by designing a pseudorandom hash family $\calH\subseteq \{h: [m]\to [n]\}$, which is samplable using $O(\log^3 n\log\log n)$ random bits. This also implies that every $h\in \calH$ admits a short description (namely, the random seed used to generate $h$). They then used a random $\mathbf{h}\sim \calH$ to replace the random oracle required by the BCM algorithm, and showed that this does not degrade the success probability of the BCM algorithm too much.

%Finally, they obtained the first truly $\polylog(n)$-space algorithm that runs in time $\Otilde{O}(n^{3/2})$.

% \paragraph*{The CJWW algorithm.} Here we briefly introduce how CJWW works. \xin{TODO: briefly explain how CJWW works.}

However, the algorithm by Chen \emph{et al.}~does not naturally extend to a smooth time-space tradeoff as the BCM algorithm does. The analysis in \cite{ChenJWW22} is involved and carefully tailored to the case $S(n) = O(\polylog(n))$. Therefore, it is not clear from their proof whether the case $S(n) = O(\polylog(n))$ and $S(n) = n^{\Omega(1)}$ have any inherent difference. It was asked as an open question in \cite{ChenJWW22} whether their algorithm can be generalized to a tradeoff result.

\medskip\noindent\textbf{Lower bounds.} One might also wonder whether the tradeoff $T(n)\cdot S(n)^{1/2} \approx \Theta(n^{3/2})$ is tight, even assuming a random oracle. While the current technique fails to give a matching lower bound (the best lower bound for \textsc{Element Distinctness} is barely superlinear \cite{DBLP:journals/toc/Ajtai05}), for a closely related problem, the same algorithmic idea still applies, and a \emph{matching} lower bound exists. We consider the \textsc{Set Intersection} problem.

\begin{CenterBox}
\textsc{Set Intersection}: Given two integer sets represented by two (not necessarily sorted) input lists $(a_1,\dots, a_n)$, $(b_1,\dots, b_n)$ which are promised to not contain duplicates, print all the elements in the intersection $A\cap B$.
\end{CenterBox}

Note that \textsc{Set Intersection} can be seen as a variant of \textsc{Element Distinctness}: now the task is to find collisions between two lists, and the algorithm needs to output \emph{all collisions}. The BCM algorithm applies to \textsc{Set Intersection} and gives the same tradeoff $T(n)\cdot S(n)^{1/2} = \Theta(n^{3/2})$.

The work by Patt-Shamir and Peleg \cite{DBLP:journals/tcs/Patt-ShamirP93} and by Dinur \cite{DBLP:conf/eurocrypt/Dinur20-lowerbound} showed that any $S(n)$-space algorithm for \textsc{Set Intersection} must use time $T(n) \ge \Omega(n^{3/2} / S(n)^{1/2})$, even if the printed elements can be in any order, and each element in $A\cap B$ is allowed to be printed multiple times. Therefore, assuming a \emph{random oracle}, we conclude that $T(n)\cdot S(n)^{1/2}\approx \Theta(n^{3/2})$ is the optimal time-space tradeoff for \text{Set Intersection}.

Chen \emph{et al.}~\cite{ChenJWW22} also gave a randomized algorithm for \textsc{Set Intersection} that runs in time $O(n^{3/2})$ and space $O(\polylog(n))$. Given the lower bound, the running time is optimal in the regime $S(n) = O(\polylog(n))$. Still, it is open whether one can generalize their algorithm to achieve the optimal time-space tradeoff.

\medskip\noindent\textbf{The CJWW pseudorandom hash family.} The main technical tool behind the CJWW algorithm is a construction of a small hash family $\calH\subseteq \{h: [m]\to [n]\}$, which is samplable using $O(\log^3 n\log\log n)$ random bits. This also implies that every $h\in \calH$ admits a short description (namely, the random seed used to generate $h$). They then use a random $\mathbf{h}\sim \calH$ to replace the random oracle required by the BCM algorithm, and show that it does not degrade the success probability of the algorithm too much (thanks to a certain pseudorandom property).

The pseudorandom hash family by Chen \emph{et al.}~is based on the iterative restriction framework \cite{DBLP:journals/acr/AjtaiW89}. Roughly speaking, the iterative restriction framework starts by assuming that $\mathbf{h}:[m]\to [n]$ is truly random function, and gradually fills in entries of $\mathbf{h}$ with pseudorandom elements. In pseudorandomness literature, people have utilized this methodology to construct pseudorandom generators (PRG) for various computational models \cite{DBLP:conf/focs/GopalanMRTV12, DBLP:conf/coco/TrevisanX13, DBLP:journals/siamcomp/HaramatyLV18, DBLP:journals/toc/LeeV20, DBLP:conf/focs/ForbesK18, DBLP:conf/stoc/MekaRT19-width3}. However, in all these PRG results, the statistical tests considered are not adaptive, in the sense that the target circuit/program always reads its input in a \emph{pre-defined} pattern. 

Remarkably, the algorithm by Chen \emph{et al.}~provides a non-trivial example, showing that the iterative restriction construction can fool some classes of highly-adaptive tests, where the future query to the (pseudorandom) hash function $\mathbf{h}$ heavily depends on the previous responses returned by $\mathbf{h}$. This motivates us to further explore the capability of pseudorandom objects constructed by the iterative restriction framework.

\subsection{Our Results}

The main result of this paper is an affirmative answer to the open questions left by \cite{ChenJWW22}. Namely, we extend the algorithm by Chen \emph{et al.}~to the whole time-space tradeoff, matching the tradeoff offered by the BCM algorithm. 
Our result is mainly based on an improved analysis of the pseudorandom hash family constructed by Chen \emph{et al.} ~\cite{ChenJWW22}. Along the way, we also prove some new pseudorandom properties of the hash family, which seem difficult to establish by other means. Therefore, we think the analysis may be of independent interest to the pseudorandomness community. We elaborate on our contributions below.

\medskip\noindent\textbf{Algorithmic results.} We start with the algorithm side. First, we show an improved time-space tradeoff for \textsc{Element Distinctness} and \textsc{Set Intersection} that beats sorting.

\begin{theorem}\label{theo:element-distinctness-algo}
For every complexity bound $S(n), T(n): \mathbb{N}\to \mathbb{N}$ such that $S(n)^{1/2}\cdot T(n) \ge n^{1.5}$, there is a Monte Carlo algorithm solving \textsc{Element Distinctness} in time $O(T(n)\cdot \polylog(n))$ and space $O(S(n)\cdot \polylog(n))$ with one-way access to random bits. Moreover, when there is a colliding pair, the algorithm reports one with high probability.
\end{theorem}

\begin{theorem}\label{theo:set-intersection-algo}
For every complexity bound $S(n), T(n): \mathbb{N}\to \mathbb{N}$ such that $S(n)^{1/2}\cdot T(n) \ge n^{1.5}$, there is a Monte Carlo algorithm solving \textsc{Set Intersection} in time $O(T(n)\cdot \polylog(n))$ and space $O(S(n)\cdot \polylog(n))$ with one-way access to random bits. The algorithm prints elements in no particular order, and the same element may be printed multiple times.
\end{theorem}

% We view our result as somewhat a milestone. 
In light of Dinur's lower bounds \cite{DBLP:conf/eurocrypt/Dinur20-lowerbound}, Theorem~\ref{theo:set-intersection-algo} is optimal up to polylogarithmic factors. Even if we believe the tradeoff given by Theorem~\ref{theo:element-distinctness-algo} is not tight, new ideas are required to separate \textsc{Element Distinctness} from its multi-output variant \textsc{Set Intersection}. Using current techniques, Theorem~\ref{theo:element-distinctness-algo} seems hard to improve, even allowing random oracles or considering the non-uniform low-space model (i.e. branching programs).

In the extremely low-space regime, the algorithm by Chen \emph{et al.}~\cite{ChenJWW22} needs $\Omega(\log^3 n\log\log n)$ bits of working memory to store the seed for the hash function. As a byproduct of our improved analysis, we reduce the seed length to $O(\log^3 n)$ bits. Consequently, now we only need $O(\log^3 n)$ bits of working space to start beating the trivial pairwise-comparison algorithm.

\begin{theorem}\label{theo:improved-seed}
Both \textsc{Element Distinctness} and \textsc{Set Intersection} can be solved by a Monte Carlo algorithm that runs in $\Otilde(n^{3/2})$ time, uses $O(\log^3 n)$ bits of working space and no random oracle.
\end{theorem}

\medskip\noindent\textbf{Pseudorandomness results.} Before discussing the new pseudorandomness result, we briefly review the BCM algorithm. Let $a:[n]\to [m]$ be a mapping with only one colliding pair $a_p = a_q, p \ne q$ (that is, except for $a_p = a_q$, all other $a_i$'s are distinct). For every hash function $h:[m]\to [n]$, define a $1$-out digraph $G_{a,h}$ on the vertex set $[n]$ with edge set $\{(x, h(a_x)) \}$. For every $x\in [n]$, let $\Out_{a,h}(x)$ denote the set of vertices reachable from $x$ in $G_{a,h}$. Alternatively, $y\in \Out_{a,h}(x)$ if and only if $y = (h\circ a)^{(s)}(x)$ for some $s\ge 0$. Consider sampling a truly random hash function $\mathbf{h}:[m]\to [n]$ and a starting vertex $\mathbf{x}\in [n]$.
%\footnote{We always use bold letters (e.g., $\mathbf{X}$) to denote random variables.}.
We have (by the birthday paradox):
\begin{align}
& \mathbb{E}_{\mathbf{h}, \mathbf{x}} [|\Out_{a,\mathbf{h}}(\mathbf{x})|] \le O(\sqrt{n}). \label{eq:size-bound} \\
& \Pr_{\mathbf{h},\mathbf{x}} [p,q\in \Out_{a,\mathbf{h}}(\mathbf{x})] \ge \Omega\left( \frac{1}{n} \right). \label{eq:2-connect}
\end{align}

Think of $a$ as the input array of an \textsc{Element Distinctness} instance. If $p,q\in \Out_{a,\mathbf{h}}(\mathbf{x})$, $a_p=a_q$ implies that $p$ and $q$ point to the same vertex in $G_{a,\mathbf{h}}$. Thus, we can find the pair $(p, q)$ by running Floyd's cycle finding algorithm (see, e.g., \cite{DBLP:books/aw/Knuth81, Pollard1975AMC}) on $G_{a,\mathbf{h}}$ with starting vertex $\mathbf{x}$, which costs $O(\sqrt{n})$ time and $O(\polylog(n))$ space. 

The BCM algorithm \cite{DBLP:conf/focs/BeameCM13} runs $\Otilde(n)$ independent trials of the cycle-finding procedure, each with a fresh hash. Since each trial succeeds with probability $\Omega(1/n)$, at least one trial succeeds in finding $(p, q)$ with high probability. By~\eqref{eq:size-bound}, the expected running time of one trial is $O(\sqrt{n})$, which brings the total running time to $\Otilde(n^{1.5})$ while the space complexity is $O(\polylog(n))$.

% If we want to replace the truly random hash $\mathbf{h}$ with a small pseudorandom hash family $\calH$. Then $\calH$ has to possess the following pseudorandom property.

Inspired by the idea behind the BCM algorithm, we formulate the notion of ``$c$-connecting property'' for pseudorandom hash families.

\begin{definition}\label{def:connecting}
Let $\calH \subseteq \{h|h:[m]\to [n]\}$ be a family of hash functions. We say $\calH$ is $c$-connecting, if for every injective\footnote{Our technique can also deal with non-injective mappings and derive bounds that depend on the number of colliding pairs in the mapping. However, we only consider injective mappings for simplicity.} mapping $a:[n]\to [m]$ and every $c$ fixed vertices $1\le u_1 < u_2 < \dots < u_c \le n$, it holds that
\[
\Pr_{\mathbf{h}\sim \calH, \mathbf{x}\sim [n]}[\forall i\in [c], u_i\in \Out_{a,\mathbf{h}}(\mathbf{x})] \ge \Omega_c\left( n^{-c/2} \right).
\]
\end{definition}

For $m \ge n^2$, the probability bound $\Omega\left( n^{-c/2} \right)$ is the best we can hope for: it is easy to show that no hash family can achieve $\omega\left( n^{-c/2} \right)$ (see Appendix~\ref{sec:connecting-ub}). Note that a hash family $\calH$ has to be $2$-connecting to drive the BCM algorithm. Thus, the main result of \cite{ChenJWW22} can be viewed as a construction of a small $2$-connecting hash family. As a byproduct of our analysis, we can generalize the construction to obtain small $c$-connecting hash families for all constant $c\ge 2$.

\begin{theorem}\label{theo:connecting-intro}
For every constant $c \ge 2$, the following is true. For every $m\ge n\ge 2$, there is a $c$-connecting pseudorandom hash family $\calH \subseteq \{h| h:[m]\to [n]\}$. The seed length to sample a function from $\calH$ is $O(\log^2(n)\log(m))$.
\end{theorem}

Besides being interesting in its own right, we hope Theorem~\ref{theo:connecting-intro} could also raise interest to study more pseudorandom properties of the $1$-out (pseudo-)random digraphs induced by the hash construction.

\section{Proof Overview}\label{sec:proof-overview}

In this section, we discuss the main idea behind our proof. We start with the construction of the pseudorandom hash family $\calH$. 

\medskip\noindent\textbf{The construction.} We present the pseudorandom hash construction below. Our construction slightly simplifies the one in \cite{ChenJWW22}, which is in turn inspired by the \emph{iterative restriction} framework \cite{DBLP:journals/acr/AjtaiW89}.

\begin{itemize}
    \item Let $t = \frac{1}{2}\log n$ and $\kappa = C \log n$, where $C$ is a sufficiently large constant.
    \item Sample $\mathbf{h}_1,\dots, \mathbf{h}_t$. Each $\mathbf{h}_i: [m]\to [n]\cup \{0\}$ is a hash function satisfying the following.
    \begin{itemize}
        \item For every $j\in [m], v\in [n]$, $\Pr_{\mathbf{h}_i}[\mathbf{h}_i(j) = 0] = \frac{1}{2}$ and $\Pr_{\mathbf{h}_i}[\mathbf{h}_i(j) = v] = \frac{1}{2n}$.
        \item $\mathbf{h}_i:[m]\to [n]\cup \{0\}$ is $\kappa$-wise independent\footnote{One way to sample such $\mathbf{h}_i$ is to first sample a $\kappa$-wise independent function $\mathbf{h}_i:[m]\to [2n]$ using the standard method, and then identify $[n+1,2n]$ with $0$.}.
    \end{itemize}
    % We also sample $\mathbf{h}_t:[m]\to [n]$ to be a $k$-wise independent hash function.
    \item Define the final hash $\mathbf{h}:[m]\to [n]$ as follows. For every $j\in [m]$, we find the smallest $q \le t$ such that $\mathbf{h}_q(j) \ne 0$ and define $\mathbf{h}(j) := \mathbf{h}_q(j)$. If no such $q$ exists, we define $\mathbf{h}(j) := 1$.
\end{itemize}

\medskip\noindent\textbf{Setup.} Recall the statistical test our hash needs to fool: After sampling a random $h$ (from either $\calH$ or other distributions), the BCM algorithm starts from a random $\mathbf{x}_1\sim [n]$ and walks on $G_{a,h}$ by iterating $\mathbf{x}_{i+1} = h(a_{\mathbf{x}_i})$ until reaching a loop. Namely, the algorithm finds $\mathbf{x}_{\mathbf{B}+1} = \mathbf{x}_{j}$ for some $1\le j \le \mathbf{B}$ (Note that $\mathbf{B}$ is a random variable depending on $\mathbf{h}$ and $\mathbf{x}_1$). Suppose $p,q\in [n]^2$ is the colliding pair in the input array (i.e., $a_p = a_q$). To lower bound the success probability of the BCM algorithm, we wish to argue that
\begin{align}
    \Pr[p,q\in \{ \mathbf{x}_{i} \}_{1\le i\le \mathbf{B}} ] \ge \Omega\left( \frac{1}{n} \right). \label{eq:goal-uv}
\end{align}

If the hash function $\mathbf{h}$ is truly random, it is easy to establish \eqref{eq:goal-uv}. On the other extreme, suppose $\mathbf{h}$ is only $K$-wise independent for some small $K\le n^{o(1)}$. In this case, after observing $\mathbf{x}_1,\dots, \mathbf{x}_{K}, \mathbf{x}_{K+1}$, the next vertex $\mathbf{x}_{K+2} = \mathbf{h}(a_{\mathbf{x}_{K+1}})$ may be highly correlated with $(\mathbf{x}_{i})_{1\le i\le K+1}$, because the sequence $(\mathbf{x}_1,\dots, \mathbf{x}_{K+1})$ implies $K$ input-output pairs for the hash (i.e., $\mathbf{h}(a_{\mathbf{x}_i}) = \mathbf{x}_{i+1}$). Since we only assume $\mathbf{h}$ is $K$-wise independent, $\mathbf{x}_{K+2}$ might even be uniquely determined by the length-$(K+1)$ walk history $(\mathbf{x}_1,\dots, \mathbf{x}_{K+1})$.

\subsection{Parallelizing the Sequential Walk}

%While the ``randomness-refreshing'' argument might be a nice way to explain the intuition, it turns out that it is not the correct way to reason about the pseudorandom walk sequence. %This is because the ``randomness-refreshing'' argument fails to break the ``sequential'' nature of the walk. If we observe the walking sequence from beginning to the end, it is not clear 
% Instead, we present a ``communication''-based perspective about the walk sequence. While it does not imply any technical observations directly, this perspective turns out to be very useful for analyzing the walk.

Let us revisit the issue when we try to use a barely $K$-wise independent hash function $\mathbf{h}$ to run the cycle-finding procedure. Being $K$-wise independent only promises to provide $K$ random elements when we query $K$ entries in $\mathbf{h}$ that are \emph{independent} of $\mathbf{h}$\footnote{More rigorously, the first query is independent of $\mathbf{h}$, and the next $K-1$ queries only depend on the results to previous queries.}. Unfortunately, due to the sequential nature of the random walk, the future queries to $\mathbf{h}$ may be heavily dependent on $\mathbf{h}$ itself. This is the key barrier one has to overcome to prove \eqref{eq:goal-uv}.

The iterative restriction construction offers a nice structure to break the sequential nature of the random walk. In a very high level, given the hierarchical construction $(\mathbf{h}_t,\mathbf{h}_{t-1},\dots, \mathbf{h}_1)$, for each $d\in [t]$, the entries to which we query $\mathbf{h}_{d}$ are mostly determined by $(\mathbf{h}_{t},\dots, \mathbf{h}_{d+1}$), and are nearly independent of $\mathbf{h}_d$ itself. In the following, we formalize this claim by considering a ``communication'' perspective of the random walk.

%We give the details below.
% Roughly speaking, towards proving \eqref{eq:goal-uv}, to break the barrier we met with $K$-wise independent hash function, we have to find a way to break the ``sequential nature'' of the random walk.

\medskip\noindent\textbf{An alternative view of the random walk.} We view the $t$ hash functions $\mathbf{h}_1,\dots, \mathbf{h}_t$ as $t$ parties, each holding one level of the hash. We also view the cycle-finding algorithm as one party. Collectively, the $(t+1)$ parties wish to generate a walk sequence for the cycle-finding procedure. By definition, they produce the walk sequence by the following protocol.

\begin{itemize}
    \item The algorithm samples the starting vertex $\mathbf{x}_1\sim [n]$ and sends it to $\mathbf{h}_{t}$.
    \item For each $d\in [t]$, when $\mathbf{h}_d$ receives a vertex $\mathbf{x}$: It first passes $\mathbf{x}$ to $\mathbf{h}_{d-1}$, and asks $\mathbf{h}_{d-1}$ (together with $\mathbf{h}_{<d-1}$) to generate a sequence starting at $\mathbf{x}$. After $\mathbf{h}_{d-1}$ returns a vertex $\mathbf{x'}$ such that $\mathbf{h}_{\le d-1}(a_\mathbf{x'})\equiv 0$. $\mathbf{h}_d$ queries $\mathbf{h}(a_{\mathbf{x'}})$: it either moves a step $\mathbf{x''} = \mathbf{h}_{d}(a_{\mathbf{x'}})$ and passes $\mathbf{x''}$ down to $\mathbf{h}_{d-1}$, or it finds that $\mathbf{h}_{d}(a_{\mathbf{x'}}) = 0$ and returns $\mathbf{x'}$ to the higher level $\mathbf{h}_{d+1}$.
\end{itemize}

One might find the protocol a bit counter-intuitive: by the definition of $h$, to compute $h(\mathbf{x}_1)$, we need to find out the \emph{smallest} $d\in [t]$ such that $h_d(\mathbf{x}_1)\ne 0$. Hence, the most natural choice seems to be sending $\mathbf{x}_1$ to $h_1$ first. However, our top-down protocol is essential in the proof. Intuitively, in a length-$L$ walk, for each $i\in[t]$, the top $i$ levels of hash functions $h_{t},\dots, h_{t-i+1}$ make roughly $\frac{L}{2^i}$ steps of the walk. These $\frac{L}{2^i}$ steps partition the walk sequence into segments, where each segment consists of steps done by $h_1,\dots, h_{i-1}$. One can see that inside the walk sequence there is an implicit hierarchical structure with respect to $h_t,\dots, h_1$ (i.e., the higher level makes fewer steps). The top-down protocol makes the hierarchical structure explicit. Also note that this protocol corresponds to Algorithm~\ref{algo:standard-walk} in the formal proof. 

% \medskip

\medskip\noindent\textbf{Simplifying assumptions.} Directly analyzing the process above seems very difficult. To gain some intuition, we make two unrealistic assumptions for now.

% To ease the illustration, we make two unrealistic assumptions for now.

\begin{enumerate}
    \item Each time $\mathbf{h}_{d}$ wants to access an entry $\mathbf{h}_d(a_{\mathbf{x}})$, it has never queried $\mathbf{h}_d(a_{\mathbf{x}})$ before. %(The walk never visits the same entry twice (i.e., let $\mathbf{x}_1,\mathbf{x}_2\dots,$ be the walk sequence. Then $\{a_{\mathbf{x}_i}\}_{i\ge 1}$ are mutually distinct). 
%     \item For each $d \in [t]$ and every $y_1,y_2,\dots,$ with probability $1$ over $\mathbf{h}_d$, there exists $i < \infty$ such that $\mathbf{h}_d(a_{y_i}) = 0$. %For a first-time reading, we recommend thinking of $n$ as very large compared with $t$ (e.g., $n \approx 2^{2^{2^{t}}}$), so that the assumption is almost true.
    \item The hash functions $\mathbf{h}_1,\dots, \mathbf{h}_t$ are truly-random (instead of bounded-wise independent).
\end{enumerate}

\medskip\noindent\textbf{Parallelizing the walk.} Given the two assumptions, the process of generating the random walk is equivalent to the following ``parallelized'' process (The equivalence might be hard to see in the first time of reading. See the ``Digest'' paragraph below for an explanation).

\begin{itemize}
    \item The algorithm samples $\mathbf{x}_1$ and passes it to $\mathbf{h}_t$. It asks $\mathbf{h}_t$ (together with $\mathbf{h}_{<t}$) to generate a walk starting at $\mathbf{x}_1$.
    \item $\mathbf{h}_t$, upon receiving $\mathbf{x}_1$, samples $q\sim \mathrm{Geom}(1/2)$\footnote{Recall $\mathrm{Geom}(1/2)$ is the geometric distribution: $\Pr[\mathrm{Geom}(1/2) = k] = 2^{-k}$ for every $k\ge 1$.} and $q-1$ vertices $\mathbf{x}_2,\dots, \mathbf{x}_{q} \sim [n]$. It then passes $\mathbf{x}_1,\mathbf{x}_2,\dots, \mathbf{x}_q$ down to $\mathbf{h}_{t-1}$. Intuitively, $\mathbf{h}_{t}$ asks $\mathbf{h}_{\le t-1}$ to sample $q$ random walks, starting at $\mathbf{x}_1,\dots, \mathbf{x}_q$.
    \item In the decreasing order of $d=(t-1),\dots, 1$: $\mathbf{h}_{d}$ receives $c_d$ vertices from $\mathbf{h}_{d+1}$ ($c_d$ is a random variable depending on $\mathbf{h}_{d+1},\dots,\mathbf{h}_{t}$ and $\mathbf{x}_1$). Then, $\mathbf{h}_d$ re-labels these vertices as $\mathbf{x}_1,\dots, \mathbf{x}_{c_d}$. For each $i\in [c_d]$, it samples $q_i\sim \mathrm{Geom}(1/2)$, and $q_i - 1$ vertices $\mathbf{x}_{i,2},\dots, \mathbf{x}_{i,q_i}$. It also sets $\mathbf{x}_{i,1} := \mathbf{x}_i$. After that, it sends the vertices $(\mathbf{x}_{i,j})_{i\in [c_d], j\in [q_i]}$ to $\mathbf{h}_{d-1}$, asking $\mathbf{h}_{\le d-1}$ to generate $\left|\sum_{i}q_i\right|$ walk sequences with the given starting points.
    
    In particular, when $d=1$, $\mathbf{h}_1$ just outputs these vertices (there is no $\mathbf{h}_0$). They constitute the final walk sequence. 
\end{itemize}

\medskip\noindent\textbf{Digest.} Under the two simplifying assumptions, we claim that the walk sequence generated by the parallelized process is identically distributed as the original walk. 
To see this, note that in the original sequence, when $\mathbf{h}_t$ passes $\mathbf{x}_1$ down to $\mathbf{h}_{t-1}$, it needs to wait for $\mathbf{h}_{\le t-1}$ to do the walk and send back the vertex $\mathbf{x'}$ such that $\mathbf{h}_{\le t-1}(a_{\mathbf{x'}}) \equiv 0$\footnote{Given Assumption $1$ that we are always visiting new entries, and Assumption $2$ that each new entry of $\mathbf{h}_{t-1}$ is $0$ with probability $\frac{1}{2}$, such $\mathbf{x'}$ exists with probability $1$.}. 
However, under Assumption~$1$, we can tell for sure that $\mathbf{h}_t(a_{\mathbf{x'}})$ must have not been reached before, 
and under Assumption~$2$, we know that $\mathbf{h}_t(a_{\mathbf{x'}})$ is a random element from $[n]\cup \{0\}$ with $\Pr[\mathbf{h}_t(a_{\mathbf{x'}})=0] = 1/2$. 

Therefore, without knowing exactly what $\mathbf{x}'$ is, we can tell that $\mathbf{h}_{t}$ sends the second query to $\mathbf{h}_{t-1}$ with probability $\frac{1}{2}$. 
When this does happen, it sends a uniformly random vertex from $[n]$. Moreover, the next time $\mathbf{h}_{t}$ receives a vertex $\mathbf{x}''$ from $\mathbf{h}_{t-1}$, the same argument applies: with probability $\frac{1}{2}$, $\mathbf{h}_{t}$ finds that $\mathbf{h}_{t}(a_{\mathbf{x}''})\ne 0$, walks a step, and sends a new vertex down to $\mathbf{h}_{t-1}$. Otherwise it ends this walk sequence.
In general, independently of $\mathbf{h}_1,\dots, \mathbf{h}_{t-1}$, the number of vertices that $\mathbf{h}_{t}$ passes down to $\mathbf{h}_{t-1}$ obeys $q\sim \mathrm{Geom}(1/2)$, 
and each vertex is uniformly random in $[n]$. 
The same argument also applies to lower level hashes $\mathbf{h}_{t-1},\dots, \mathbf{h}_1$.

We call this new process a parallelized walk, because (under two assumptions) this process breaks the ``sequential'' natural of the random walk. Intuitively, for each $d\in [t]$, $\mathbf{h}_{d}$ receives (roughly) $c_d\approx 2^{t-d}$ starting vertices from $\mathbf{h}_{d+1}$ and is asked to generate $c_d$ sequences. Knowing all these starting vertices, $\mathbf{h}_d$ can process these $c_d$ requests ``in parallel''. After finishing its job, $\mathbf{h}_{d}$ will then pass $c_{d-1}\approx 2^{t-d+1}$ starting vertices to the next level in a batch. One can visualize the parallelized walk as a hierarchical structure with $t$ levels $(\mathbf{h}_t,\dots, \mathbf{h}_1)$ and roughly $2^t = \sqrt{n}$ branches (see Figure~\ref{fig:tree-example} for an illustration). In the following, we denote this structure as the walk tree (see Section~\ref{sec:cjww} for its formal definition).

Next, we will explain how to remove the two simplifying assumptions and prove the desired probability lower bound \eqref{eq:goal-uv}. 

\medskip\noindent\textbf{Lower-bounding \eqref{eq:goal-uv} and removing Assumption $2$.} We first explain how to remove Assumption $2$ and sketch our strategy for proving \eqref{eq:goal-uv}. Having parallelized the walk, this is indeed straightforward. Since we only want to know if a given pair $(p,q)$ have been reached in the walk, we can enumerate two ``branches'' in the walk tree, and observe if they hit $p, q$. Since we have parallelized the walk sequence, each vertex in the walk tree only depends on its ancestors. Recall $\kappa = O(\log n)$ and $\Pr[\mathrm{Geom}(1/2) > \kappa/10] < \frac{1}{\mathrm{poly}(n)}$. With high probability, each $\mathbf{h}_{d}$ makes no more than $\frac{\kappa}{10}$ steps in a branch. Hence, using $\kappa$-wise independent hash functions suffices for fooling the observation of two branches.

% For a simple example, suppose $t=3$. Suppose $\mathbf{h}_{3}$ sampled $q = 3$. Then $\kappa$-wise independence TODO.

Also note that by setting $\kappa$ larger by a constant factor, the same argument holds even if we observe a constant number of branches in the walk tree. Namely, if we only observe $O(1)$ branches in the walk tree, we cannot distinguish between the case that $\mathbf{h}_i$'s are truly random and the case that they are barely $\kappa$-wise independent. In the following, we refer to this property as ``constant-wise independence'' of the walk tree. This observation is important for the $c$-connecting property, as well as the tradeoff result.

Since we have set $t\approx \frac{1}{2}\log n$, there will be roughly $2^{t}$ branches in the walk tree. We have shown that these branches appear to be ``constant-wise'' independent. Therefore, the probability that two of those branches hit $p,q$ is roughly
\[
\frac{2^{t}}{n}\cdot \frac{2^{t}}{n} \approx \frac{1}{n}
\]
as desired.

\medskip\noindent\textbf{Removing Assumption~$1$.} Now we explain the idea for removing Assumption~$1$, which is the most technical part of the proof. To begin with, note that Assumption $1$ is indeed true for a prefix of the walk sequence (namely, before the walk reaches a vertex $\mathbf{x}_{i}$ such that $a_{\mathbf{x}_{i}}=a_{\mathbf{x}_{j}}$ for some $j < i$).

There is an ``obvious'' way to achieve Assumption $1$: during the random walk, if some $\mathbf{h}_{d}$ attempts to query an entry $\mathbf{h}_{d}(a_{\mathbf{x}_{i}})$ for the second (or more) time, it samples a new element from $[n]\cup \{0\}$ to replace $\mathbf{h}_d(a_{\mathbf{x}_{i}})$. In this way, effectively, the hash function is always accessing new entries, and we can lower-bound the probability of hitting both $p$ and $q$ by $\frac{2^{2t}}{n^2} \approx \frac{1}{n}$. However, in this process we also count in some invalid contribution: if the walk visits the same vertex $\mathbf{x}$ twice before hitting $p$ and $q$, then we know the contribution from this walk is ``invalid''. To subtract invalid contribution, we enumerate $i,j,u,v$ such that $i,j<\max(u,v)$, and calculate the probability that $a_{\mathbf{x}_{i}} = a_{\mathbf{x}_j}$ and $(\mathbf{x}_{u},\mathbf{x}_{v})=(p,q)$. In this process, we only observe $4$ branches in the walk tree, by the constant-wise independence of the walk tree, we can (roughly) upper bound the probability by $\frac{1}{n^3}$ (Fixing $a_{\mathbf{x}_i}$, we have $a_{\mathbf{x}_j} = a_{\mathbf{x}_i}$ w.p. $\frac{1}{n}$, and $(\mathbf{x}_{u},\mathbf{x}_{v})=(p,q)$ w.p. $\frac{1}{n^2}$). Summing up $(i,j,u,v)$, the amount of invalid contribution is bounded by $\frac{2^{4t}}{n^3}$. Therefore, $\frac{2^{2t}}{n^2} - \frac{2^{4t}}{n^3}$ would be a valid lower bound for $\Pr[p,q\in \{\mathbf{x}_i\}]$. In the actual proof, we set $t = \frac{1}{2}\log n - 5$. This gives $\frac{2^{2t}}{n^2} - \frac{2^{4t}}{n^3} \ge (2^{-10}-2^{-20})\cdot \frac{1}{n} \ge \Omega(1/n)$, as desired.

Unfortunately, there is a subtle but critical flaw in this argument: the ``resampling'' operation breaks the $\kappa$-wise independence of the hash function! Consider the following $T$-round interaction between a $\kappa$-wise independent hash function $\mathbf{h}_d$ and an adversary $\calA$.

\begin{itemize}
    \item In the $i$-th round, based on the interaction history, $\calA$ chooses and sends $y\in [m]$ to $\mathbf{h}_d$. $\mathbf{h}_d$ sends $\mathbf{h}_d(y)$ back to $\calA$. Meanwhile, $\mathbf{h}_d$ also resets $\mathbf{h}_d(y)$ to a uniformly random element.
\end{itemize}

If $\mathbf{h}_d$ was sampled as a truly-random hash, $\mathbf{h}_d$ would stay truly-random after the interaction. However, if $\mathbf{h}_d$ is from a $\kappa$-wise independent distribution, then $\calA$ might cause $\mathbf{h}_d$ to be highly biased after the interaction. Back to our example, it is not clear how $\mathbf{h}_{d}$ is interacting with $\mathbf{h}_{\le d-1}$. In the worst case, if $\mathbf{h}_{\le d-1}$ interacts with $\mathbf{h}_d$ ``adversarially'' and $\mathbf{h}_d$ chooses to resample for all queries, then $\mathbf{h}_{d}$ will no longer be $\kappa$-wise independent after the first few rounds of communication.

We are ready to introduce our final idea, which is the key to simplify the analysis. Since describing the idea precisely requires quite a bit technical work, we only present the high-level idea here, and refer interested readers to Section~\ref{sec:cjww} for its detail (in particular, see Algorithm~\ref{algo:extended-walk} and Lemma~\ref{lemma:extend-is-random}). Roughly speaking, we show that we do not need to do the ``resampling'' for every query to the hash functions. Recall that to prove the lower bound, we only need to observe two branches in the walk tree. We show that, we only need to do ``resample'' in the ``observed part'' of the walk tree. There are two intuitions for this idea. First, in order to prove that the observed part is random, we do not really care about the remaining part of the walk tree. Therefore, only resampling for the observed part suffices to establish the proof. Second, in the observed part, each hash function is queried for at most $\kappa$ times, the $\kappa$-wise independence suffices to ensure the randomness of query results (in particular, the ``adversarial attack'' issue mentioned above does not exist anymore). %This is only a high-level description. For interested readers, we refer to Algorithm~\ref{algo:extended-walk} and Lemma~\ref{lemma:extend-is-random} for the implementation of the idea.

\subsection{Comparison with CJWW}

Our idea of using iterative restriction construction is directly inspired by \cite{ChenJWW22}. Therefore, it is worthwhile to compare our results with theirs. While our technical analysis shares some similarities with \cite{ChenJWW22}, in order to simplify the \cite{ChenJWW22} analysis to the best possible extent and extend it to our new applications (i.e., the tradeoff result and the $c$-connecting property), we need a host of new ideas, both conceptual ones and technical ones.

%\begin{itemize}
\medskip\noindent\textbf{Conceptual idea.} Conceptually, we propose to view the walk sequence as generated by ``communication between $t$ independent hash functions''. While this conceptual idea does not bring any immediate technical consequence, it is nevertheless crucial for obtaining the final proof. The main structure that the iterative restriction framework offers us is that the $t$ building blocks $\mathbf{h}_1,\dots, \mathbf{h}_t$ are independent. It turns out the ``communication'' perspective is a desired way to exploit this structure. Essentially, when we ``parallelize'' the random walk, we are reducing the many-round communication between $\mathbf{h}_i$'s to a one-way, top-down communication from $\mathbf{h}_{t}$ down to $\mathbf{h}_{1}$, which is the key for the final proof. 
    
Moreover, from the communication perspective, it is easy to show that we can use $\kappa$-wise \emph{almost} independent hash functions to generate $\mathbf{h}_1,\dots,\mathbf{h}_d$. This is a claim that seems difficult to obtain through the analysis by \cite{ChenJWW22} (in fact, this answers a question posed in \cite{ChenJWW22}). Although this observation does not reduce the overall seed length\footnote{We need $\Theta(\log n)$ bits to describe an entry of the hash $h:[m]\to [n]$. Therefore, requiring $\mathbf{h}$ to be almost $\kappa$-wise independent reduces to constructing an almost $\Theta(\kappa \log n)$-wise independent binary string, which brings the total seed length to $\Theta(\kappa \log n)$. On the other hand, sampling a perfect $\kappa$-wise independent hash also requires $\Theta(\kappa (\log n+\log m))$ random bits.}, it adds one more evidence suggesting that the communication perspective provides a fairly powerful ``conceptual method'' to analyze pseudorandom objects that consist of several independent building blocks. We hope this idea can help understand more pseudorandom properties of the iterative restriction framework.

The concept of ``communication'' is not new in the pseudorandomness literature (See, e.g., \cite{DBLP:conf/stoc/ImpagliazzoNW94, DBLP:journals/jcss/NisanZ96, DBLP:conf/focs/ImpagliazzoMZ12, DBLP:journals/siamcomp/BravermanRRY14}). However, in all the previous PRG analyses, the messages in communication are always short. For example, to prove the extractor-based PRG for read-once branching program, the standard method works by splitting the program into two halves, and arguing that the ``message'' passed from the first half to the second is so small that we can use an extractor to ``refresh'' the random seed (see \cite{DBLP:journals/siamcomp/BravermanRRY14}). There, it is easy to see that the communication consists of only one short message. In contrast, the cycle-finding procedure considered in our work involves interactive communications among the $d$ levels of hash functions. As a result, the number of messages exchanged may be unbounded. It is highly non-trivial to reduce the analysis to a one-way, top-down communication protocol.

\medskip\noindent\textbf{Technical idea.} Just like \cite{ChenJWW22}, we have to deal with a lot of technicalities to implement all these ideas and intuitions. Along the way, we simplify the proof by \cite{ChenJWW22} from various aspects. Among these simplifications, the most crucial one is the introduction of a new ``extended random walk'' (see Algorithm~\ref{algo:extended-walk}), which is the instantiation of the ``resample-as-you-observe'' idea mentioned in the last section. It saves a lot of case-by-case analysis as was needed in the previous proof, and lends itself well to extensions: having established all necessary machinery in Section~\ref{sec:cjww}, the $c$-connecting property admits a fairly straightforward proof. For the tradeoff result, only one more idea is required: see Section~\ref{sec:tech-extension}.

Finally, we mention that in the formal proof (Section~\ref{sec:cjww}), we borrow some terminologies (walk tree, indexing) from \cite{ChenJWW22} to stay aligned with previous work.

% \subsubsection{Comparison with CJWW} 

% Our idea of using iterative restriction construction is directly inspired by \cite{ChenJWW22}. Therefore, it is worthwhile to compare our results with theirs. First, we share the same high-level intuition (as explained in Section~\ref{sec:const-intuition}) with \cite{ChenJWW22} regarding the construction. 

% However, to turn the intuition into a formal proof, Chen \emph{et al.}~employed a carefully tailored case-by-case analysis to argue the $2$-connecting property of the construction. In contrast, our proof uses a more systematic approach to analyze the walk on the pseudorandom graph. In particular, we view our design of the extended walks as a key contribution of this work. It not only allows us to recover the CJWW result in a straightforward way, but also yields many new results, including the tradeoff result, the general $c$-connecting property, and the reduced seed length. We believe ideas and techniques developed in our work can help understand more properties of pseudorandom objects constructed from the iterative restriction framework.

% We also remark that the ideas of tracking vertices in the walk by $t$-dimensional vectors and designing extended walks are similar to related ideas in \cite{ChenJWW22} (i.e., their ``dependency tree'', ``indexing scheme'', and their notion of ``extended walk''). However, the extended walks we propose are both easier to analyze and more powerful for deriving the pseudorandomness results.

\subsection{Extensions}\label{sec:tech-extension}

From the communication perspective, we have shown that the parallelized walk sequence appears random if we only observe a constant number of vertices in the walk (i.e., the ``constant-wise independent'' property of the walk tree). This allows us to prove the $c$-connecting property (Theorem~\ref{theo:connecting-intro}) easily. However, to obtain the tradeoff result, new ingredients are required.

% Since our result shows that the walk on the pseudorandom hash function is ``constant-wise independent'' (albeit on the walk tree, not in the typical walk sequence), we can easily generalize our analysis to prove the $c$-connecting property (Theorem~\ref{theo:connecting-intro}). However, to obtain the tradeoff result, new ingredients are required.

\medskip\noindent\textbf{Review of the BCM tradeoff.} Let us first review the tradeoff given by the BCM algorithm \cite{DBLP:conf/focs/BeameCM13}. Let $k\le n$. Suppose we have access to a random hash $h: [m]\to [n]$ and $\Otilde(k)$ bits of working memory. In this case, the BCM algorithm randomly selects $k$ starting vertices and produces $k$ walks. We assume $a_p = a_q$ is the only colliding pair in the input array. \cite{DBLP:conf/focs/BeameCM13} shows that, if there are two walks hitting $p$ and $q$, the algorithm can find the pair $(p,q)$ with $\Otilde(k)$ bits of working memory. \cite{DBLP:conf/focs/BeameCM13} further shows that in a single trial, the algorithm succeeds in hitting both $p$ and $q$ with probability $\Omega\left(\frac{k}{n} \right)$, while the running time of one trial is bounded by $\Otilde(\sqrt{kn})$. Running $\Otilde(n/k)$ trials independently succeeds in finding the pair with high probability. The total running time is $\Otilde(n^{3/2}/\sqrt{k})$. 

Now, to fool the tradeoff algorithm with the pseudorandom hash family, it seems we need at least $k$-wise independence to simultaneously monitor the $k$ walks. While we can store the seed for $k$-wise independent hash functions (we have $\Otilde(k)$ working space now), we cannot afford the $\Omega(k)$ time overload to evaluate a $k$-wise independent hash function.

\subsubsection*{A ``local'' analysis of the BCM tradeoff}

We show $k$-wise independence is \emph{not} necessary. In particular, we observe that there is a ``local'' proof of the BCM tradeoff. Roughly speaking, by ``local'' we mean that to lower bound the success probability in one trial, one only needs to consider some simple probabilities that only involve a constant number of vertices in the $k$ walks. Consequently, the hash family constructed from $O(\log n)$-wise independent primitives yields the same $\Omega(k/n)$ success probability, allowing us to remove the random oracle assumption.

% For example, recall that in proving~\eqref{eq:informal-argue}, we only need to consider $\Pr[\mathbf{x}_i = u \land \mathbf{x}_j = v]$ for every $1\le i < j \le B$\footnote{This is actually an oversimplification, as $B$ is also a random variable depending on the hash. Nevertheless, it is not hard to give a strictly localized proof. }. 

%\paragraph*{``Localizing'' the analysis.} 
To illustrate the idea, we sketch the ``localized proof'' for truly random hash $\mathbf{h}:[m]\to [n]$ below. Suppose $\mathbf{x}^1,\dots, \mathbf{x}^k$ are the $k$ random starting vertices. Let $c > 0$ be a sufficiently small constant. For each starting vertex $\mathbf{x}^i$, consider the first $L = c\sqrt{n/k}$ steps of move in the $i$-th walk. Namely, for each $i\in [k]$, we consider 
\[
\mathbf{x}^i_1 =\mathbf{x}^i, ~~~ \mathbf{x}^i_2 = \mathbf{h}(a_{\mathbf{x}^i_1}),~~~ \dots,~~~ \mathbf{x}^i_L = \mathbf{h}(a_{\mathbf{x}^i_{L-1}}).
\]
Note that it takes $O(\sqrt{nk})$ time to go through these sequences (this corresponds to the running time of the BCM algorithm). Our goal is to prove $\Pr[p, q\in \{ \mathbf{x}^{i}_j \}_{i\in [k], j\in [L]}] \ge \Omega(k/n)$.

\medskip\noindent\textbf{A coupling-based proof.} For the purpose of analysis, we imagine $k$ sequences $(\mathbf{y}^i_1,\dots, \mathbf{y}^i_L), \forall i\in [k]$ that are jointly distributed with $\mathbf{x}^i_j$ and are defined as follows. For each $i\in [k]$, we set $\mathbf{y}^i_1 = \mathbf{x}^i_1$. For each $j=1,\dots, L-1$, if $a_{\mathbf{x}^i_j} = a_{\mathbf{x}^{i'}_{j'}}$ for some $(i',j') < (i, j)$ (in the lexicographical order), we set $\mathbf{y}^i_{j+1},\dots, \mathbf{y}^i_{L}$ as uniform and independent elements from $[n]$ and complete the construction for the $i$-th sequence. Otherwise let $\mathbf{y}^i_{j+1} = \mathbf{x}^i_{j+1}$.

If $\mathbf{h}$ is a truly random hash, the list $(\mathbf{y}^{i}_j )_{i\in [k], j\in [L]}$ contains $kL$ independent and uniform elements. This is because every $\mathbf{y}^i_j$ is obtained by either querying a new entry in $\mathbf{h}$ or sampling a uniformly random element. For every pair $((i_1,j_1), (i_2,j_2))\in ([k]\times [L])^2$, let $\calE^1(i_1,j_1,i_2,j_2)$ denote the event that all of following hold. 
\begin{enumerate}
    \item First, $\mathbf{y}^{i_1}_{j_1} = p$ and $\mathbf{y}^{i_2}_{j_2} = q$.
    \item For every $(i_1, j_3), (i_4, j_4)$ such that $j_3 < j_1$ and $(i_4, j_4) < (i_1, j_3)$, we have $a_{\mathbf{y}^{i_1}_{j_3}} \ne a_{\mathbf{y}^{i_4}_{j_4}}$.
    \item For every $(i_2, j_5), (i_6, j_6)$ such that $j_5 < j_2$ and $(i_6, j_6) < (i_2, j_5)$, we have $a_{\mathbf{y}^{i_2}_{j_5}} \ne a_{\mathbf{y}^{i_6}_{j_6}}$.
    \item For all $(i_7, j_7) < (i_1, j_1)$, it holds $\mathbf{y}^{i_7}_{j_7}\ne p$. For all $(i_8, j_8) < (i_2, j_2)$, it holds $\mathbf{y}^{i_8}_{j_8}\ne q$.
\end{enumerate}
Here, Conditions 2 and 3 ensure that $\mathbf{y}^{i_1}_{j_1} = \mathbf{x}^{i_1}_{j_1}$ and $\mathbf{y}^{i_2}_{j_2} = \mathbf{x}^{i_2}_{j_2}$. Condition $4$ ensures the events $\{\calE^1(i_1,j_1,i_2,j_2)\}_{i_1,j_1,i_2,j_2}$ are mutually disjoint. Moreover, we observe that each $\calE^{1}(i_1,j_1,i_2,j_2)$ implies $\mathbbm{1}\{p, q\in \{ \mathbf{x}^{i}_j \}_{i\in [k], j\in [L]}\}$. Hence, we may conclude that
\[
\Pr[p, q\in \{ \mathbf{x}^{i}_j \}_{i\in [k], j\in [L]}] \ge \sum_{(i_1,j_1),(i_2,j_2)} \Pr[\calE^1(i_1,j_1,i_2,j_2)].
\]
To lower bound $\Pr[\calE^{1}(i_1,j_1,i_2,j_2)]$, we further decompose it into even simpler events by applying a union bound:
\[
\begin{aligned}
 \Pr[\calE^1(i_1,j_1,i_2,j_2)] 
& \ge \Pr[\mathbf{y}^{i_1}_{j_1} = p \land \mathbf{y}^{i_2}_{j_2} = q] - \\
&~~~~ \sum_{j_3, (i_4, j_4)}  \Pr[\mathbf{y}^{i_1}_{j_1} = p \land \mathbf{y}^{i_2}_{j_2} = q \land a_{\mathbf{y}^{i_1}_{j_3}} = a_{\mathbf{y}^{i_4}_{j_4}}] - \\
&~~~~ \sum_{j_5, (i_6, j_6)}  \Pr[\mathbf{y}^{i_1}_{j_1} = p \land \mathbf{y}^{i_2}_{j_2} = q \land a_{\mathbf{y}^{i_2}_{j_5}} = a_{\mathbf{y}^{i_6}_{j_6}}] - \\
&~~~~ \sum_{(i_7,j_7)}  \Pr[\mathbf{y}^{i_1}_{j_1} = p \land \mathbf{y}^{i_2}_{j_2} = q \land \mathbf{y}^{i_7}_{j_7} = p ] - \\
&~~~~ \sum_{(i_8,j_8)}  \Pr[\mathbf{y}^{i_1}_{j_1} = p \land \mathbf{y}^{i_2}_{j_2} = q \land \mathbf{y}^{i_8}_{j_8} = q ] \\
& \ge \frac{1}{n^2} - O\left(\frac{k L^2}{n^3}\right) - O\left( \frac{kL}{n^3} \right) \\
& \ge \frac{1}{n^2} - O\left(\frac{c^2}{n^2}\right).
\end{aligned}
\]
In the summation above, the enumeration of indices ($j_3, i_4, j_4,$ etc.) follows the rule specified by Condition $2$-$4$, which we omit for brevity. 

Choosing $c > 0$ to be small enough, we can lower bound $\Pr[\calE^1(i_1,j_1,i_2,j_2)]$ by $\Omega(1/n^2)$. Finally, we may take a summation over all $(i_1,j_1), (i_2, j_2)$ to get:
\[
\Pr\left[p, q\in \{ \mathbf{x}^{i}_j \}_{i\in [k], j\in [L]} \right] \ge \sum_{(i_1,j_1),(i_2,j_2)} \Pr[\calE^1(i_1,j_1,i_2,j_2)] \ge \Omega\left( \frac{k^2L^2}{n^2} \right) \ge \Omega\left( \frac{k}{n} \right).
\]

\medskip\noindent\textbf{Conclusion.} Note that the proof above is highly local, in the sense that we reduce the task of lower-bounding $\Pr[p, q\in \{\mathbf{x}^i_j\}]$ to analyzing a collection of simpler events, each involving only a constant number of vertices. To show the same lower bound for the pseudorandom hash, we use the same high-level proof strategy. But we will work with the walk tree and use our Lemma~\ref{lemma:extend-is-random}.

\subsection{Future Directions}

% In this work, we show new time-space tradeoff for \textsc{Element Distinctness} and \textsc{Set Intersection} with typical one-way access to random bits. Our technique is from the theory of pseudorandomness.

Our work raises several directions for further research. We highlight two of them below. We also refer interested readers to \cite{ChenJWW22} for discussions about more related work.

\medskip\noindent\textbf{The power of iterative restriction.} The iterative restriction approach was first developed by Ajtai and Wigderson in their seminal work \cite{DBLP:journals/acr/AjtaiW89}, where they gave the first non-trivial pseudorandom generator (PRG) for constant-depth Boolean circuits (a.k.a. $\AC^0$ circuits). In recent years, people have successfully applied this framework to give PRGs for various computational models \cite{DBLP:conf/focs/GopalanMRTV12, DBLP:conf/coco/TrevisanX13, DBLP:journals/siamcomp/HaramatyLV18, DBLP:journals/toc/LeeV20, DBLP:conf/focs/ForbesK18, DBLP:conf/stoc/MekaRT19-width3}. However, as we have mentioned, in all these PRG results, the target circuit/program always reads its input in a pre-defined pattern. 

Both \cite{ChenJWW22} and our new result suggest that the iterative restriction is more versatile than we thought: they can fool some highly adaptive tests, in which the future query to the hash function heavily depends on previous responses. Also, the analysis is drastically different from the common paradigm in the PRG analysis (i.e., the ``simplify-under-restrictions'' lemma and hybrid argument combo). It would be interesting to see if there is a deeper connection between the new results and previous PRG results. In particular, it is known that the iterative restriction framework can fool fixed-order read-once branching program \cite{DBLP:conf/focs/ForbesK18} (known as ROBP in literature). Note that the cycle-finding procedure can by captured by a special class of branching programs, where the program reads variables in an adaptive fashion, with the promise that every ``acceptance path'' is read-once. Can we formulate a computation model that (1) captures ROBP and the cycle-finding procedure as special cases, and (2) can be fooled by the iterative restriction construction?

%As a more open-ended question, is there a general characterization to the power and limitation of the iterative restriction approach?

% \paragraph*{(Pseudo-)random $1$-out graphs.} There is a rich literature studying statistical properties of random graphs (see, e.g., \cite{FK15-random-graph}). However, little is known about pseudorandomly generated graphs\footnote{We note that one exception might be the work by Alon and Nussboim \cite{DBLP:conf/focs/AlonN08-k-wise} that studies $k$-wise independent Erd\"os-Reny\'i graphs.}. Our work (and \cite{ChenJWW22}) shows how we can ``locally'' analyze some global properties of the $1$-out pseudorandom graph. Note that one can alternatively view the hash function as a $1$-left-regular $m$-by-$n$ bipartite graph. Does our technique help prove other ``global'' pseudorandom properties for those graphs induced by the hash family?

% As a direct corollary of our analysis, we showed the $c$-connecting property of the pseudorandom family (i.e., Theorem~\ref{theo:connecting-intro}) in this paper. Could Theorem~\ref{theo:connecting-intro} be useful in designing low-space algorithms?

\medskip\noindent\textbf{Reduce the seed length further.} Can we reduce the seed length for the hash family further? One natural choice for such improvement would be reducing the $O(\log n)$-level pseudorandom hash family to constant levels. Even more ambitiously, what if we use two hash functions $\mathbf{h}_1:[m]\to [n]\cup \{0\}$ and $\mathbf{h}_2:[m]\to [n]$ and construct $\mathbf{h} = \mathbf{h}_1 + \mathbbm{1}[\mathbf{h}_1 = 0]\cdot \mathbf{h}_2$. Intuitively, we let the walk alternate between $\mathbf{h}_1$ and $\mathbf{h}_2$, with the hope that each $\mathbf{h}_i$ can ``mix'' the bias introduced by the other hash, so that the overall walk sequence appears to be (pseudo)random.

\subsection{Paper Organization}

The rest of the paper is organized as follows. We introduce necessary background knowledge in Section~\ref{sec:preliminaries}. In Section~\ref{sec:cjww}, we develop core tools and lemmas for our analysis. We also prove Theorem~\ref{theo:improved-seed} in the same section. In Section~\ref{sec:tradeoff}, we prove the tradeoff results (i.e., Theorems~\ref{theo:element-distinctness-algo} and \ref{theo:set-intersection-algo}). Finally, we prove Theorem~\ref{theo:connecting-intro} in Section~\ref{sec:connecting}.

%See Figure~\ref{fig:walk-example} for an example.

\section{Preliminaries}\label{sec:preliminaries}

We assume word RAM model in this paper. The space complexity of an algorithm is defined as the size of its working memory. Besides the working memory, the algorithm also has read-only random access to the input and one-way access to an infinitely long tape of random bits, which do not count towards the space complexity.

$[n]$ denotes $\{1,\dots, n\}$. For a sequence $a = (a_1,\dots, a_n)\in [m]^n$, define its second frequency moment as $F_2(a) = \sum_{1\le i,j\le n} \mathbbm{1}[a_i = a_j]$. Also define $F_{\infty}(a) = \max_{y\in [m]} \{ | x\in [n] : a_x = y | \}$ as the number of occurrences of the most frequent element in $a$. Note that $F_{\infty}(a)\le \sqrt{F_2(a)}$. For \textsc{Element Distinctness} and \textsc{Set Intersection}, we assume the input array contains integers from $[m]$, where $m\le \poly(n)$.

We always use boldface letters (e.g., $\textbf{X}$) to denote random variables. For a random variable $\mathbf{X}$, we use $\supp(\mathbf{X})$ to denote its support. We use $(p_i)_{i\in [L]}$ to denote a list $(p_1,\dots, p_L)$. When the size of the list $L$ is clear from the context, we may omit the outer subscript and simply write $(p_i)$ or $(p_i)_i$. For a statement $E$, we use $\mathbbm{1}\{E\}$ denote the indicator function of $E$, where $\mathbbm{1}\{E\}$ equals $1$ if and only if $E$ is true, and equals $0$ otherwise.

Recall the definition of bounded independence.

\begin{definition}
Let $n, m$ be two integers. Let $\calH$ by a distribution over hash functions mapping $\bits^n$ into $\bits^m$. We say that $\calH$ is $k$-wise independent, if for any $k$ input-output pairs $(x_1,y_1),\dots, (x_k, y_k) \in \bits^{n}\times \bits^m$ where $x_1,\dots, x_t$ are distinct, %\avishay{You should mention that $x_1, \ldots, x_k$ are distinct}
it holds that
\[
\Pr_{h\sim \calH} [ \forall i\in [k], h(x_i) = y_i ]  = 2^{-km}.
\]
\end{definition}

We have the following standard construction of bounded-independence hash functions (check e.g., \cite[Chapter~3.5.5]{DBLP:journals/fttcs/Vadhan12-pseudorandomness}).
%\avishay{If you're citing a book, better to include the chapter or even subchapter, e.g. \cite[Chapter~ XX.YY]{DBLP:journals/fttcs/Vadhan12-pseudorandomness}}

\begin{lemma}\label{lemma:k-wise}
For every $n, m, k\ge 1$, there is an explicit $k$-wise independent hash functions $\calH$ that maps $\bits^n$ into $\bits^m$. One can sample a function in $\calH$ using $O(k(n+m))$ random bits. Given a seed $s$, let $h_s\in \mathcal{H}$ be the function described by $s$. One can evaluate $h_s(x)$ in $\poly(n,m,k)$ time.
\end{lemma}

% \paragraph*{The Problem.} The input are two integers $n,m\ge 1$, and a read-only sequence $a_1,\dots, a_n$, where each $a_i$ is an integer from $[m]$. The \textsc{Element Distinctness} asks one to decide whether the sequence $(a_i)_{i=1}^n$ contains distinct elements. That is, we wish to decide if $a_i\ne a_j$ holds for every $1\le i < j \le n$.

\section{The Pseudorandom Hash Family}\label{sec:cjww}

In this section, we give a significantly simpler analysis of the Chen-Jin-Williams-Wu result. Our analysis yields an improved seed length $O(\log^3 n)$. This section also lays the foundation for the tradeoff results (Section~\ref{sec:tradeoff}) as well as the $c$-connecting property (Section~\ref{sec:connecting}). Tools and lemmas developed in this section can be used to deduce those extensions easily.

This section is organized as follows. In Section~\ref{sec:setup}, we show the construction of the pseudorandom hash family, and state the pseudorandom property we need from it (i.e., Lemmas~\ref{lemma:single-ub} and \ref{lemma:double-lb}). Assuming them, we prove Theorem~\ref{theo:improved-seed} quickly. Towards proving Lemmas~\ref{lemma:single-ub} and \ref{lemma:double-lb}, we develop some technical tools in Section~\ref{sec:cjww-standard} and \ref{sec:cjww-extend}. We prove two lemmas in Section~\ref{sec:cjww-wrapup}.

\subsection{Setup and Proof of Theorem~\ref{theo:improved-seed}}\label{sec:setup}

\paragraph*{The hash construction.} We formally state the construction of the pseudorandom hash family $\calH^{n, m, t, \kappa}$, which is parameterized by four integers $n, m, t, \kappa\in \mathbb{N}^+$.

\begin{itemize}
    \item Sample $\mathbf{h}_1,\dots, \mathbf{h}_t$. For every $i \in [t]$, $\mathbf{h}_i: [m]\to [n]\cup \{0\}$ is a hash function satisfying the following.
    \begin{itemize}
        \item For every $j\in [m], v\in [n]$, $\Pr_{\mathbf{h}_i}[\mathbf{h}_i(j) = 0] = \frac{1}{2}$ and $\Pr_{\mathbf{h}_i}[\mathbf{h}_i(j) = v] = \frac{1}{2n}$.
        \item $\mathbf{h}_i$ is $\kappa$-wise independent.
    \end{itemize}
    % We also sample $\mathbf{h}_t:[m]\to [n]$ to be a $k$-wise independent hash function.
    \item Define the final hash $\mathbf{h}:[m]\to [n] \cup \{-1\}$ as follows. For every $j\in [m]$, we find the smallest $q \le t$ such that $\mathbf{h}_q(j) \ne 0$ and define $\mathbf{h}(j) := \mathbf{h}_q(j)$. If no such $q$ exists, we define $\mathbf{h}(j) := -1$.
\end{itemize}

For technical reasons, we need the codomain of $h\in \supp(\calH^{n,m,t,\kappa})$ to be $[n]\cup \{-1\}$. By Lemma~\ref{lemma:k-wise}, it requires $O(t \kappa \log(n+m))$ bits to sample $\mathbf{h}\sim \calH^{n,m,t,\kappa}$.

\paragraph*{The digraph.} Let $a_1,\dots, a_n$ be an integer array where for every $x\in [n]$, $a_x\in [m]$. For every function $h:[m]\to [n]$, define from $a$ and $h$ a digraph $G_{a,h}$. The vertex set for $G_{a, h}$ is $[n]$. For each $x\in [n]$, if $\mathbf{h}(a_x) \ne -1$, we add a directed edge $(x, \mathbf{h}(a_x))$. Let $\Out_{a,h}(x)$ denote the set of vertices reachable from $x$ in $G_{a,h}$. For a set $A \subseteq [n]$ of vertices, define $\Out_{a,h}(A)$ as the set of vertices reachable from at least one vertex in $A$ on $G_{a,h}$. Define from $G_{a,h}$ a mapping $f_{a,h}:[n]\to [n]\cup \{0\}$ such that $f_{a,h}(x) = y$ if $(x,y)\in G_{a,h}$, and $f_{a,h}(x) = 0$ if no such $y$ exists. Equivalently, $f_{a,h}(x) = h(a_x)$.

\paragraph*{The BCM Algorithm.} Next, recall the cycle-finding algorithm by Beame \emph{et al.}~%The following formulation is by \cite{ChenJWW22}.

\begin{lemma}[\cite{DBLP:conf/focs/BeameCM13}, Theorem 2.1]\label{lemma:bcm-cycle}
Assuming oracle access to $f_{a,h}:[n]\to [n]\cup \{\star\}$, there is a deterministic algorithm $\COLLIDE(A)$ that takes a set $A \subseteq [n]$ of vertices and finds all the pairs $(y, \{ u : u\in \Out_{a,h}(A), a_u = y\} )$. The algorithm uses $O(|A|\log n)$ space and $\Otilde(|\Out_{a,h}(A)|)$ time.
\end{lemma}

To illustrate, suppose that $\COLLIDE$ reports a pair $(y, \{u, v\})$. This indicates that $a_u = a_v = y$ and a collision is found. Algorithms developed in this section (Section~\ref{sec:cjww}) always call $\COLLIDE$ with a singleton $\{x\}$. In Section~\ref{sec:tradeoff}, we will design algorithms that call $\COLLIDE$ with multiple vertices.

\paragraph*{Pseudorandomness of the hash.} The main technical results in this section are the following.

\begin{lemma}\label{lemma:single-ub}
    For every $t \le \frac{1}{2}\log n$, let $\kappa = 20\log n$. Sample $\mathbf{h}\sim \calH^{n,m,t,\kappa}$ and $\mathbf{x}\sim [n]$. Then for every $u\in [n]$, it holds that:
    \[
    \Pr_{\mathbf{h}, \mathbf{x}}[u\in \Out_{a,\mathbf{h}}(\mathbf{x})] \le O\left(\frac{2^t}{n}\right).
    \]
\end{lemma}

\begin{lemma}\label{lemma:double-lb}
    Let $t = \frac{1}{2}\log n$ and $\kappa = 20\log n$. Sample $\mathbf{h}\sim \calH^{n,m,t,\kappa}$ and $\mathbf{x}\sim [n]$. Then for every $u, v\in [n]$, $u\ne v$, it holds that:
    \[
    \Pr_{\mathbf{h}, \mathbf{x}}[u, v\in \Out_{a,\mathbf{h}}(\mathbf{x})] \ge \Omega\left( \frac{1}{F_2(a)} \right).
    \]
\end{lemma}

Assuming Lemma~\ref{lemma:single-ub} and \ref{lemma:double-lb}, we can prove Theorem~\ref{theo:improved-seed}.

\begin{reminder}{Theorem~\ref{theo:improved-seed}}
Both \textsc{Element Distinctness} and \textsc{Set Intersection} can be solved by a Monte Carlo algorithm that runs in $\Otilde(n^{3/2})$ time, uses $O(\log^3 n)$ bits of working space and no random oracle.
\end{reminder}

\begin{proof} We start with the algorithm for \textsc{Element Distinctness}. 

\paragraph*{Solving \textsc{Element Distinctness}.} Let $a\in [m]^n$ be the input. We set $t = \frac{1}{2}\log n$ and $\kappa = 20\log n$. Our algorithm repeats the following process for $\Theta(n\log n)$ times:

\begin{itemize}
    \item Draw a random hash $\mathbf{h}\sim \calH^{n,m,t,\kappa}$ and a starting vertex $\mathbf{x}\sim [n]$. Try to find a colliding pair by running $\COLLIDE(\mathbf{x})$ on $G_{a,\mathbf{h}}$. 
\end{itemize}
The algorithm reports YES if it does not find any colliding pair. Otherwise it reports NO. 

By Lemma~\ref{lemma:single-ub} and \ref{lemma:bcm-cycle}, the expected running time of one trial is $\Otilde\left( n \cdot \frac{2^t}{n} \right) \le \Otilde(\sqrt{n})$. Next, we argue that if $a$ is a NO instance to \textsc{Element Distinctness}. Then with probability $\Omega\left( \frac{1}{n} \right)$, one trial succeeds in finding a pair $(p, q)$ such that $p\ne q$ but $a_p = a_q$.

Indeed, for every $1\le p < q \le n$ with $a_p = a_q$, let $\calE(p,q)$ denote the event the $\COLLIDE(\mathbf{x})$ outputs $(p,q)$ on $G_{a, \mathbf{h}}$. Then by Lemma~\ref{lemma:double-lb}, we have $\Pr_{\mathbf{x},\mathbf{h}}[\calE(p,q)] \ge \Omega\left( \frac{1}{F_2(a)} \right)$. Also note that events $\{\calE(p,q):p<q,a_p=a_q\}$ are mutually disjoint (because $\Out_{a,\mathbf{h}}(\mathbf{x})$ cannot contain two colliding pairs). Therefore, we have
\[
\Pr_{\mathbf{x}, \mathbf{h}}[\text{find a colliding pair in $\Out_{a,\mathbf{h}}(\mathbf{x})$}] \ge \sum_{1\le p < q \le n, a_p = a_q} \Omega\left( \frac{1}{F_2(a)} \right) \ge \Omega\left( \frac{F_2(a) - n}{2\cdot F_2(a)} \right) \ge \Omega\left(\frac{1}{n}\right)
\]
provided that $F_2(a) \ge n+1$ (i.e., the input $a$ contains at least one colliding pair).

If $a$ is a YES instance, the algorithm always reports YES. If $a$ is a NO instance, the algorithm finds a colliding pair in at least one trial with probability $1 - \left( 1 - \Omega\left( \frac{1}{n} \right) \right)^{n\log n}\ge 1 - n^{\Omega(1)}$. The expected running time is $\Otilde(n^{3/2})$. The space usage is bounded by $O(\log n)$ plus the space required to store (the description of) $h$, which is $O(t\kappa \log(n+m)) \le O(\log^3 n)$. Overall the space complexity is $O(\log^3 n)$.

\paragraph*{Solving \textsc{Set Intersection}.} Now we present the algorithm for \textsc{Set Intersection}. Suppose $a,b\in [m]^n$ are the input arrays. Define $c\in [m]^{2n}$ as the concatenation of $a$ and $b$. The algorithm repeats the following process for $O(n\log^2 n)$ times:
\begin{itemize}
    \item Sample $\mathbf{h}\sim \calH^{n,m,t,\kappa}$ and $\mathbf{x}\sim [n]$. Run $\COLLIDE(\mathbf{x})$ on $G_{c,\mathbf{h}}$. If $\COLLIDE(\mathbf{x})$ returns a colliding pair $(p,q)$, print $c_p$.
\end{itemize}
The running time is $\Otilde(n^{3/2})$. We argue the correctness now. Suppose $c_p = c_q$ is a colliding pair. With probability $\Omega\left( \frac{1}{F_2(c)} \right) = \Omega\left( \frac{1}{n} \right)$, the algorithm finds $(p,q)$ in one trial. Since we have $O(n\log^2 n)$ independent trials, the probability that the algorithm misses $(p,q)$ is bounded by $\left( 1 - \Omega\left( \frac{1}{n} \right)\right)^{n\log^2 n} \le n^{-\omega(1)}$. Union-bounding over all colliding pairs concludes the proof.
\end{proof}

Next, we prove Lemma~\ref{lemma:single-ub} and \ref{lemma:double-lb}.

% We consider the \textsc{Element Distinctness} and \text{Set Intersection} problem. Suppose the input instance for \text{Element Distinctness} is $a_1,\dots, a_n$. 

\subsection{The Recursive Perspective of the Walk}\label{sec:cjww-standard}

To better exploit the hierarchical structure of $h$, we consider a recursion-based perspective of the walk on the graph $G_{a, h}$, as shown in Algorithm~\ref{algo:standard-walk}. Algorithm~\ref{algo:standard-walk} starts with a given input $x\in [n]$ and produces a (infinite size) tensor $T\colon \mathbb{N}^t \to [n]\cup \{\star\}$. In the following, we use $\walk^{\std}$ to refer to Algorithm~\ref{algo:standard-walk} and use $T=\walk^\std(h, x)$ to denote the resulting tensor when running $\walk^\std$ on $G_{a,h}$ with starting vertex $x\in [n]$.

\begin{algorithm2e}[h]
\caption{The Standard Walk}
\label{algo:standard-walk}
\SetKwInput{KwInput}{Input}                % Set the Input
\SetKwInput{KwVariable}{Global Variables}              % set the Output
\LinesNumbered

\DontPrintSemicolon
  
    \KwInput{$n,m\ge 1$. $t$ hash functions $h_1,\dots, h_t:[m]\to [n]\cup \{0\}$. The array $(a_1,\dots, a_n)\in [m]^n$. The starting vertex $x\in [n]$.}
    
    % Set Function Names
    \SetKwFunction{FMain}{Main}
    \SetKwFunction{FWalk}{stdwalk}
    \SetKw{Break}{break}
    
    \KwVariable{
        \; ~~~~ A tensor $T:\mathbb{N}^t\to [n]\cup \{\star\}$, initialized with $\star$'s.
        \; ~~~~ A set $D\subseteq [m]$, initialized with $\emptyset$.
    }
    
    \SetKwProg{Fn}{Program}{:}{\KwRet}
    \Fn{\FMain}{
        $\ell \gets (0, 0, \dots, 0) \in \mathbb{N}^{t}$ \tcp*{The index $\ell$ tracks the progress of the walk}
        $T(\ell) \gets x$ \tcp*{$T(\vec{0})$ records the starting vertex}
        $\stdwalk(t, x, \ell)$ \tcp*{Start the walk}
        \KwRet $T$\;
    }
    
    % Write Function with word ``Function''
    \SetKwProg{Fn}{Function}{:}{}
    \Fn{\FWalk{$i$, $x$, $\ell$}}{
        \If{$i=0$}{     
            \KwRet $x$ \tcp*{There is no level-$0$ hash}
        }
        \While{\emph{True}}{
            $x\gets \stdwalk(i-1, x, \ell)$  \tcp*{Move according to $h_{<i}$ first}
            \If{$a_x\in D$}{                
                \KwRet $x$                  \tcp*{If $h(a_x)$ has been queried, halt the walk}
            }
            $\ell_i \gets \ell_i + 1$    \tcp*{Prepare to make a level-$i$ move}
            \If{$h_i(a_x) = 0$}{
                \Break ;            \tcp*{If $h_i(a_x)=0$, need to look at $h_{>i}(a_x)$}
            }
            \Else{              
                $D\gets D \cup \{a_x\}$ \tcp*{``Remember'' that $h(a_x)$ has been used}
                $T(\ell) \gets h_i(a_x)$  \tcp*{Record the move in the tensor $T$}
                $x\gets h_i(a_x)$     \tcp*{Make a level-$i$ move}

            }
        }
        \KwRet $x$\;
    }
\end{algorithm2e}

%\begin{remark}
%Note that Algorithm~\ref{algo:standard-walk} will eventually visit the same vertex $x\in [m]$ twice and enter a loop. Hence, Algorithm~\ref{algo:standard-walk} cannot really reach Line $6$ in its execution. We include Line $6$ in the pseudo-code only for better exhibition. Nevertheless, the vector $\ell$ in $\walk(*, *, \ell)$ is always getting larger, and the algorithm never updates the same entry $T(\ell)$ twice. Therefore, the tensor $T$ is always well-defined.
%\end{remark}

\subsubsection{Understanding the standard walk}

Before we continue, we set up necessary pieces of notation and state some basic facts about the structure of $\walk^\std$. We start with the definition of ``index'' to the tensor $T$.

\begin{definition}\label{def:index}
We use the term \emph{index} to refer to $t$-dimensional integer vectors $\ell = (\ell_1,\dots, \ell_t)$. We introduce the following total ordering for indices: for two indices $\ell^1 \ne \ell^2$, let $i\in [t]$ be the \emph{largest} integer such that $\ell^1_i \ne \ell^2_i$. We say $\ell^1 < \ell^2$ if $\ell^1_i < \ell^2_i$. The width of an index $\ell$ is defined as $\width(\ell) := \max_{1\le i\le t} \{ \ell_i \}$. Call an index $\ell$ $\tau$-bounded if $\width(\ell) \le \tau$. The level of an index is $\level(\ell) := \max\{q : \forall i< q, \ell_i = 0\}$. In particular, we define $\level(\vec{0}) := t + 1$. For every $0\le i\le t$, we use $\ell_{<i}, \ell_{>i}$ to denote the length-$(i-1)$ prefix and length-$(t-i)$ suffix of $\ell$, respectively.
\end{definition}

\emph{Throughout the whole paper}, we reserve the letters ``$\ell$'', ``$r$'' and ``$\omega$'' for indices. 

When Line $17$-$19$ is executed in a function call $\stdwalk(i, x, \ell)$, we say the algorithm makes a level-$i$ move. Note that we use the index $\ell$ in the function call to track the ``progress'' of the walk. The index should be interpreted as follows. Suppose we call $\stdwalk(i, x, \ell)$ with $\ell = (\ell_1, \ell_2,\dots, \ell_t)$. Then it is always guaranteed that $\level(\ell) > i$. Moreover, $\ell_{t}$ counts the number of level-$t$ moves before reaching $\stdwalk(i, x, \ell)$. For every $j\in [i+1, t-1]$, $\ell_{j}$ further counts the number of level-$j$ moves after the very last level-$(j+1)$ move. See Figure~\ref{fig:walk-example} for an example.

\begin{figure}[h]
\centering
\includegraphics[width=14cm]{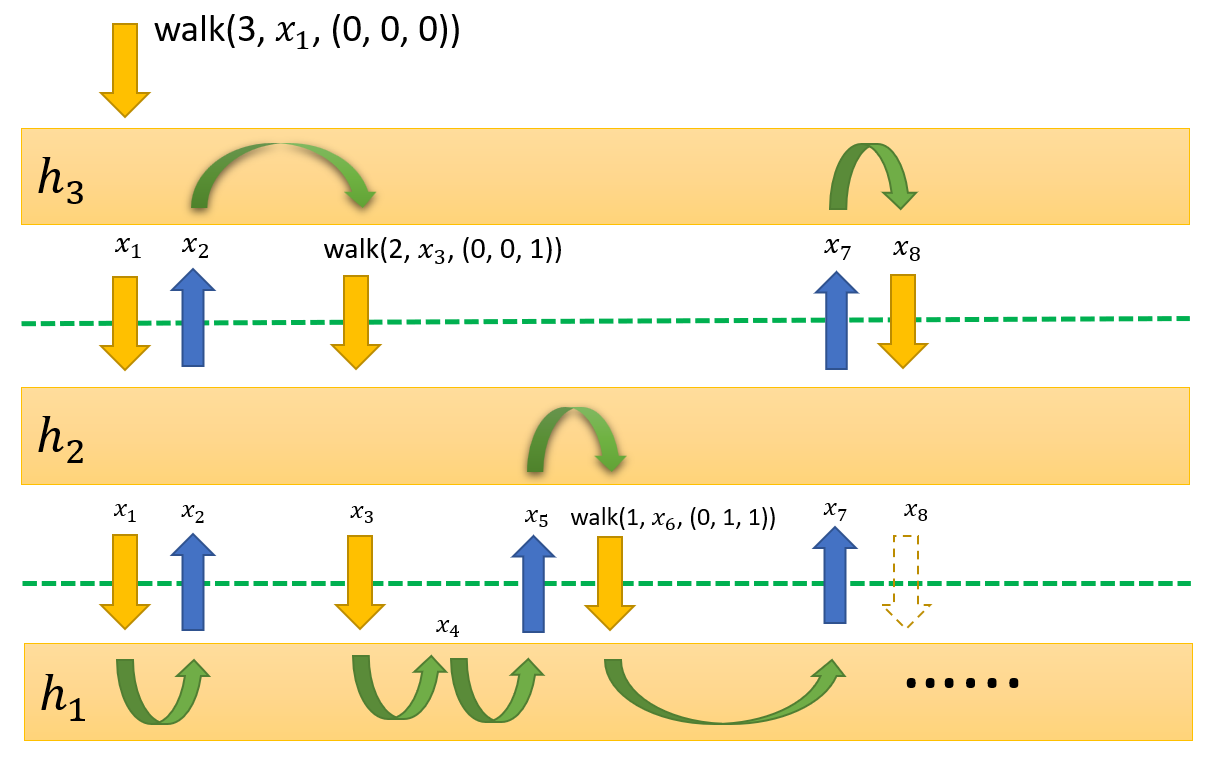}
\caption{One running example of the standard walk with $t=3$ levels. Orange and blue arrows represent function calls and returns, respectively. Note that we highlighted $3$ function calls $\stdwalk(3, x_1, (0,0,0)), \stdwalk(2, x_3, (0,0,1)), \stdwalk(1, x_6, (0,1,1))$ for better illustration.}
\label{fig:walk-example}
\end{figure}

We further observe that entries of the tensor $T$ are connected as a rooted and ordered tree.

\begin{definition}\label{def:walk-tree}
For every $t\ge 1$, define the level-$t$ walk tree $\calT^{(t)}$ as follows. $\calT^{(t)}$ is a rooted and ordered tree of infinite size. Vertices of $\calT^{(t)}$ are indexed by tuples $\{\ell = (\ell_1,\dots, \ell_t)\in \mathbb{N}^t\}$. The root of $\calT^t$ is $\vec{0} = (0, \dots, 0)$. For every pair $(\ell, r)$, $\ell$ is the parent of $r$, if there exists $i\in [t]$ such that $\ell_{>i} = r_{>i}$, $\ell_{<i} = r_{<i} \equiv 0$ (i.e. the length-$(i-1)$ prefixes of $\ell$ and $r$ are all zeros) and $r_{i} = \ell_i + 1$. Children of a vertex are sorted in the order of their indices (as per Definition~\ref{def:index}). For every $\ell\in \mathbb{N}^t$, let $P^t(\ell)$ denote the set of vertices from root to $\ell$ (inclusive). For a set of indices $S = \{ \ell^1, \dots, \ell^q \}$, define $P^{(t)}(S) = \bigcup_{1\le i\le q} P^{(t)}(\ell^i)$.
\end{definition}

When it is clear from context, we will omit the superscript ``$t$'' in $\calT^{(t)}$ and $P^{(t)}$. The following observations show the connections between the tensor $T$ and the walk tree $\calT$.

\begin{figure}[h]
\centering
\includegraphics[width=14cm]{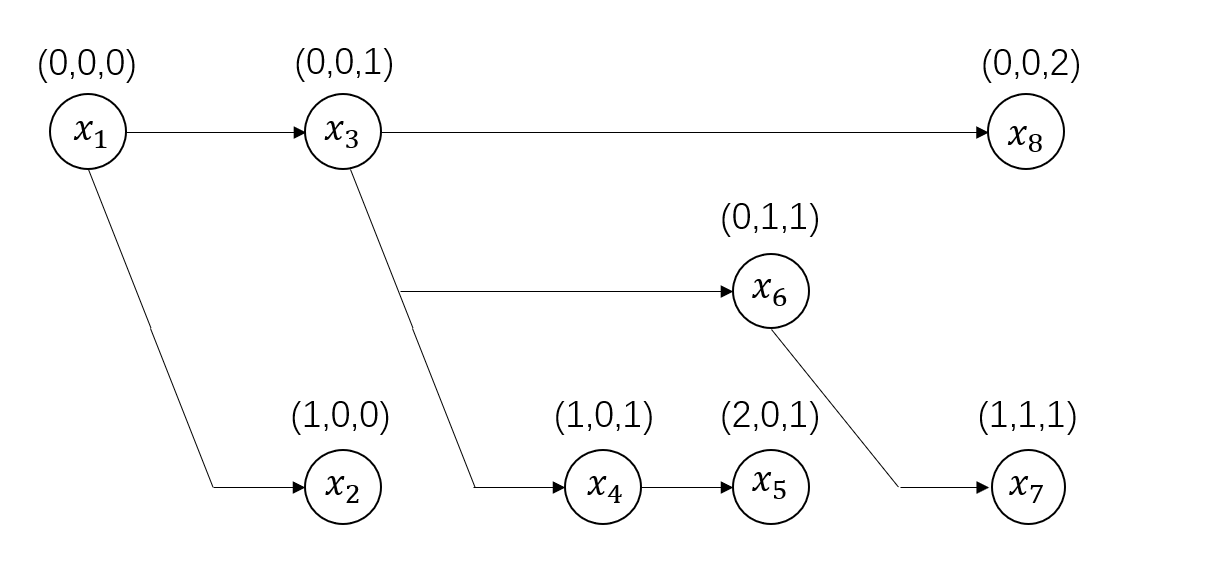}
\caption{The level-$3$ walk tree $\calT^{(3)}$ corresponding to the running example in Figure \ref{fig:walk-example}.}
\label{fig:tree-example}
\end{figure}

\begin{enumerate}
    \item For every $r\ne \vec{0}$, suppose Algorithm~\ref{algo:standard-walk} has updated $T(r)$ in Line 18. Then, for every $\ell$ being an ancestor of $r$ in $\calT$, it holds that $T(\ell) \ne \star$. Contrapositively, if $T(\ell) = \star$, then for every $r$ being a descendant of $\ell$, it holds $T(r) = \star$.
    \item For every $\ell \ne \vec{0}$, define
    \[
    \probe(\ell) = 
    \begin{cases}
    a_x & \text{ If $T(\ell)$ was set by accessing $h_{\level(\ell)}(a_x)$ in Line 18 of Algorithm~\ref{algo:standard-walk}, } \\
    \star & \text{Otherwise (i.e., Algorithm~\ref{algo:standard-walk} has never tried to write $T(\ell)$).}
    \end{cases}
    \]
    For every $\ell$ with $\probe(\ell) \ne \star$, we have $\probe(\ell) = a_{T(\pre(\ell))}$, where $\pre(\ell) := \sup\{ r : r < \ell, T(r) \ne \star\}$. Intuitively, $\pre(\ell)$ is the last entry of $T$ updated before updating $T(\ell)$. We only consider $\pre(\ell)$ for those $\ell$'s with $T(\ell) \ne \star$. We observe that $\pre(\ell)$ is either the parent of $\ell$ or a descendant of $\ell$'s parent. Also note that the definition of $\probe, \pre$ is associated with a tensor $T$.
\end{enumerate}

We introduce an equivalence relation ``$\simeq$'' for $[n]$, where we say $u \simeq v$ if $a_u = a_v$. Note that if we find $T(\ell) \simeq T(r)$ for some $\ell < r$, it means that the walk enters a loop after updating $T(r)$. Finally, note that we use a set $D$ to remember the entries of $h$ that we have queried. When $\walk^\std$ reaches a point $r$ such that $T(\ell)\simeq T(r)$ holds for some $\ell < r$, it will terminate the walk.

\subsubsection{Analyzing the standard walk} 

Back to our discussion, we first observe that the tensor $T$ exactly represents the set $\Out_{a,h}(x)$ in the sense that
\[
\Out_{a,h}(x) = \{ u\in [n]: \exists \ell\in \mathbb{N}^t, T(\ell) = u \}.
\]
Now, consider sampling a list of random $\mathbf{h}_1,\dots, \mathbf{h}_t$ together with a starting vertex $\mathbf{x}$. Let $\mathbf{h}, \mathbf{T}$ be two random variables denoting the resulting hash function and tensor, respectively. Suppose $u,v\in [n]$ is a pair such that $a_u= a_v, u\ne v$. We wish to find the probability that $u,v\in \Out_{a,\mathbf{h}}(\mathbf{x})$. For every pair of indices $\ell^1\ne \ell^2$, denote $\calE(\ell^1,\ell^2)$ as the event ``$\mathbf{T}(\ell^1)=u\land \mathbf{T}(\ell^2) = v$''. Because the standard walk halts immediately after finding a pair $\ell < r$ such that $\mathbf{T}(\ell) \simeq \mathbf{T}(r)$. We observe that $\{\calE(\ell^1,\ell^2)\}$ are mutually disjoint. Therefore, we have
\begin{align}
&  \Pr_{\mathbf{h_1},\dots, \mathbf{h}_t, \mathbf{x}}[u,v\in \Out_{a,\mathbf{h}}(\mathbf{x})] 
  = \sum_{\ell^1 \ne \ell^2}\Pr_{\mathbf{h}_1,\dots,\mathbf{h}_t, \mathbf{x}}\left[  \mathbf{T}(\ell^1) = u \land \mathbf{T}(\ell^2) = v \right]. \label{eq:translate-to-tree}
\end{align}

% Set a threshold $\tau = 5\log n$. Next, our goal is to lower bound for each \emph{$\tau$-bounded} (see Definition~\ref{def:index}) pair $(\ell_1, \ell_2)$ the probability in \eqref{eq:translate-to-tree} by $c\cdot \frac{2^{-|\ell_1|-|\ell_2|}}{n^2}$ for some constant $c > 0$. Assuming this is true for now, we can easily deduce that
% \begin{align}
% & \sum_{\ell^1 < \ell^2: \text{both $\tau$-bounded} }\Pr_{\mathbf{h}_1,\dots,\mathbf{h}_t, \mathbf{x}}\left[  \{ \mathbf{T}(\ell^1), \mathbf{T}(\ell^2)\} = \{u, v\} \land \left( \forall \ell^3 < \ell^4 < \ell^2, \mathbf{T}(\ell^3) \not\simeq \mathbf{T}(\ell^4) \right)\right] \notag \\
% & \ge \frac{1}{2} \left( \sum_{\ell^1, \ell^2: \text{both  $\tau$-bounded}} \frac{c\cdot 2^{-|\ell_1|-|\ell_2|}}{n^2} - \sum_{\ell^1 = \ell^2: \text{$\tau$-bounded}} \frac{c\cdot 2^{-2|\ell^1|}}{n^2}  \right) \notag \\ 
% & \ge \frac{c}{2\cdot n^2} \left( \sum_{i=0}^{\tau-1} 2^{-i} \right)^{2t} - \frac{c}{n^2} \left( \sum_{i=0}^{\tau-1} 2^{-2i} \right)^t \notag \\
% & \ge \Omega\left(\frac{2^{2t}}{n^2}\right) - O\left(\frac{2^t}{n^2}\right) \ge \Omega\left( \frac{1}{n} \right). \label{eq:main-lowerbound}
% \end{align}
In the following, we introduce new tools (``extended walks'') to prove the lower bound for every $\tau$-bounded pair $(\ell^1,\ell^2)$.

\subsection{The Extended Walk}\label{sec:cjww-extend}
 
Analyzing the distribution $(\mathbf{T}(\ell^1), \mathbf{T}(\ell^2))$ through $\walk^\std$ directly seems very hard. To circumvent the issue, we will couple the standard walk with a family of carefully designed extended walks.

Fix $\tau = 5\log n$. For any finite set $S\subseteq \mathbb{N}^t$ of $\tau$-bounded\footnote{See Definition~\ref{def:index}.} indices, we construct a randomized $S$-extended walk as shown in Algorithm~\ref{algo:extended-walk}. In the following, we also use $\walk^S$ to denote the $S$-extended walk and use $\mathbf{T}^S \sim \walk^S(h,x)$ to denote the resulting tensor\footnote{Note that even if $h, x$ are fixed, the extended walk still has its internal randomness.} when running $\walk^S$ on $h$ with starting vertex $x$. 

Compared with $\walk^\std$, we implement two major modifications in $\walk^S$:
\begin{itemize}
    \item When $\walk^S$ encounters a loop, it does not halt the walk immediately, neither does it run forever. There is a threshold $\tau = 5 \log n$. We implement $\walk^S$ so that inside each function call $\extwalk(i, x, \ell)$, $\walk^S$ makes at most $\tau$ steps of level-$i$ move (Line $9$). This ensures that the algorithm always has a chance of visiting every index in $S$.
    \item Second, we use $t$ sets $D_1, \dots, D_t \subseteq [n]$ to ``remember'' the coordinates of $h_1,\dots, h_t$ that we have probed along $P(S)$ (Line $22$-$23$). When moving along $P(S)$ and trying to access $h_i(a_x)$, we always check if $h_i(a_x)$ has been used for any previous index in $P(S)$ (Line $12$). If so, we sample a uniform one-time edge to avoid accessing $h_i(a_x)$ (Line $13$).
\end{itemize}

Next, we demonstrate the insight behind Algorithm~\ref{algo:extended-walk} by showing its two advantages. In Section~\ref{sec:analyze-extended-walk}, we show that analyzing the distribution of $(\mathbf{T}^S(\ell^1), \mathbf{T}^S(\ell^2))$ is significantly easier than analyzing $(\mathbf{T}(\ell^1), \mathbf{T}(\ell^2))$ directly. In Section~\ref{sec:extended-couple}, we show how lower bounds for the extended walk translate to that for the standard walk.

\begin{algorithm2e}[H]
\caption{The $S$-Extended Walk}
\label{algo:extended-walk}
\SetKwInput{KwInput}{Input}                % Set the Input
\SetKwInput{KwVariable}{Global Variables}              % set the Output
\LinesNumbered
\DontPrintSemicolon
  
    \KwInput{
        \; ~~~~ $n,m, t\ge 1$. $t$ hash functions $h_1,\dots, h_t$. The input array $a_1,\dots, a_n$.
        \; ~~~~ A starting vertex $x\in [n]$. 
        \; ~~~~ A set of indices $S \subseteq \mathbb{N}^t$.
    }
    
    % Set Function Names
    \SetKwFunction{FMain}{Main}
    \SetKwFunction{FWalk}{extwalk}
    \SetKw{Break}{break}
    
    \KwVariable{
        \; ~~~~ A constant $\tau \gets 5 \log n$
        \; ~~~~ A tensor $T:[0, \tau]^t\to [n]\cup \{\star\}$, initialized with $\star$'s
        \; ~~~~ $t$ sets $D_1,D_2,\dots, D_t\gets \emptyset$. \tcp*{Different sets for different levels}
    }
    
    \SetKwProg{Fn}{Program}{:}{\KwRet}
    \Fn{\FMain}{
        $\ell \gets (0, 0, \dots, 0) \in \mathbb{N}^{t}$ \;
        $T(\ell) \gets x$ \;
        $\extwalk(t, x, \ell)$ \;
        \KwRet $T$\;
    }
    
    % Write Function with word ``Function''
    \SetKwProg{Fn}{Function}{:}{}
    \Fn{\FWalk{$i$, $x$, $\ell$}}{
        \If{$i=0$}{
            \KwRet $x$ \;
        }
        \While(\tcp*[f]{Make at most $\tau$ steps of level-$i$ move}){$\ell_i < \tau$}{
            $x\gets \extwalk(i-1, x, \ell)$ \;
            $\ell_i \gets \ell_i + 1$ \;
            \If(\tcp*[f]{Avoid accessing $h_i(a_x)$ twice in $P(S)$}){$\ell \in P(S)$ and $a_x\in D_i$}{ 
                $(\alpha, \beta) \sim \{0,1\}\times [n]$ \tcp*{Sample $(\alpha,\beta)$ uniformly randomly}
            }
            \Else(\tcp*[f]{Access $h_i(a_x)$ for the first time in $P(S)$}){
                \If{$h_i(a_x) = 0$}{
                    $(\alpha, \beta)\gets (0, 0)$ 
                }
                \Else{
                    $(\alpha, \beta)\gets (1, h_i(a_x))$ 
                }
            }
            \If{$\alpha = 0$}{
                \Break \tcp*{If $\alpha = 0$, end the current level-$i$ walk}
            }
            \Else{
                \If{$\ell \in P(S)$}{
                    $D_i \gets D_i \cup \{a_x\}$ \tcp*{``Remember'' $a_x$ only when $\ell\in P(S)$}
                }
                $T(\ell) \gets \beta$ \;
                $x\gets \beta$ \tcp*{If $\alpha = 1$, make a level-$i$ move}
            }
        }
        \KwRet $x$\;
    }
\end{algorithm2e}

% We use $\supp(\walk^S(h,x))$ to denote the support of $\walk^S(h,a)$.

% $(\mathbf{T}(\ell^1), \mathbf{T}(\ell^2))$ translate to that for  $(\mathbf{T}^S(\ell^1), \mathbf{T}^S(\ell^2))$ is closely related to $(\mathbf{T}(\ell^1), \mathbf{T}(\ell^2))$ so that we can translate the lower bound for the former one into a bound for the latter one.

\subsubsection{Analyzing the extended walk}\label{sec:analyze-extended-walk}

The extended walk behaves nicely if we only observe $\mathbf{T}^S(\ell)$ for those $\ell\in S$. Formally, we have the following lemma.

\begin{lemma}\label{lemma:extend-is-random}
Let $S = \{ \ell^1, \dots, \ell^c \}$ be $c$ $\tau$-bounded indices. Suppose $\mathbf{h}\sim \calH^{n,m,t,c\tau}$ and $\mathbf{x}\sim [n]$. Let $\mathbf{T}^S\sim \walk^S(\mathbf{h}, \mathbf{x})$. For every $u_1,u_2,\dots, u_c \in [n]$, we have
\[
\Pr_{\mathbf{h}, \mathbf{x}, \mathbf{T}^S} \left[ \mathbf{T}^S(\ell^i) = u_i, \forall i\in [c] \right] = \frac{2^{-|P(S)|+1}}{n^{c}}.
\]
\end{lemma}

%To prove Lemma~\ref{lemma:extend-is-random}, we first prove the following lemma.

\begin{proof}
Let $\mathbf{h}_1,\dots, \mathbf{h}_t$ denote the building hashes for $\calH^{n,m,t,c\tau}$. Recall that each of $\mathbf{h}_i$ is $c\tau$-wise independent. For each $d\in [t+1]$, define $B^d\subseteq P(S)$ as the subset of $P(S)$ that contains all level-$d$ indices. Let $B^{\ge d} = \bigcup_{j=d}^{t+1} B^j$. Denote $b_d = |B^{\ge d}|$. Let $r^1,\dots, r^{b_d}$ enumerate all the indices in $B^{\ge d}$. 

\begin{claim}
For every $d\in [t+1]$, the following is true. For every $v_1,\dots, v_{b_d}\in [n]$, conditioning on $\mathbf{T}^S(r^i) = v_i, \forall i\in [b_d]$, $\mathbf{h}_1,\dots, \mathbf{h}_{d-1}$ are still uniformly distributed.
\end{claim}

We prove the claim by \emph{downwards} induction on $d$. For the case $d=t+1$, we have $B^{d} = \{\vec{0}\}$. Since $\mathbf{T}(\vec{0}) = \mathbf{x}$, it is clearly independent of $\mathbf{h}_1,\dots, \mathbf{h}_{t}$. Suppose the claim holds for $d+1\le t+1$, we prove it for the case of $d$. We condition on 
\[
\big[ \mathbf{T}^S(r) = T^S(r), \forall r \in B^{\ge d+1} \big],
\]
where ${T^S(r)}_{r\in B^{\ge d+1}}$ is a list of non-star elements. By the induction hypothesis, $\mathbf{h}_1,\dots, \mathbf{h}_{d}$ are uniformly distributed.

We partition $B^{d}$ into $s$ groups $B^d = \bigsqcup_{j=1}^s B^d_j$, where two indices lie in the same group if and only if they have the same $(t-d)$-suffix. Note that $s\le c$. For each $j\in [s]$, denote $B^d_j = \{ \ell^{d,j,1},\ell^{d,j,2},\dots, \ell^{d,j,|B^d_j|}\}$ where $\ell^{d,j,1} < \ell^{d,j,2} < \dots < \ell^{d,j,|B^d_j|}$. Note that we have $\ell^{d,j,1}_d = 1$ (otherwise $\ell^{d,j,1}$ would not be the first index in $B^d_j$). Let $f^{d,j}$ be the parent of $\ell^{d,j,1}$. Observe that $f^{d,j}\in B^{\ge d+1}$ and we have conditioned on that $\mathbf{T}^S(f^{d,j}) = T^S(f^{d,j}) \ne \star$. For brevity, we also denote $\ell^{d,j,0} := f^{d,j}$.

Suppose $B^d_1,\dots, B^d_s$ are sorted in the order of their indices\footnote{Since indices in a group have the same suffix, indices from different groups do not have interleaving orders.}. For every $j\in [s]$, given $\mathbf{T}^S(f^{d,j}) = T^S(f^{d,j})$, we know that $\walk^S(\mathbf{h},\mathbf{x})$ has a function call of the form $\extwalk(d, T^S(f^{d,j}), f^{d,j})$. 

Fix a list of $h_1,\dots, h_{d-1}$. We study how $(\mathbf{T}^S(\ell^{d,j,p}))_{j\in [s], p\in [|B^d_j|]}$ depends on $\mathbf{h}_d$ by simulating these function calls $\big(\extwalk(d, T^S(f^{d,j}), f^{d,j})\big)_{j\in [s]}$ in order\footnote{There might be other level-$d$ function calls that lie between these calls. They might even depend on $\mathbf{h}_1,\dots, \mathbf{h}_d$ arbitrarily. However, since these calls always return and do not change the set $D_d$, we do not need to observe them.}.

In more detail, the simulation consists of $s$ rounds. It either (1) finds $\mathbf{T}(\ell) = \star$ for some $\ell \in B^d$ and outputs FAIL, or (2) outputs a list $(\mathbf{T}^S(\ell^{d,j,p}))_{j\in [s], p\in [|B^d_j|]}$. The simulation also maintains the set $D_d$, which is empty in the beginning of the simulation. For $j=1,\dots, s$, the simulation runs the While-loop in $\extwalk(d, {T}^S(f^{d,j}),f^{d,j})$ for $|B^d_j|$ turns. For each $p\in [|B^d_j|]$, assume that $\extwalk(d, T^S(f^{d,j}), f^{d,j})$ does not return before the $p$-th turn. Then:
\begin{itemize}
    \item $\extwalk(d, T^S(f^{d,j}), f^{d,j})$ recursively calls $\extwalk(d-1, \mathbf{T}^S(\ell^{d,j,p-1}), \ell^{d,j,p-1})$. Here, $\mathbf{T}^S(\ell^{d,j,p-1})$ might be (1) an entry of $\mathbf{h}_d$, (2) $T^S(\ell^{d,j,0})$ (if $p = 1$), or (3) a random element coming from Line 13 of Algorithm~\ref{algo:extended-walk}.
    \item Having fixed $h_1,\dots, h_{d-1}$, the simulation can simulate $\extwalk(d-1, \mathbf{T}^S(\ell^{d,j,p-1}), \ell^{d,j,p-1})$, which returns a vertex $x^{j,p}$. Note that $x^{j,p}$ only depends on $\mathbf{T}^S(\ell^{d,j,p-1})$ and $h_1,\dots, h_{d-1}$. Now, we consider two cases.
    \begin{itemize}
        \item If $a_{x^{j,p}}\notin D_d$, then the simulation queries $\mathbf{h}_d(a_{x^{j,p}})$. With probability $\frac{1}{2}$, it finds $\mathbf{h}_d(a_{x^{j,p}}) = 0$. In this case, the simulation fails. Otherwise, a new entry of $\mathbf{h}_d$ is probed and it sets $\mathbf{T}^S(\ell^{d,j,p}) = \mathbf{h}_{d}(a_{x^{j,p}})$.
        \item If $a_{x^{j,p}} \in D_d$, again with probability $\frac{1}{2}$ the simulation fails due to $\mathbf{\alpha} = 0$ (Line $13$). Otherwise the simulation samples $\mathbf{T}^S(\ell^{d,j,p})$ to be a uniformly random element from $[n]$.
    \end{itemize}
    \item After this turn of simulation, the simulation updates $D_d \gets D_d \cup \{a_{x^{j,p}}\}$.
\end{itemize}

By the design of Algorithm~\ref{algo:extended-walk}, the simulation simulates (a part of) $\walk^S(\mathbf{h},\mathbf{x})$ faithfully. Note that the simulation consists of $s$ function calls and $|B^d|$ turns in total. For every fixed $h_1,\dots, h_{d-1}$, with probability $2^{-|B^d|}$, the simulation does not fail. We condition on this event. Further observe that the simulation queries at most $|B^d| \le c\tau$ entries from $\mathbf{h}_d$, and each entry of $\mathbf{h}_d$ is used for at most once. Since $\mathbf{h}_d$ is $c\tau$-wise independent, we conclude that the output of the simulation $(\mathbf{T}^S(\ell^{d,j,p}))_{j,p}$ contains independent and uniform elements from $[n]$, \emph{regardless of how $h_1,\dots, h_{d-1}$ behave}.

Now, since the output distribution of the simulation does not change with $h_1,\dots, h_{d-1}$, we can use Bayes' rule to conclude that conditioning on an output of the simulation, the distributions of $\mathbf{h}_1,\dots, \mathbf{h}_{d-1}$ are still uniform. This completes the induction.

To conclude, we use the induction for $t$ turns, which consists of $t$ simulations (one for each layer $d\in [t]$). With probability $\prod_{d=1}^{t} 2^{-|B^d|} = 2^{-|P(S)|+1}$, none of the simulations fails. Conditioning on this event, $(\mathbf{T}^S(\ell))_{\ell \in P(S)}$ contains $|P(S)|$ uniform and independent elements from $[n]$, which implies that $(\mathbf{T}^S(\ell^i))_{\ell^i \in S}$ contains uniform and independent elements. This completes the proof.
\end{proof}

\begin{remark}
If $\mathbf{h}_1,\dots, \mathbf{h}_t$ are $\delta$-almost $c\tau$-wise independent, the same argument still holds up to a small error $O(n^{c\tau} t \delta)$. Since this does not reduce the overall seed length to sample $\mathbf{h}_1,\dots, \mathbf{h}_t$, we do not give the formal proof here.
\end{remark}

\begin{remark}
Lemma~\ref{lemma:extend-is-random} is the \emph{only} place where the pseudorandom property (i.e., bounded independence) of $\mathbf{h}_1,\dots, \mathbf{h}_t$ is exploited. Except for Lemma~\ref{lemma:extend-is-random}, all other claims and lemmas in Section~\ref{sec:cjww-extend} still hold if we fix a list of $h_1,\dots, h_t$ and $x$.
\end{remark}

\subsubsection{Coupling with the standard walk}\label{sec:extended-couple}

Suppose $S = \{\ell^1,\ell^2,\dots, \ell^c \}$ where $\ell^1  < \dots < \ell^c$. Fix $h=(h_1,\dots, h_t)$ and the starting vertex $x$. We discuss how $\walk^S(h,x)$ is related to $\walk^{\std}(h,x)$. Let $T = \walk^\std(h,x)$ and $\mathbf{T}^S \sim \walk^S(h,x)$. Note that $\walk^S(h, x)$ does not deviate from $\walk^{\std}(h, x)$ until
\begin{enumerate}
    \item the walk reaches a point $\ell^4$ where $T(\ell^3) \simeq T(\ell^4)$ for some $\ell^3 < \ell^4$, or
    \item the standard walk makes $\tau$ consecutive level-$i$ moves inside a single call $\extwalk(i,*, *)$.
\end{enumerate}
If neither of the two cases happens before visiting $\ell^c$, with probability $1$ we have $(\mathbf{T}^S(\ell^1),\dots,\mathbf{T}^S(\ell^c)) = (T(\ell^1),\dots, T(\ell^c))$. On the other hand, seeing $\mathbf{T}^S(\ell^i) = u_i$ does not necessarily indicate $T(\ell^i) = u_i$. By the discussion above, there might be two types of ``false positives'', and we have the following two types of refutations for invalid contributions.

% We discuss how $\walk^{S}$ is related to $\walk^{\std}$. For illustration, we first consider the case that $S$ contains two indices $\{\ell^1, \ell^2\}$ and we want to lower bound the probability appeared in the right hand side of \eqref{eq:translate-to-tree}. Later we shall generalize the results to all finite $S$. We start by defining a predicate on the hash $h = (h_1,\dots, h_t)$ and starting vertex $x\in S$:
% \[
% \begin{aligned}
% & \good^{\ell^1,\ell^2}(h, x) :=  \\
% & \mathbbm{1}\left\{ \{T(\ell^1), T(\ell^2)\} = \{u, v\} \land \left( \forall \ell^3 < \ell^4 < \ell^2, T(\ell^3) \not\simeq T(\ell^4) \right) \land \left(\forall \ell^5 \le \ell^2, \width(\ell^5) \ge \tau \Rightarrow T(\ell^5) = \star \right)\right\}.
% \end{aligned}
% \]
% Therefore, if $\good^{\ell^1,\ell^2}(h, x) = 1$, then with probability $1$ we have $\left\{\mathbf{T}^S(\ell^1), \mathbf{T}^S(\ell^2)\right\} = \{u, v\}$.

\begin{definition}\label{def:type-1}
Let $S$ be a set of indices and $\max S$ be the largest index in $S$. For a list of hashes $h_1,\dots, h_t$ and a starting vertex $x\in [n]$, let $T^S\in \supp(\walk^S(h, x))$. Let $r^1,r^2$ be two indices such that $r^1<r^2 < \max S$. We call $\{r^1,r^2\}$ a \emph{type-$1$ refutation} for $T^S$, if $T^S(r^1) \simeq T^S(r^2)$.
\end{definition}

\begin{definition}\label{def:type-2}
Let $S$ be a set of indices and $\max S$ be the largest index in $S$. For a list of hashes $h_1,\dots, h_t$ and a starting vertex $x\in [n]$, let $T^S\in \supp(\walk^S(h, x))$. Let $r < \max S$ be an index with $\width(r) = \tau$. We call $\{r\}$ a \emph{type-$2$ refutation} for $T^S$, if $T^S(r) \ne \star$.
\end{definition}

% Consider the extended walk $\mathbf{T}^S$. If $\mathbf{T}^S$ admits a refutation of either type, it implies that $\walk^S$ either enters a loop before reaching $\ell^2$ or makes $\tau$ consecutive level-$i$ moves without breaking out to the higher level. In either case, the contribution from this $\mathbf{T}^S$ would be a ``false positive'' even if we observed $\{\mathbf{T}^S(\ell^1), \mathbf{T}^S(\ell^2)\} = \{u, v\}$. On the other hand, if $\mathbf{T}^S$ does not admit any refutation, then we know $\walk^S$ simulates $\walk^{\std}$ faithfully up to the point $\ell^2$, and the contribution from this $\mathbf{T}^S$ would be valid. Hence, the following equation is established.
% \begin{align}
% &~~~~ \Pr_{\mathbf{h}, \mathbf{x}}[\good^{\ell^1, \ell^2}(\mathbf{h}, \mathbf{x})] \notag \\ 
% &\ge 
% \Pr_{\mathbf{h},\mathbf{x}, \mathbf{T}^S}\left[ \left\{\mathbf{T}^S(\ell^1), \mathbf{T}^S(\ell^2)\right\} = \{u, v\} \right] - \notag \\
% &~~~~ \sum_{\ell_3<\ell_4<\ell_2}\Pr_{\mathbf{h},\mathbf{x}, \mathbf{T}^S}\left[ \left\{\mathbf{T}^S(\ell^1), \mathbf{T}^S(\ell^2)\right\} = \{u, v\} \land (\ell^3,\ell^4) \text{ is a type-$1$ refutation for $\mathbf{T}^S$} \right] - \notag \\
% &~~~~ \sum_{\ell_5<\ell_2}\Pr_{\mathbf{h},\mathbf{x}, \mathbf{T}^S}\left[ \left\{\mathbf{T}^S(\ell^1), \mathbf{T}^S(\ell^2)\right\} = \{u, v\} \land (\ell^5) \text{ is a type-$2$ refutation for $\mathbf{T}^S$} \right]. \label{eq:naive-subtract}
% \end{align}

Since refutations involve indices (e.g. $r$) that may not belong to $S$. It is not clear how we can analyze $\mathbf{T}^S(r)$. Let $T^S\in \supp(\walk^S(h,a))$. In the following, we propose a notion of ``surgical refutation'' for $T^S$. We show 
\begin{itemize}
    \item If $T^S$ admits at least one refutation (of either type), it also admits a surgical refutation.
    \item It is easy to analyze the probability of having a surgical refutation.
\end{itemize}
% Having established the two claims, we can finally lower bound $\Pr_{\mathbf{h}, \mathbf{x}}[\good^{\ell^1, \ell^2}(\mathbf{x}, \mathbf{x})]$ by $c\cdot \frac{2^{-|\ell_1|-|\ell_2|}}{n^2}$. Plugging this bound back in \eqref{eq:main-lowerbound} concludes the proof.

\paragraph*{Fine-grained subtraction of invalid contribution.} Now we present the definition and analysis of ``surgical refutations'' mentioned above. In this part, we fix $(h_1,\dots, h_t)$ and $x\in [n]$. For every $T^S\in \supp(\walk^S(h,x))$, we study the structure of $T^S$. We associate with $T^S$ the information $\probe:\mathbb{N}^{t}\to [m]\cup \{\star\}$ and $\pre:\mathbb{N}^{t}\to [m]\cup \{\star\}$. They are defined as follows.
\begin{itemize}
\item For every $\ell \ne \vec{0}$, define
\[
\probe(\ell) = 
\begin{cases}
a_x, & \text{ If the walk is at vertex $x$ before updating $T^S(\ell)$ in Line 24 of Algorithm~\ref{algo:extended-walk}, } \\
\star, & \text{Otherwise (i.e., Algorithm~\ref{algo:extended-walk} has never tried to write $T^S(\ell)$).}
\end{cases}
\]
\item For every $\ell$ with $\probe(\ell) \ne \star$, we define $\pre(\ell) := \sup\{ r : r < \ell, T^S(r) \ne \star\}$. By definition, we have $\probe(\ell) = a_{T^S(\pre(\ell))}$. Intuitively, $\pre(\ell)$ is the last entry of $T^S$ updated before updating $T^S(\ell)$. We only consider $\pre(\ell)$ for those $\ell$'s with $T^S(\ell) \ne \star$. We observe that $\pre(\ell)$ is either the parent\footnote{See Definition \ref{def:walk-tree} for the definition of parent.} of $\ell$, or a descendant of $\ell$'s parent.
\end{itemize}

\begin{definition}\label{def:surgical}
Let $T^S\in \supp(\walk^S(h, x))$. Let $A$ be a set of at most two indices. We call $A$ a \emph{surgical} refutation for $T^S$, if all of the following hold.
\begin{itemize}
    \item $A$ is a type-$1$ or type-$2$ refutation as per Definition~\ref{def:type-1} or~\ref{def:type-2}.
    \item For every $i\in [t]$ and every two distinct level-$i$ indices $r^1, r^2\in P(A) \setminus P(S)$, it holds that $\probe(r^1) \ne \probe(r^2)$.
    \item For every $i\in [t]$, every two level-$i$ indices $r^1\in P(A) \setminus P(S)$ and $r^2\in P(S)$, it holds that $\probe(r^1) \ne \probe(r^2)$.
\end{itemize}
\end{definition}

We say two indices $r^1, r^2$ are \emph{conflicting}, if they are at the same level and $\probe(r^1) = \probe(r^2)$. We establish two lemmas about surgical refutations. First, we shall prove that if a tensor $T^S$ admits a refutation, it must admit a surgical refutation.

\begin{lemma}\label{lemma:exists-surgical}
Let $T^S\in \supp(\walk^S(h, x))$. If $T^S$ admits a refutation of either type, then $T^S$ also admits a surgical refutation.
\end{lemma}

\begin{proof}
Depending on whether $T^S$ admits any type-$1$ refutation, we consider two cases. 

\paragraph*{Case 1.} Suppose $T^S$ does not admit any type-$1$ refutation. Then for every $r^1<r^2 < \max(S)$, it holds that $T^S(r^1)\not\simeq T^S(r^2)$. Consequently, we have $\probe(r^1)\ne \probe(r^2)$ for every $r^1 < r^2\le \max(S)$. Hence, any type-$2$ refutation for $T^S$ would be a surgical refutation.

\vspace{-3mm}
\paragraph*{Case 2.} Now let us consider the case that $T^S$ admits at least one type-$1$ refutation. Take $A = \{\ell^1, \ell^2\}$ to be an arbitrary type-$1$ refutation. If $A$ is a surgical refutation, we are done. Otherwise, take the \emph{largest} $i\in [t]$ such that the surgical requirement (Definition~\ref{def:surgical}) is violated at level $i$. We try to find a new refutation based on $A$, as follows.

Suppose there are $b > 0$ level-$i$ indices in $P(A)\setminus P(S)$. We sort them in the increasing order. Let $r^1 < r^2 < \dots < r^b$ denote the sorted sequence. Choose the first $q\in [b]$ such that one of the following holds.
\begin{itemize}
    \item There is $p < q$ such that $r^q$ is conflicting with $r^p$.
    \item There is $r^*\in P(S)$ such that $r^q$ is conflicting with $r^*$.
\end{itemize}
In the former case, let $\ell^3 = \pre(r^q), \ell^4 = \pre(r^p)$. In the latter case, let $\ell^3 = \pre(r^q), \ell^4 = \pre(r^*)$. 

Now we consider $A' = \{\ell^3, \ell^4\}$. Since $T^S(\ell^3)\simeq T^S(\ell^4)$, $A'$ is a refutation. For every $r\in P(A')$ of level $j > i$, we have $r\in P(A) \cup P(S)$. Therefore, by our choice of $i$, indices in $P(A')\setminus P(S)$ of level larger than $i$ are not conflicting with indices in $P(A')\cup P(S)$. For every level-$i$ index $r\in P(A')\setminus P(S)$, we have $r = r^u$ for some $u < q$. By our choice of $q$, $r$ is not conflicting with indices in $P(A') \cup P(S)$.

Therefore, given a non-surgical refutation that violates the requirement at level $i$, we can find a new refutation $A'$ for $T^S$ where the surgical requirement is satisfied for every level $j\ge i$. If $A'$ still fails to be a surgical refutation, we take the \emph{largest} $j < i$ such that the surgical requirement is violated at level $j$. We use the same procedure to replace $A'$ with another refutation that satisfies the requirement for level $j$ and above. Since there are only $t$ levels in total, we can repeat this process until finding a surgical refutation. This completes the proof.
\end{proof}

Our second lemma shows that surgical refutations are easier to analyze.

\begin{lemma}\label{lemma:surgical-upper-bound}
Fix $h=(h_1,\dots, h_t)$ and $x\in [n]$. Let $S = \{\ell^1,\dots, \ell^c\}$ be a set of $\tau$-bounded indices and $A$ be a set of at most two indices. Let $\mathbf{T}^S \sim \walk^S(h, x)$ and $\mathbf{T}^{S\cup A} \sim \walk^{S\cup A}(h, x)$. Then, for every $u_1,\dots, u_c \in [n]$, it holds that
\[
\begin{aligned}
& ~~~~ \Pr_{\mathbf{T}^S}\left[ \left( \forall i\in [c],  \mathbf{T}^S(\ell^i)  = u_i  \right) \land A \text{ is a surgical refutation for $\mathbf{T}^S$} \right] \\
&\le \Pr_{\mathbf{T}^{S\cup A}}\left[ \left( \forall i\in [c],  \mathbf{T}^{S\cup A}(\ell^i)  = u_i  \right) \land A \text{ is a refutation for $\mathbf{T}^{S\cup A}$} \right].
\end{aligned}
\]
\end{lemma}

Intuitively, Lemma~\ref{lemma:surgical-upper-bound} says that we can reduce the question of analyzing surgical refutations to analyzing general refutations in a related extended walk (i.e., $\walk^{S\cup A}$). The latter question is much easier: since we only observe the ``extended'' part in $\walk^{S\cup A}$, we can upper bound it just by Lemma~\ref{lemma:extend-is-random}. 

\begin{proof}
Let $T^* \in \supp(\walk^S(h,x))$ be an instantiation of $\mathbf{T}^S$ satisfying the predicate in the lemma statement. Namely, $T^*(\ell^i) = u_i$ for every $i\in [c]$ and $A$ is a surgical refutation for $T^*$ with respect to $S$. We claim
\[
\Pr_{\mathbf{T}^S}[\mathbf{T}^S = T^*] = \Pr_{\mathbf{T}^{S\cup A}}[\mathbf{T}^{S\cup A} = T^*].
\]
Taking a summation over all such $T^*$ would conclude the proof, as
\[
\begin{aligned}
& ~~~~ \Pr_{\mathbf{T}^S}\left[ \left( \forall i\in [c],  \mathbf{T}^S(\ell^i)  = u_i  \right) \land A \text{ is a surgical refutation for $\mathbf{T}^S$} \right] \\
& = \sum_{T^*} \Pr_{\mathbf{T}^S}[\mathbf{T}^S = T^*] \cdot \mathbbm{1}\left\{\left( \forall i\in [c],  T^*(\ell^i)  = u_i  \right) \land A \text{ is a surgical refutation for $T^*$}\right\} \\
& \le \sum_{T^*} \Pr_{\mathbf{T}^{S\cup A}}[\mathbf{T}^{S\cup A} = T^*] \cdot \mathbbm{1}\left\{\left( \forall i\in [c],  T^*(\ell^i)  = u_i  \right) \land A \text{ is a refutation for $T^*$}\right\} \\ 
&= \Pr_{\mathbf{T}^{S\cup A}}\left[ \left( \forall i\in [c],  \mathbf{T}^{S\cup A}(\ell^i)  = u_i  \right) \land A \text{ is a refutation for $\mathbf{T}^{S\cup A}$} \right].
\end{aligned}
\]

% Now we justify the claim. Since $A$ is a surgical refutation for $T^*$, for every two distinct level-$i$ indices $r^1, r^2\in P(A)\setminus P(S)$, we have $\probe(r^1)\ne \probe(r^2)$. Also for every $r^1\in P(A)\setminus P(S)$ and $r^2\in P(S)$, we have $\probe(r^1) \ne \probe(r^2)$.

Now we justify the claim. Imagine running $\walk^{S}(h, x)$ and $\walk^{S\cup A}(h, x)$ in parallel. During the execution of $\walk^S(h,a)$ and $\walk^{S\cup A}(h,a)$, we say a \emph{regular conflict} happens at level $i$, if there are two level-$i$ vertices $r^1, r^2 \in P(S)$ such that $r^1\ne r^2$ and $\probe(r^1) = \probe(r^2)$. We say an \emph{extra conflict} happens at level $i$, if there are two distinct level-$i$ indices $r^1,r^2\in P(A\cup S)$ such that $\probe(r^1) = \probe(r^2)$ and $\{r^1,r^2\} \not\subseteq P(S)$. Note that both $\walk^S$ and $\walk^{S\cup A}$ need to sample a one-time edge (Line $13$ of Algorithm~\ref{algo:extended-walk}) whenever there is a regular conflict. In addition, $\walk^{S\cup A}$ also needs to sample a one-time edge for each extra conflict.

To produce $T^*$ in $\walk^S(h, x)$, there is only one way to sample edges for regular conflicts. Moreover, since $A$ is a surgical refutation for $T^*$ \footnote{We stress that the definition of surgical refutation for $T^*$ always implicitly depends on $S$.}, by Definition~\ref{def:surgical}, there is no extra conflict in producing $T^*$. Therefore, if $\walk^{S\cup A}(h, x)$ samples the same edge as $\walk^S(h, x)$ does for each regular conflict, $\walk^{S\cup A}$ will not encounter any extra conflict and will output the same result $T^*$. This shows that $\Pr_{\mathbf{T}^S}[\mathbf{T}^S = T^*] = \Pr_{\mathbf{T}^{S\cup A}}[\mathbf{T}^{S\cup A} = T^*]$ as desired.
\end{proof}

Combining Lemma~\ref{lemma:exists-surgical}, \ref{lemma:surgical-upper-bound} and the discussion, the following lemma is established.

\begin{lemma}\label{lemma:extend-lb-standard}
For every list of hashes $h_1,\dots, h_t$ and starting vertex $x\in [n]$, consider the standard walk $T = \walk^{\std}(h, x)$. For every set $S= \{ \ell^1 < \ell^2 < \dots < \ell^c\}$ of $\tau$-bounded indices and every $u_1,\dots, u_c\in [n]$, it holds that
\[
\begin{aligned}
&~~~~ \mathbbm{1}\left\{ T(\ell^i) = u_i, \forall i\in [c] \right\} \\
&\ge \Pr_{\mathbf{T}^S\sim \walk^S(h, x)} \left[ \mathbf{T}^S(\ell^i) = u_i, \forall i\in [c] \right] -\\
&~~~~ \sum_{\substack{r^1 < r^2 < \ell^c \\ A = \{r^1, r^2\}}} \Pr_{\mathbf{T}^{S\cup A} \sim \walk^{S\cup A}(h, x)} \left[ \left(  \mathbf{T}^{S\cup A}(\ell^i) = u_i, \forall i\in [c] \right) \land (\mathbf{T}^{S\cup A}(r^1) \simeq \mathbf{T}^{S\cup A}(r^2)) \right] - \\
&~~~~ \sum_{\substack{r^3 < \ell^c \\ A = \{r^3\}: \width(r^3) = \tau} } \Pr_{\mathbf{T}^{S\cup A} \sim \walk^{S\cup A}(h, x)} \left[ \left(  \mathbf{T}^{S\cup A}(\ell^i) = u_i, \forall i\in [c] \right) \land \mathbf{T}^{S\cup A}(r^3)\ne \star \right].
\end{aligned}
\]
\end{lemma}

\begin{proof}
If $\mathbbm{1}\left\{ T(\ell^i) = u_i, \forall i\in [c] \right\} = 1$, the inequality holds because the right hand side is bounded by $1$. Now suppose that $\mathbbm{1} \left\{ T(\ell^i) = u_i, \forall i\in [c] \right\} = 0$. For every $T^*\in \supp(\walk^S(h,x))$ such that $T^*(\ell^i) = u_i, \forall i\in [c]$, by Lemma~\ref{lemma:exists-surgical} there is a surgical refutation $A$ for $T^*$. By Lemma~\ref{lemma:surgical-upper-bound}, the contribution from $T^*$ is subtracted in the probability term regarding $\mathbf{T}^{S\cup A} \sim \walk^{S\cup A}(h,x)$.
\end{proof}

Lemma~\ref{lemma:extend-lb-standard} is somewhat remarkable in that it holds for every fixed $h$ and $x$. One can extend Lemma~\ref{lemma:extend-lb-standard} to an average version by taking expectations over $\mathbf{h}$ and $\mathbf{x}$ (We will use the average version in all applications).

\subsection{Wrapping-up}\label{sec:cjww-wrapup}

We are ready to prove Lemma~\ref{lemma:single-ub} and \ref{lemma:double-lb}. We start with the easier one.

\begin{reminder}{Lemma~\ref{lemma:single-ub}}
    For every $t \le \frac{1}{2}\log n$, let $\kappa = 20\log n$. Sample $\mathbf{h}\sim \calH^{n,m,t,\kappa}$ and $\mathbf{x}\sim [n]$. Then for every $u\in [n]$, it holds that:
    \[
    \Pr_{\mathbf{h}, \mathbf{x}}[u\in \Out_{a,\mathbf{h}}(\mathbf{x})] \le O\left(\frac{2^t}{n}\right).
    \]
\end{reminder}

\begin{proof}
For every fixed $h_1,\dots, h_t$ and $x$, consider $T = \walk^{\std}(h,x)$.
% \xin{TODO: fix this "with probability 1" thing.}
\begin{itemize}
    \item We say $T$ is a good walk, if $\walk^{\std}(h, x)$ does not make $\tau$ consecutive level-$t$ moves inside a function call. For every $\tau$-bounded index $\ell$ such that $T(\ell) = u$, we have $\mathbf{T}^{\{\ell\}}(\ell) = u$ almost surely\footnote{Recall $h_1,\dots,h_t$, $x$ have been fixed. The probability is only over the internal randomness of $\walk^{\{\ell\}}$. Moreover, since $T$ is good and $T(\ell) \ne \star$, we know $\walk^{\{\ell\}}$ does not deviate from $\walk^{\std}$ before reaching $\ell$.}. We can upper-bound the probability that $T$ is good and hits $u$ by summing up $\sum_{\ell} \Pr[\mathbf{T}^{\{\ell\}}(\ell) = u]$.
    \item Otherwise, we say $T$ is a bad walk. For a worst case analysis, we may assume that a bad walk always hits $u$. To bound the probability of being bad, suppose $\ell$ is the smallest index with $\width(\ell) = \tau$ and $T(\ell) \ne \star$. Then, we have $\mathbf{T}^{\{\ell\}}(\ell)\ne \star$ almost surely (because $\walk^{\{\ell\}}$ does not deviate from $\walk^{\std}$ before reaching $\ell$). Therefore, we can upper-bound the probability of being a bad walk by $\sum_{\ell:\width(\ell)=\tau}\Pr[\mathbf{T}^{\{\ell\}}(\ell)\ne \star]$.
\end{itemize}
Now we take expectation over $\mathbf{h}$ and $\mathbf{x}$. We have
\[
\Pr_{\mathbf{h},\mathbf{x}}[u\in \Out_{a,\mathbf{h}}(\mathbf{x})] = \sum_{\ell} \Pr_{\mathbf{T}^{\ell}\sim \walk^{\{\ell\}}(\mathbf{h},\mathbf{x})}[\mathbf{T}^{\{\ell\}} = u] + \sum_{\ell:\width(\ell) = \tau} \Pr_{\mathbf{T}^{\ell}\sim \walk^{\{\ell\}}(\mathbf{h},\mathbf{x})}[ \mathbf{T}^{\{\ell\}}(\ell)\ne \star]\\
\]
Using Lemma~\ref{lemma:extend-is-random}, the first term is bounded by
\[
\sum_{\ell}\frac{2^{1-|P(\ell)|}}{n} = \frac{1}{n} \left( \prod_{i=1}^{t} \left( \sum_{\ell_i=0}^{\infty}  2^{-\ell_i} \right)\right) \le \frac{2^t}{n}.
\]
For the second term, we enumerate $d\in[t]$ such that $\ell_d = \tau$. Recall $\tau = 5\log n$. Using Lemma~\ref{lemma:extend-is-random} again, we can bound the second term by
\[
\sum_{d=1}^{\tau} 2^{-\tau}\cdot 2^{t} \le 2^{-\tau}\cdot \tau \cdot2^{t} \le \frac{1}{n^3}.
\]
Combing two bounds together completes the proof.
\end{proof}

Before proving Lemma~\ref{lemma:double-lb}, we need one more technical lemma and its corollary. The proofs are in Appendixs~\ref{app:technical}.

\begin{lemma}\label{theo:counting}
For any fixed positive integer $c$ and $t$, let $S^t$ be the set of all $t$-dimensional indices. Denote $f(c,t) = \sum_{S\subseteq S^t, |S|=c}2^{1-|P^{(t)}(S)|}$. Then, we have
\[
f(c,t) \le 2^{c\cdot t}.
\]
\end{lemma}

\begin{corollary}\label{corol:counting}
For any fixed positive integer $c$ and $t$, it holds that
\[
\sum_{\ell^1,\ldots,\ell^c\in S^t}2^{1-|P^{(t)}(\{\ell^1,\ldots,\ell^c\})|}\le c!\cdot 2^{c\cdot (t+1)}
\]
\end{corollary}

We prove Lemma~\ref{lemma:double-lb} now. Recall the statement.

\begin{reminder}{Lemma~\ref{lemma:double-lb}}
    Let $t = \frac{1}{2}\log n$ and $\kappa = 20\log n$. Sample $\mathbf{h}\sim \calH^{n,m,t,\kappa}$ and $\mathbf{x}\sim [n]$. Then for every $u, v\in [n]$, $u\ne v$, it holds that:
    \[
    \Pr_{\mathbf{h}, \mathbf{x}}[u, v\in \Out_{a,\mathbf{h}}(\mathbf{x})] \ge \Omega\left( \frac{1}{F_2(a)} \right).
    \]
\end{reminder}

\begin{proof}

First note that $\frac{1}{2n^2}$ is a trivial lower bound: with probability $\frac{1}{n}$ we have $\mathbf{x} = u$. Conditioning on this being true, with probability $\frac{1}{2n}$ we have $\mathbf{h}(\mathbf{x}) = v$. In the following, we assume that $F_2(a)\le c n^2$ holds for a sufficiently small constant $c$. In this case, we prove a lower bound of $\Omega(1/F_2(a))$.

Let us first recall \eqref{eq:translate-to-tree}, re-stated below.
\begin{align}
\Pr_{\mathbf{h_1},\dots, \mathbf{h}_t, \mathbf{x}}[u,v\in \Out_{a,\mathbf{h}}(\mathbf{x})] 
= \sum_{\ell^1 \ne \ell^2}\Pr_{\mathbf{h}_1,\dots,\mathbf{h}_t, \mathbf{x}}\left[   \mathbf{T}(\ell^1) = u \land \mathbf{T}(\ell^2) = v \right] \label{eq:translate-to-tree-2}
\end{align}

We assume $t:= \frac{1}{2}\log n\ge 5$. Let $5\le \gamma \le t$ be a parameter to be specified later. Let $I^{\gamma} = \{ \ell\in [0, \tau]^t : \ell_{> \gamma} \equiv 0 \}$ be the set of indices with its $(t-\gamma)$-suffix being all-zero. We only consider the contribution of $\tau$-bounded pairs from $I^\gamma\times I^{\gamma}$ to the right hand side of \eqref{eq:translate-to-tree-2}. Fix one such pair $(\ell^1, \ell^2)$. Suppose $\ell^1 < \ell^2$ (the case that $\ell^1 > \ell^2$ is analogous). By Lemma~\ref{lemma:extend-is-random} and \ref{lemma:extend-lb-standard}, we have:
\[
\begin{aligned}
& ~~~~ \Pr_{\mathbf{h}_1,\dots,\mathbf{h}_t, \mathbf{x}}\left[  (\mathbf{T}(\ell^1), \mathbf{T}(\ell^2)) = (u,v) \right]  \\
& \ge \Pr_{\mathbf{T}^{e} \sim \mathbf{T}^{\{\ell^1,\ell^2\}}}[ (\mathbf{T}^e(\ell^1), \mathbf{T}^e(\ell^2)) = (u, v) ] - & (\triangleq E_1(\ell^1,\ell^2)) \\
& ~~~~ \sum_{\ell^3 < \ell^4 < \ell^2} \Pr_{\mathbf{T}^e \sim \mathbf{T}^{\{\ell^1,\ell^2, \ell^3, \ell^4\}}}[ ( \mathbf{T}^e(\ell^1), \mathbf{T}^e(\ell^2)) = (u, v) \land \mathbf{T}^e(\ell^3) \simeq \mathbf{T}^e(\ell^4) ] -  & (\triangleq E_2(\ell^1,\ell^2)) \\
& ~~~~ \sum_{\ell^5 < \ell^2: \width(\ell^5) = \tau} \Pr_{\mathbf{T}^e \sim \mathbf{T}^{\{\ell^1,\ell^2,  \ell^5\}}}[ (\mathbf{T}^e(\ell^1), \mathbf{T}^e(\ell^2))= (u, v) \land \mathbf{T}^e(\ell^5) \ne \star ]  & ( \triangleq E_3(\ell^1,\ell^2)).\\
%& \ge \frac{2^{1-|P(\{\ell^1,\ell^2\})|}}{n^2} - \sum_{\ell^3, \ell^4} \frac{2^{1-|P(\{\ell^1,\ell^2, \ell^3,\ell^4\})|} \cdot F_2(a) }{n^4} - \sum_{\ell^5:\width(\ell^5) = \tau} \frac{2^{1-|P(\{\ell^1,\ell^2, \ell^5\})|} }{n^2}
\end{aligned}
\]
Note that we use $E_i(\ell^1,\ell^2), i\in \{1,2,3\}$ to denote three quantities on the right hand side. Now, we sum up all $\tau$-bounded pairs from $I^\gamma\times I^\gamma$, and bound the summation of $E_i(\ell^1,\ell^2), i\in [3]$ by a series of straightforward but somewhat lengthy manipulation. 

First, for $E_1(\ell^1,\ell^2)$ we have
\[
\mathrm{SE}_1 := \sum_{\ell^1\ne \ell^2} E_1(\ell^1,\ell^2) = \sum_{\ell^1\ne \ell^2} \frac{2^{1-|P(\{\ell^1,\ell^2\})|}}{n^2} = \sum_{\ell^1, \ell^2}\frac{2^{1-|P(\{\ell^1,\ell^2\})|}}{n^2}-\sum_{\ell}\frac{2^{1-|P(\ell)|}}{n^2}. 
\]
% \begin{align}
% & ~~~~ \Pr_{\mathbf{h_1},\dots, \mathbf{h}_t, \mathbf{x}}[u,v\in \Out_{a,\mathbf{h}}(\mathbf{x})] \notag \\
% %& = \sum_{\ell^1 < \ell^2}  \Pr_{\mathbf{h}_1,\dots,\mathbf{h}_t, \mathbf{x}}\left[  \{ \mathbf{T}(\ell^1), \mathbf{T}(\ell^2)\} = \{u, v\} \land \left( \forall \ell^3 < \ell^4 < \ell^2, \mathbf{T}(\ell^3) \not\simeq \mathbf{T}(\ell^4) \right)\right] \\
% & \ge \sum_{\ell^1<\ell^2} \frac{2^{1-|P(\{\ell^1,\ell^2\})|}}{n^2}-\sum_{\ell^1, \ell^2, \ell^3, \ell^4} \frac{2^{1-|P(\{\ell^1,\ell^2, \ell^3,\ell^4\})|} \cdot F_2(a) }{n^4} -\sum_{\ell^1,\ell^2,\ell^5:\width(\ell^5) = \tau} \frac{2^{1-|P(\{\ell^1,\ell^2, \ell^5\})|}}{n^2}.\label{eq:cjww-3-terms} 
% \end{align}
%We emphasize that all the involved indices are enumerated over $I^\gamma$ (instead of $[0, \tau]^t$). Next, we bound $3$ terms above separately by a series of straightforward but somewhat lengthy manipulation. For the first one, we observe that
% \begin{align}
% & \sum_{\ell^1<\ell^2} \frac{2^{1-|P(\{\ell^1,\ell^2\})|}}{n^2} = \frac{1}{2} \left(\sum_{\ell^1, \ell^2}\frac{2^{1-|P(\{\ell^1,\ell^2\})|}}{n^2}-\sum_{\ell}\frac{2^{1-|P(\ell)|}}{n^2} \right)  \label{eq:cjww-first-term}
% \end{align}
Using Lemma~\ref{lemma:combinatorics}, we have $\sum_{\ell}\frac{2^{1-|P(\ell)|}}{n^2} \le \frac{2^{\gamma}}{n^2}$. Note that $|P(\{\ell^1,\ell^2\})|\le |P(\ell^1)|+|P(\ell^2)|-1$. Therefore,
\begin{align}
\sum_{\ell^1,\ell^2\in I^\gamma\times I^\gamma} \frac{2^{1-|P(\{\ell^1,\ell^2\})|}}{n^2}
&\ge \left(\sum_{\ell}\frac{2^{1-|P(\ell)|}}{n}\right)^2 \notag \\
&=\frac{1}{n^2}\left(\sum_{\ell}\prod_{i=1}^{\gamma}2^{-\ell_i}\right)^2 \notag \\
&= \frac{1}{n^2}\left(\prod_{i=1}^{\gamma}\sum_{\ell_i=0}^{\tau}2^{-\ell_i}\right)^2 \notag \\
&= \frac{1}{n^2}\cdot (2-2^{-\tau})^{2\gamma} \notag  \\
&\ge \frac{2^{2\gamma}}{2 n^2}. \label{eq:cjww-first-trick}
\end{align}
Here, the last inequality holds since $\gamma \le t \ll 2^{\tau}$. Therefore, we conclude that
\begin{align}
\mathrm{SE}_1 = \sum_{\ell^1\ne \ell^2} \frac{2^{1-|P(\{\ell^1,\ell^2\})|}}{n^2} \ge  \frac{2^{2\gamma}}{2n^2} - \frac{2^{\gamma}}{n^2} \ge \frac{2^{2\gamma}}{3n^2}.\label{eq:cjww-end-1}
\end{align}
The last inequality holds because we have assumed that $\gamma \ge 5$. 

Next, we turn to $\sum_{\ell^1\ne \ell^2}E_2(\ell^1,\ell^2)$. We have
\[
\mathrm{SE}_2 := \sum_{\ell^1 \ne \ell^2} E_2(\ell^1,\ell^2) \le \sum_{\ell^1, \ell^2, \ell^3, \ell^4} \frac{2^{1-|P(\{\ell^1,\ell^2, \ell^3,\ell^4\})|} \cdot F_2(a) }{n^4} + \sum_{\ell^1,\ell^2,r} \frac{2^{1-|P\{\ell^1,\ell^2, r\}|}\cdot F_{\infty}(a) }{n^3}.
\]
We explain the right hand side. For every fixed pair $\ell^1 < \ell^2$ (the case that $\ell^1 > \ell^2$ is analogous), the first summation considers the case that $\ell^1\notin \{\ell^3,\ell^4\}$. In this case, with probability $\frac{F_2(a)}{n^2}$ we have $\mathbf{T}^e(\ell^3)\simeq \mathbf{T}^e(\ell^4)$. The second term considers the case that $\ell^1\in \{\ell^3,\ell^4\}$. In this case, let $r\in \{\ell^3, \ell^4\} \setminus \{\ell^1\}$. Conditioning on $\mathbf{T}^e(\ell^1) = u, \mathbf{T}^e(\ell^2)=v$ (which happens with probability $\frac{1}{n^2}$), the probability that $\mathbf{T}^e(r) \simeq \mathbf{T}^e(\ell^1)$ is bounded by $\frac{F_{\infty}(a)}{n}$. Also note that the summation on the right hand side may enumerate $\ell^3,\ell^4$ such that $\ell^3 = \ell^4$. This would not be a problem since we are upper-bounding $\mathrm{SE}_2$.

Since indices in $I^{\gamma}$ can be equivalently seen as $\gamma$-dimensional indices, we use Corollary \ref{corol:counting} to deduce that
\[
\begin{aligned}
 & \sum_{\ell^1, \ell^2, \ell^3, \ell^4} \frac{2^{1-|P(\{\ell^1,\ell^2, \ell^3,\ell^4\})|} \cdot F_2(a) }{n^4} 
\le \frac{F_2(a)}{n^4}\cdot 24\cdot  2^{4(\gamma+1)},  \\
& \sum_{\ell^1,\ell^2,r} \frac{2^{1-|P\{\ell^1,\ell^2, r\}|}\cdot F_{\infty}(a) }{n^3} \le \frac{\sqrt{F_2(a)}}{n^3} \cdot 6 \cdot 2^{3(\gamma+1)}.
\end{aligned}
\]
Consequently, 
\begin{align}
    \mathrm{SE}_2 \le \frac{F_2(a)}{n^4}\cdot 24\cdot  2^{4(\gamma+1)} + \frac{\sqrt{F_2(a)}}{n^3} \cdot 6 \cdot 2^{3(\gamma+1)}. \label{eq:cjww-end-2}
\end{align}
Now we handle the last summation: $\sum_{\ell^1 \ne \ell^2} E_3(\ell^1,\ell^2)$.
\[
\begin{aligned}
    \mathrm{SE}_3 := \sum_{\ell^1 \ne \ell^2} E_3(\ell^1,\ell^2) \le \sum_{\ell^1,\ell^2,\ell^5:\width(\ell^5) = \tau} \frac{2^{1-|P(\{\ell^1,\ell^2, \ell^5\})|}}{n^2}.
\end{aligned}
\]
For every integer $i\in [\gamma]$ and $t$-dimensional index $\ell$, define $\ell_{-i}:=(\ell_1,\ldots,\ell_{i-1},\ell_{i+1},\ldots,\ell_{t})$, which is a $(t-1)$-dimensional index. Since $\width(\ell^5) = \tau$, we enumerate $i\in [\gamma]$ and $\ell^5$ such that $\ell^5_i=\tau$. We observe that
\[
|P^{(t)}(\{\ell^1,\ell^2,\ell^5\})|\ge \tau+|P^{(t-1)}(\{\ell^1_{-i},\ell^2_{-i},\ell^5_{-i}\})|.
\]
Consequently, we have
\begin{align}
\mathrm{SE}_3 &\le \sum_{\ell^1,\ell^2,\ell^5:\width(\ell^5) = \tau} \frac{2^{1-|P(\{\ell^1,\ell^2, \ell^5\})|}}{n^2} \notag \\
&\le \frac{1}{n^2}\sum_{i=1}^{\gamma}~\sum_{\ell^1,\ell^2,\ell^5:\ell^5_i= \tau}2^{1-|P^{(\gamma)}(\{\ell^1,\ell^2, \ell^5\})|} \notag \\
&\le \frac{1}{n^2}\sum_{i=1}^{\gamma}~\sum_{\ell^1,\ell^2,\ell^5:\ell^5_i= \tau}2^{1-\tau-|P^{(\gamma-1)}(\{\ell^1_{-i},\ell^2_{-i}, \ell^5_{-i}\})|} \notag \\
&= \frac{1}{n^2}\sum_{i=1}^{\gamma}2^{-\tau}\cdot (\tau+1)^2\cdot \sum_{\ell^1_{-i},\ell^2_{-i},\ell^5_{-i}}2^{1-|P^{(\gamma-1)}(\{\ell^1_{-i},\ell^2_{-i}, \ell^5_{-i}\})|} \notag \\
&\le \frac{1}{n^2}\sum_{i=1}^{\gamma}2^{-\tau}\cdot (\tau+1)^2\cdot 3!\cdot 2^{3\gamma} & \text{(Corollary \ref{corol:counting})} \notag \\
&= \frac{6\gamma(\tau+1)^2}{n^2}\cdot 2^{3\gamma-\tau}. \label{eq:cjww-third-trick}
\end{align}
Combining \eqref{eq:cjww-end-1}, \eqref{eq:cjww-end-2} and \eqref{eq:cjww-third-trick} together, we have
\[
\begin{aligned}
&~~~~ \Pr_{\mathbf{h_1},\dots, \mathbf{h}_t, \mathbf{x}}[u,v\in \Out_{a,\mathbf{h}}(\mathbf{x})] \ge \mathrm{SE}_1 - \mathrm{SE}_2 - \mathrm{SE}_3 \\
&\ge \frac{2^{2\gamma}}{3 n^2} - \frac{F_2(a)}{n^4}\cdot 24\cdot 2^{4(\gamma+1)}- \frac{\sqrt{F_2(a)}}{n^3} \cdot 6 \cdot 2^{3(\gamma+1)} - \frac{6\gamma(\tau+1)^2}{n^2}\cdot 2^{3\gamma-\tau}. \\
\end{aligned}
\]
Recall the threshold $\tau = 5\log n$. Now, we set $\gamma =\frac{1}{2}\log(n^2/F_2(a))-C$, where $C = 10$. Since we have assumed that $F_2(a)\le cn^2$ holds for a sufficiently small constant $c > 0$, we may assume that $5\le \gamma \le t$. Then,
\[
\begin{aligned}
\Pr_{\mathbf{h_1},\dots, \mathbf{h}_t, \mathbf{x}}[u,v\in \Out_{a,\mathbf{h}}(\mathbf{x})] 
\ge \frac{1}{F_2(a)} \left( \frac{1}{3\cdot 2^{2C}} - \frac{24\cdot 2^4}{2^{4C}}  - \frac{6\cdot 2^3}{2^{3C}} \right) - O\left( \frac{1}{n^3} \right) \ge \Omega\left( \frac{1}{F_2(a)} \right),
\end{aligned}
\]
as desired.
\end{proof}

\section{Improved Time-Space Trade-Off}\label{sec:tradeoff}

In this section, we prove the main result of the paper: Theorem~\ref{theo:element-distinctness-algo} and \ref{theo:set-intersection-algo}. In Section~\ref{sec:tradeoff-setup}, we state a new property we need from the hash family (Lemma~\ref{lemma:multi-lb}). Assuming it, we prove Theorem~\ref{theo:element-distinctness-algo} and \ref{theo:set-intersection-algo} quickly. We prove Lemma~\ref{lemma:multi-lb} in the rest of the section.

\subsection{Setup and Proof of Main Results}\label{sec:tradeoff-setup}

Suppose we have $\Otilde(k)$ bits of working memory for solving \textsc{Element Distinctness} and \textsc{Set Intersection}. We will use the hash family $\calH^{n, m, t, \kappa}$ with parameter $t = \frac{1}{2}\log(n/k)$ and $\kappa = 20\log n$. We first state the pseudorandomness properties of $\calH^{n,m,t,\kappa}$ we need for the tradeoff algorithm. Assuming these properties, we show the algorithms.

\begin{lemma}\label{lemma:multi-ub}
    Let $t = \frac{1}{2}\log (n/k)$ and $\kappa = 20\log n$. Sample $\mathbf{h}\sim \calH^{n,m,t,\kappa}$ and $\mathbf{x}_1,\dots, \mathbf{x}_k\sim [n]$. For every $u\in [n]$, it holds that:
    \[
    \Pr_{\mathbf{h}, \mathbf{A}}[u\in \Out_{a,\mathbf{h}}(\{\mathbf{x}_i\})] \le O\left( \sqrt{\frac{k}{n}} \right).
    \]
\end{lemma}

Since Lemma~\ref{lemma:multi-ub} is an easy consequence of Lemma~\ref{lemma:single-ub}, we show its proof here.

\begin{proof}
We use Lemma~\ref{lemma:single-ub} and a union bound.
\[
\Pr_{\mathbf{h}, \mathbf{A}}[u\in \Out_{a,\mathbf{h}}(\{\mathbf{x}_i\})] \le \sum_{i=1}^{k} \Pr_{\mathbf{h}, \mathbf{x}_i}[u\in \Out_{a,\mathbf{h}}(\mathbf{x}_i)] \le k\cdot O\left( \frac{2^t}{n} \right) \le O\left( \sqrt{\frac{k}{n}} \right).
\]
\end{proof}

\begin{lemma}\label{lemma:multi-lb}
    Let $t = \frac{1}{2}\log (n/k)$ and $\kappa = 20\log n$. Sample $\mathbf{h}\sim \calH^{n,m,t,\kappa}$ and $\mathbf{x}_1,\dots, \mathbf{x}_k\sim [n]$. Then for every $u, v\in [n]$, $u\ne v$, it holds that:
    \[
    \Pr_{\mathbf{h}, (\mathbf{x}_i)}[u, v\in \Out_{a,\mathbf{h}}(\{\mathbf{x}_i\})] \ge \Omega\left( \frac{k}{F_2(a)} \right).
    \]
\end{lemma}

We defer the proof of Lemma~\ref{lemma:multi-lb} to Section~\ref{sec:proof-wrap-up}.

\paragraph*{Time-space upper bounds without random oracle.} Assuming Lemma~\ref{lemma:multi-lb}. We prove Theorem~\ref{theo:element-distinctness-algo} and \ref{theo:set-intersection-algo}.

\begin{reminder}{Theorem~\ref{theo:element-distinctness-algo}}
For every complexity bounds $S(n), T(n): \mathbb{N}\to \mathbb{N}$ such that $S(n)^{1/2}\cdot T(n) \ge n^{1.5}$, there is a Monte Carlo algorithm solving \textsc{Element Distinctness} in time $O(T(n)\cdot \polylog(n))$ and space $O(S(n)\cdot \polylog(n))$ with one-way access to random bits. Moreover, when there is a colliding pair, the algorithm reports one with high probability.
\end{reminder}

\begin{proof}
Let $a\in [m]^n$ be the input. Denote $S:=S(n)$. We set $t^{\mathsf{single}} = \frac{1}{2}\log(n)$, $t = \frac{1}{2}\log (n/S(n))$ and $\kappa = 20\log n$. 

\paragraph*{The algorithm.} Our algorithm first repeats the following trial for $\Theta(\log n)$ times.
\begin{itemize}
    \item Draw a random hash $\mathbf{h}\sim \calH^{n,m,t^{\mathsf{single}}, \kappa}$ and a starting vertex $\mathbf{x}\sim [n]$. Try to find a colliding pair by running $\COLLIDE(\mathbf{x})$ on $G_{a,\mathbf{h}}$.
\end{itemize}

The algorithm then repeats the following trial for $\Theta\left( \frac{n\log n}{S(n)} \right)$ times.

\begin{itemize}
    \item Draw a random hash $\mathbf{h}\sim \calH^{n,m,t,\kappa}$ and $S$ starting vertices $\mathbf{x}_1,\dots, \mathbf{x}_S\sim [n]$. Try to find a colliding pair by running $\COLLIDE(\left\{\mathbf{x}_1,\dots, \mathbf{x}_S\right\})$ on $G_{a,\mathbf{h}}$. 
\end{itemize}
The algorithm reports YES if it does not find any colliding pair. Otherwise it reports NO. 

\paragraph*{Time and space.} In expectation, the first bunch of trials takes $\Otilde(\sqrt{n})$ time by Lemma~\ref{lemma:bcm-cycle} and \ref{lemma:single-ub}. By Lemma~\ref{lemma:multi-ub} and \ref{lemma:bcm-cycle}, each trial in the second bunch takes $\Otilde(\sqrt{n\cdot S})$ time, which brings the total time to $\Otilde(n^{3/2}/S^{1/2})$. Since $S\le n$, the expected running time of the whole algorithm is $\Otilde(n^{3/2}/S^{1/2})$. The space complexity of the algorithm is $\Otilde(S) + O(\log^3 n) \le \Otilde(S)$.

\paragraph*{Correctness.} If $a$ is a YES instance, the algorithm always reports YES. If $a$ is a NO instance, depending on whether $F_2(a)\le 2n$, we consider two cases.
\begin{itemize}
    \item Suppose $F_2(a)\ge 2n$. In this case, by Lemma~\ref{lemma:double-lb} (see also the proof of Theorem~\ref{theo:improved-seed}), each trial in the first bunch succeeds in finding a colliding pair with probability $\Omega\left( \frac{F_2(a)-n}{F_2(a)} \right) \ge \Omega(1)$. With probability $1-n^{-\Omega(1)}$, at least one trial in the first bunch succeeds in finding a colliding pair.
    \item Otherwise we have $F_2(a)\le 2n \le O(n)$. Take $u\ne v$ to be a colliding pair (namely, $a_u = a_v$). By Lemma~\ref{lemma:multi-lb}, each trial in the second bunch succeeds in finding $(u,v)$ with probability $\Omega\left( \frac{S}{n} \right)$. Since we have $O((n/S)\log n)$ independent trials, with probability $1-n^{-\Omega(1)}$, the second bunch of trials succeeds in finding $(u,v)$.  
\end{itemize}
Combining two cases together shows that the algorithm finds a colliding pair with probability $1-n^{-\Omega(1)}$. This completes the proof.
\end{proof}

\begin{reminder}{Theorem~\ref{theo:set-intersection-algo}}
For every complexity bounds $S(n), T(n): \mathbb{N}\to \mathbb{N}$ such that $S(n)^{1/2}\cdot T(n) \ge n^{1.5}$, there is a Monte Carlo algorithm solving \textsc{Set Intersection} in time $O(T(n)\cdot \polylog(n))$ and space $O(S(n)\cdot \polylog(n))$ with one-way access to random bits. The algorithm prints elements in no particular order, and the same element may be printed multiple times.
\end{reminder}

\begin{proof}
Denote $S:=S(n)$. Suppose $a,b\in [m]^n$ are the input arrays. Define $c\in [m]^{2n}$ as the concatenation of $a$ and $b$. The algorithm repeats the following process for $O\left( \frac{n\log^2 n}{S(n)} \right)$ times:
\begin{itemize}
    \item Sample $\mathbf{h}\sim \calH^{n,m,t,\kappa}$ and $S$ starting vertices $\mathbf{x}_1,\dots \mathbf{x}_S\sim [n]$. Run $\COLLIDE(\left\{\mathbf{x}_1,\dots, \mathbf{x}_S\right\})$ on $G_{c,\mathbf{h}}$. Print \emph{all} colliding pairs found by $\COLLIDE(\mathbf{x})$.
\end{itemize}
The total running time is $\Otilde(n^{3/2}/S^{1/2})$. We argue the correctness now. Suppose $c_p = c_q$ is a colliding pair. With probability $\Omega\left( \frac{S}{F_2(c)} \right) = \Omega\left( \frac{S}{n} \right)$, the algorithm finds $(p,q)$ in one trial. Since we have $O\left( \frac{n\log^2 n}{S(n)} \right)$ independent trials, the probability that the algorithm misses $(p,q)$ is bounded by $\left( 1 - \Omega\left( \frac{S}{n} \right)\right)^{\frac{n\log^2 n}{S(n)}} \le n^{-\omega(1)}$. Union-bounding over all colliding pairs concludes the proof.
\end{proof}

\subsection{The Multi-Walks}

Towards proving Lemma~\ref{lemma:multi-lb}, we need to generalize the technical tools developed in Section~\ref{sec:cjww-standard} and \ref{sec:cjww-extend} to handle multiple starting vertices.

\paragraph*{The standard multi-walk.} Generalizing the idea of the standard walk (Algorithm~\ref{algo:standard-walk}), we consider the following standard multi-walk. In the multi-walk, we need to generate $k$ walk trees, one for each starting vertex $x_i$. By adding one more dimension (i.e., the $(t+1)$-th dimension) in the index, we connect the $k$ walk trees into a larger tree, which we call ``multi-walk tree''. The multi-walk returns a tensor $T:\mathbb{N}^{t+1}\to [n]\cup \{\star\}$, which is also denoted by $T=\Mwalk^\std(h, (x_i)_i)$.

\begin{algorithm2e}[H]
\caption{The Standard Multi-Walk}
\label{algo:standard-multi-walk}
\SetKwInput{KwInput}{Input}                % Set the Input
\SetKwInput{KwVariable}{Global Variables}              % set the Output
\LinesNumbered

\DontPrintSemicolon
  
    \KwInput{$n,m\ge 1$. $t$ hash functions $h_1,\dots, h_t:[m]\to [n]\cup \{0\}$. The array $(a_1,\dots, a_n)\in [m]^n$. $k$ starting vertices $x_1,\dots, x_k\in [n]$}

    % Set Function Names
    \SetKwFunction{FMain}{Main}
    \SetKwFunction{FWalk}{stdwalk}
    \SetKw{Break}{break}
    
    \KwVariable{
        \; ~~~~ A tensor $T:\mathbb{N}^{t}\times [k]\to [n]\cup \{\star\}$, initialized with $\star$'s.
        \; ~~~~ A set $D\subseteq [m]$, initialized with $\emptyset$.
    }
    
    \SetKwProg{Fn}{Program}{:}{\KwRet}
    \Fn{\FMain}{
        $\ell \gets (0, 0, \dots, 0) \in \mathbb{N}^{t+1}$ \;
        \For{$i = 1,\dots, k$}{
            $\ell_{t+1} \gets \ell_{t+1} + 1$ \tcp*{$\ell_{t+1} = i$ indicates the $i$-th walk tree}
            $T(\ell) \gets x_i$ \;
            $\stdwalk(t, x, \ell)$ \tcp*{The $\stdwalk$ subroutine in Algorithm~\ref{algo:standard-walk}}
        }
        \KwRet $T$\;
    }
\end{algorithm2e}

We naturally generalize the definition of index to a definition of multi-index.

\begin{definition}\label{def:index}
We use the term \emph{multi-index} to refer to $(t+1)$-dimensional integer vectors $\ell = (\ell_1,\dots, \ell_{t+1})$. For two indices $\ell^1 \ne \ell^2$, let $i\in [t]$ be the \emph{largest} integer such that $\ell^1_i \ne \ell^2_i$. Then we say $\ell^1 < \ell^2$ if $\ell^1_i < \ell^2_i$. The width of a multi-index $\ell$ is defined as $\width(\ell) := \max_{1\le i\le t} \{ \ell_i \}$ (we do NOT consider $\ell_{t+1}$ here). Call a multi-index $\ell$ $\tau$-bounded if $\width(\ell) \le \tau$. The level of a multi-index is $\level(\ell) := \max\{q : \forall i< q, \ell_i = 0\}$. For every $0\le i\le t+1$, we use $\ell_{<i}, \ell_{>i}$ to denote the length-$(i-1)$ prefix and length-$(t+1-i)$ suffix of $\ell$, respectively.
\end{definition}

The definition of ``walk tree'' naturally generalizes to a definition of multi-walk tree. For a set $S$ of multi-indices, let $P(S)$ denote the union of paths from root to multi-indices in $S$. Let $\widetilde{P}(S)$ denote the subset of $P(S)$ that contains all indices of level less than $t+1$.

\paragraph*{The extended multi-walk(s).} For every set $S = \{\ell^1,\dots, \ell^c\}$ of $\tau$-bounded multi-indices, we consider the following (randomized) $S$-extended multi-walk. Let $\mathbf{T}^S\sim \Mwalk^S(h, (x_i)_i)$ be the resulting tensor when running Algorithm~\ref{algo:extend-multi-walk} on $(h, (x_i)_i)$.

\begin{algorithm2e}[H]
\caption{The $S$-Extended Multi-Walk}
\label{algo:extend-multi-walk}
\SetKwInput{KwInput}{Input}                % Set the Input
\SetKwInput{KwVariable}{Global Variables}              % set the Output
\LinesNumbered

\DontPrintSemicolon
    \KwInput{
        \; ~~~~ $n,m, t\ge 1$. $t$ hash functions $h_1,\dots, h_t$. The input array $a_1,\dots, a_n$.
        \; ~~~~ $k$ starting vertex $x_1,\dots, x_k\in [n]$. 
        \; ~~~~ A set of indices $S \subseteq \mathbb{N}^{t+1}$.
    }
    
    % Set Function Names
    \SetKwFunction{FMain}{Main}
    \SetKwFunction{FWalk}{extwalk}
    \SetKw{Break}{break}
    
    \KwVariable{
        \; ~~~~ A constant $\tau \gets 5 \log n$
        \; ~~~~ A tensor $T:[0, \tau]^t\times [k]\to [n]\cup \{\star\}$, initialized with $\star$'s
        \; ~~~~ $t$ sets $D_1,D_2,\dots, D_t$, each initialized with $\emptyset$
    }
    
    \SetKwProg{Fn}{Program}{:}{\KwRet}
    \Fn{\FMain}{
        $\ell \gets (0, 0, \dots, 0) \in \mathbb{N}^{t+1}$ \;
        \For{$i = 1,\dots, k$}{
            $\ell_{t+1} \gets \ell_{t+1} + 1$ \;
            $T(\ell) \gets x_i$ \;
            $\extwalk(t, x, \ell)$ \tcp*{The $\extwalk$ subroutine in Algorithm~\ref{algo:extended-walk}}
        }
        \KwRet $T$\;
    }
\end{algorithm2e}

\paragraph*{Comparing $\mathbf{T}^S$ with $T$.} Fix $h$ and $(x_i)$. Let us compare $T^S\in \supp(\Mwalk^S(h, (x_i)_i))$ with $T=\Mwalk^\std(h, (x_i)_i)$. Note that both $\Mwalk^S$ and $\Mwalk^\std$ compute the multi-walk tree by computing $k$ walk trees. For each $i\in [k]$, consider the $i$-th walk tree (i.e., the walk that starts at $x_i$). Inside the $i$-th walk tree, the extended walk does not deviate from the standard walk until one the following events happens.
\begin{enumerate}
    \item A collision is found. Namely, there is $\ell^1 < \ell^2$, $\ell^2_{t+1} = i$ such that $T(\ell^1) \simeq T(\ell^2)$. In this case, we say the deviation happens due to {a collision}.
    \item The standard walk makes $\tau$-consecutive level-$d$ moves inside a function call $\stdwalk(d, *, *)$.  In this case, we say the deviation happens due to a \emph{long hike}.
\end{enumerate}

\subsection{Technical Preparations}

This section proves useful facts about the standard and extended multi-walks. Lemmas in this section are proved by properly extending the ideas developed in Section~\ref{sec:cjww-extend}. 

\paragraph*{Analyzing the extended multi-walk.} To begin with, the following lemma is the analog of Lemma~\ref{lemma:extend-is-random} in the multi-walk case.

\begin{lemma}\label{lemma:multi-extend-is-random}
Let $S = \{ \ell^1, \dots, \ell^c \}$ be $c$ $\tau$-bounded multi-indices. Suppose $\mathbf{h}_1,\dots, \mathbf{h}_t$ are $c\tau$-wise independent and $c$ vertices $\mathbf{x}_1,\dots, \mathbf{x}_c\sim [n]$ are uniformly chosen. Let $\mathbf{T}^S\sim \Mwalk^S(\mathbf{h}, (\mathbf{x}_i))$. For every $u_1,u_2,\dots, u_c \in [n]$, we have
\[
\Pr_{\mathbf{h}, \mathbf{x}, \mathbf{T}^S} \left[ \mathbf{T}^S(\ell^i) = u_i, \forall i\in [1,c] \right] = \frac{2^{-|\widetilde{P}(S)|}}{n^{c}}.
\]
% Here, $\widetilde{P}(S)$ denotes the set of multi-indices in $P(S)$ with level less than $t+1$.
\end{lemma}

\begin{proofsketch}
The proof is identical to that of Lemma~\ref{lemma:extend-is-random}. Namely, the proof is by downwards induction on $d\in [t+1]$. For each $d\in [t+1]$, having observed $\mathbf{T}(\ell)$ for every $\ell\in P(S)$ of level larger than $d$, $\mathbf{h}_1,\dots, \mathbf{h}_d$ are still uniformly random. The term $2^{-|\widetilde{P}(S)|}$ accounts for the fact that we may observe $\mathbf{h}_d(y) = 0$ when tracing the paths in $P(S)$.
\end{proofsketch}

\paragraph*{Coupling.} Fix a pair of vertices $u,v\in [n], u\ne v$ in $G_{a,\mathbf{h}}$. We will prove Lemma~\ref{lemma:multi-lb} by coupling the standard multi-walk with a family of extended multi-walks. First, we define the following collection of good events. For every pair of $\tau$-bounded multi-indices $\ell^1 \ne \ell^2$, let $\good^{\ell^1,\ell^2}(h, (x_i)_{i\in [k]})$ be the following event about $T=\Mwalk^{\std}(h, (x_i)_i)$.

\begin{enumerate}
    \item $T(\ell^1) = u$ and $T(\ell^2) = v$.
    \item For every $\ell^3 < \ell^1$, $T(\ell^3) \ne u$.
    \item For every $\ell^4 < \ell^2$, $T(\ell^4) \ne v$.
\end{enumerate}
Having imposed Condition 2 and 3, events $\{ \good^{\ell^1,\ell^2}\}_{\ell^1,\ell^2}$ are mutually disjoint. Therefore, we have
\begin{align}
\Pr_{\mathbf{h},(\mathbf{x}_i)}[u,v\in \Out_{a, \mathbf{h}}(\{ \mathbf{x}_i \})] \ge \sum_{\ell^1\ne \ell^2} \Pr_{\mathbf{h},\mathbf{x}_i}[\good^{\ell^1,\ell^2}(\mathbf{h}, (\mathbf{x}_i))]. \label{eq:multi-goal}
\end{align}

The following lemma lower-bounds $\Pr_{\mathbf{h},\mathbf{x}_i}[\good^{\ell^1,\ell^2}(\mathbf{h}, (\mathbf{x}_i))]$ for every pair $(\ell^1, \ell^2)$. It is proved by extending ideas behind Lemma~\ref{lemma:exists-surgical}-\ref{lemma:extend-lb-standard}.

\begin{lemma}\label{lemma:multi-surgical-upperbound}
Fix $h_1,\dots, h_t$, $(x_1,\dots, x_k)$. Let $S = \{\ell^1,\ell^2\}$. We have
\begin{align}
&~~~~ \mathbbm{1}\{\good^{\ell^1,\ell^2}(h, (x_i)_i)\} \notag \\
&\ge \Pr_{\mathbf{T}^S\sim \Mwalk^S(h, (x_i)_i)} \left[ \mathbf{T}^S(\ell^1) = u \land \mathbf{T}^S(\ell^2) = v \right] - \notag \\
&~~~~ \sum_{\substack{r^1 < r^2 < \ell^1: r^2_{t+1} = \ell^1_{t+1} \\ A = \{r^1, r^2\}}} \Pr_{\mathbf{T}^{S\cup A}} \left[ \left(  \mathbf{T}^{S\cup A}(\ell^1) = u \land \mathbf{T}^{S\cup A}(\ell^2) = v \right) \land (\mathbf{T}^{S\cup A}(r^1) \simeq \mathbf{T}^{S\cup A}(r^2)) \right] - \notag \\
&~~~~ \sum_{\substack{r^1 < r^2 < \ell^2: r^2_{t+1} = \ell^2_{t+1} \\ A = \{r^1, r^2\}}} \Pr_{\mathbf{T}^{S\cup A}} \left[ \left(  \mathbf{T}^{S\cup A}(\ell^1) = u \land \mathbf{T}^{S\cup A}(\ell^2) = v \right) \land (\mathbf{T}^{S\cup A}(r^1) \simeq \mathbf{T}^{S\cup A}(r^2)) \right] - \notag \\
&~~~~ \sum_{\substack{r < \max(\ell^1, \ell^2) \\ A = \{r\}} } \Pr_{\mathbf{T}^{S\cup A}} \left[ \left(  \mathbf{T}^{S\cup A}(\ell^1) = u \land \mathbf{T}^{S\cup A}(\ell^2) = v \right) \land \mathbf{T}^{S\cup A}(r) \in \{u, v\} \right] - \notag \\
&~~~~ \sum_{\substack{r < \max(\ell^1, \ell^2): \width(r) = \tau \\ A = \{r\}} } \Pr_{\mathbf{T}^{S\cup A}} \left[ \left(  \mathbf{T}^{S\cup A}(\ell^1) = u \land \mathbf{T}^{S\cup A}(\ell^2) = v \right) \land \mathbf{T}^{S\cup A}(r)\ne \star \right]. \label{eq:surgical-lowerbound}
\end{align}
\end{lemma}

\begin{proof}

The right hand side is always bounded by $1$. Therefore, if $\good^{\ell^1,\ell^2}(h, (x_i)_i) = 1$, there is nothing to prove. Now suppose $\good^{\ell^1,\ell^2}(h, (x_i)_i) = 0$. Consider a tensor $T^* \in \supp(\Mwalk^{S}(h, (x_i)_i)$ such that $T^*(\ell^1) = u, T^*(\ell^2) = v$.

\paragraph*{Defining refutations.} We claim the following.

\begin{claim}\label{claim:exists-refute}
Suppose $\good^{\ell^1,\ell^2}(h, (x_i)_i) = 0$. Consider a tensor $T^* \in \supp(\Mwalk^{S}(h, (x_i)_i)$ such that $T^*(\ell^1) = u, T^*(\ell^2) = v$. At least one of the following is true for $T^*$.
\begin{enumerate}
    \item There is $\{r^1, r^2\}$ such that $r^1 < r^2 < \ell^1$, $r^2_{t+1} = \ell^1_{t+1}$ and $T^*(r^1) \simeq T^*(r^2)$.
    \item There is $\{r^1, r^2\}$ such that $r^1 < r^2 < \ell^2$, $r^2_{t+1} = \ell^2_{t+1}$ and $T^*(r^1) \simeq T^*(r^2)$.
    \item There is $\{r\}$ such that $r < \max(\ell^1, \ell^2)$, $\width(r) = \tau$ and $T^*(r) = T(r) \ne \star$.
    \item There is $\{r\}$ such that $r < \max(\ell^1, \ell^2)$ and $T^*(r) =  T(r) \in \{u, v\}$.
\end{enumerate}
\end{claim}

\begin{subproof}
Consider $T = \Mwalk^{\std}(h, (x_i)_i)$. Suppose $(T(\ell^1), T(\ell^2))\ne (u,v)$. Since we have $(T^*(\ell^1),T^*(\ell^2)) = (u, v)$, this implies that the $S$-extended multi-walk must have deviated from the standard multi-walk before reaching $\ell^1$ or $\ell^2$. The deviation happens due to either a collision or a long hike. In either case, at least one of Condition $1$-$3$ happens.
% Inspecting the execution of Algorithm~\ref{algo:extend-multi-walk}, we note that the deviation happens only when at least one of Condition $1$-$3$ happens.

Now suppose $(T(\ell^1), T(\ell^2))= (u,v)$ but $\good^{\ell^1,\ell^2}(u,v) = 0$. We show that Condition $3$ or $4$ holds for $T^*$. In fact, $\good^{\ell^1,\ell^2}(u,v) = 0$ implies that there is $r<\max(\ell^1,\ell^2)$ such that $T(r) \in \{u, v\}$. Now, if we have $T^*(r) \in \{u, v\}$, Condition $4$ holds for $T^*$ and we are done. Otherwise we know the $S$-extended walk deviates from $\Mwalk^{\std}(x,(x_i))$ before reaching $r$. Since $T(r)\ne \star$, the deviation is not due to a collision. Then there must be one $r' < r$ with $\width(r') = \tau$ such that $T^*(r') = T(r') \ne \star$. In this case, Condition $3$ holds for $T^*$.
\end{subproof}

Let $A$ be a set of at most two multi-indices. We call $A$ a refutation for $T^*$, if it satisfies one of the $4$ conditions in Claim~\ref{claim:exists-refute}. 

\paragraph*{Surgical refutations.} We extend the idea of ``surgical refutations'' to multi-walks. Before we continue, we associate with $T^* \in \supp(\Mwalk^{S}(h, (x_i)_i)$ the following information.
\begin{enumerate}
    \item For every $\ell$ of level less than $t+1$, define
    \[
    \probe(\ell) = 
    \begin{cases}
    a_x, & \text{ If $T^*(\ell)$ was set by accessing $h_{\level(\ell)}(x)$. } \\
    \star, & \text{Otherwise (namely, $\Mwalk^S(h, (x_i)_i)$ has never tried to write $T^*(\ell)$).}
    \end{cases}
    \]
    \item For every $\ell$ with $\probe(\ell) \ne \star$, we define $\pre(\ell) := \sup\{ r : r < \ell, T^*(r) \ne \star\}$. Note that we have $\probe(\ell) = a_{T^*(\pre(\ell))}$.
\end{enumerate}

Suppose $A$ is a refutation for $T^*$. We further call $A$ a surgical refutation for $T^*$, if all of the following hold.
\begin{itemize}
    \item $A$ is a refutation satisfying at least one condition from Claim~\ref{claim:exists-refute}.
    \item For every $i\le t$ and every two distinct level-$i$ indices $r^1, r^2\in \widetilde{P}(A) \cup \widetilde{P}(S)$ such that $\{r^1,r^2\}\not\subseteq P(S)$, it holds that $\probe(r^1) \ne \probe(r^2)$.
%     \item For every $i\le t$, every two level-$i$ indices $r^1\in \widetilde{P}(A) \setminus \widetilde{P}(S)$ and $r^2\in \widetilde{P}(S)$, it holds that $\probe(r^1) \ne \probe(r^2)$.
\end{itemize}

We have the following claim.

\begin{claim}\label{claim:exists-surgical}
Let $T^* \in \supp(\Mwalk^{S}(h, (x_i)_i)$ be a tensor with $T^*(\ell^1) = u, T^*(\ell^2) = v$. There is a surgical refutation for $T^*$.
\end{claim}

\begin{subproof}[Proofsketch]
By Claim~\ref{claim:exists-refute}, we know there exists at least one refutation for $T^*$. We argue the existence of surgical refutation below.

We first argue that if there is no refutation satisfying Condition $1$ or $2$, then every refutation is surgical. Take $\{r\}$ to be a refutation that satisfies Condition $3$ or $4$. Suppose $\probe(r^1) = \probe(r^2)$ for some $r^1,r^2\in P(r)\cup P(S)$ such that $\probe(r^1) = \probe(r^2)$, $r^1\ne r^2$ and $\{r^1,r^2\}\not\subseteq P(S)$. We argue that at least one of $r^1,r^2$ lies in $P(S)$. If it is not the case, then we know $P(r)$ contains two entries that have the same $\probe$ value. This is contradictory to $T(r)\ne \star$\footnote{Recall that the standard multi-walk halts a walk tree immediately after finding a collision. Therefore, if $P(r)$ contains a collision, the standard does not have a chance to set $T(r)$ to be non-star.}. However, if one of $r^1,r^2$ lies in $P(S)$, then we can find a refutation satisfying Condition $1$ or $2$ by taking $\pre(r^1)$ and $\pre(r^2)$ and noting that $\pre(\ell)_{t+1}=\ell_{t+1}$ for every multi-index $\ell$. This is again a contradiction. So, such $r^1, r^2$ do not exist and $\{r\}$ is itself surgical.

% Note that $\probe(r^1)=\probe(r^2)$ implies that $T^*(\pre(r^1))\simeq T^*(\pre(r^2))$. Also observe that $\pre(r)_{t+1} = r_{t+1}$ holds for every $r\in \widetilde{P}(S)$. Therefore, if there is no refutation satisfying Condition $1$ or $2$, then every refutation is surgical. 
Now we consider the case that there exists a refutation satisfying Condition $1$ or $2$, it is easy to see that we can use the same argument as Lemma~\ref{lemma:exists-surgical} to find a surgical refutation satisfying Condition $1$ or $2$.
\end{subproof}

\paragraph*{Wrapping-up.} We are ready to conclude the proof now. For every $T^* \in \supp(\Mwalk^{S}(h, (x_i)_i)$ such that $T^*(\ell^1) = u, T^*(\ell^2) = v$, by Claim~\ref{claim:exists-surgical}, there is a surgical refutation for $T^*$. Let $A$ be the surgical refutation. Using the same argument as the proof for Lemma~\ref{lemma:surgical-upper-bound}, we have
\[
\Pr_{\mathbf{T}^S\sim \Mwalk^S(h, (x_i)_i)} \left[ \mathbf{T}^S = T^* \right] = \Pr_{\mathbf{T}^{S\cup A}\sim \Mwalk^{S\cup A}} \left[ \mathbf{T}^{S\cup A} = T^* \right].
\]
Summing up all relevant $T^*\in\supp(\Mwalk^{S}(h, (x_i)_i)$ completes the proof of Lemma~\ref{lemma:multi-surgical-upperbound}.
\end{proof}

\subsection{Proof of Lemma~\ref{lemma:multi-lb}}\label{sec:proof-wrap-up}

Having established Lemma~\ref{lemma:multi-extend-is-random} and \ref{lemma:multi-surgical-upperbound}, we are ready to prove Lemma~\ref{lemma:multi-lb}.

\begin{reminder}{Lemma~\ref{lemma:multi-lb}}
    Let $t = \frac{1}{2}\log (n/k)$ and $\kappa = 20\log n$. Sample $\mathbf{h}\sim \calH^{n,m,t,\kappa}$ and $\mathbf{x}_1,\dots, \mathbf{x}_k\sim [n]$. Then for every $u, v\in [n]$, $u\ne v$, it holds that:
    \[
    \Pr_{\mathbf{h}, (\mathbf{x}_i)}[u, v\in \Out_{a,\mathbf{h}}(\{\mathbf{x}_i\})] \ge \Omega\left( \frac{k}{F_2(a)} \right).
    \]
\end{reminder}

\begin{proof}
We first observe that $\Omega\left(  \frac{k^2}{n^2} \right)$ is a trivial lower bound when $k\ge 2$. To see this, note that with probability $\Omega\left( \frac{k}{n} \right)$ we have $u\in \{\mathbf{x}_1,\dots, \mathbf{x}_{k/2}\}$. Also, with probability $\Omega\left( \frac{k}{n} \right)$ we have $v\in \{ \mathbf{x}_{k/2+1},\dots, \mathbf{x}_{k} \}$. Therefore, with probability $\Omega\left( \frac{k^2}{n^2} \right)$ we have $u,v\in \{\mathbf{x}_i\}$.

For every $\gamma \le t$, we compare $\calH^{n,m,\gamma, \kappa}$ and $\calH^{n,m,t, \kappa}$. Note that having more layers of hash would not degrade the connectivity of the graph $G_{a,h}$. Therefore, it is easy to see that
\[
\Pr_{\mathbf{h}\sim \calH^{n,m,t,\kappa}, (\mathbf{x}_i)}[u, v\in \Out_{a,\mathbf{h}}(\{\mathbf{x}_i\})] \ge \Pr_{\mathbf{h}\sim \calH^{n,m,\gamma,\kappa}, (\mathbf{x}_i)}[u, v\in \Out_{a,\mathbf{h}}(\{\mathbf{x}_i\})].
\]
Let $C \ge 1$ be a large constant to be specified later. Set $\gamma = \frac{1}{2}\log(n^2/(kF_2(a))) - C$. Since $\Omega\left( \frac{k^2}{n^2} \right)$ is a trivial lower bound, we assume $\frac{k^2}{n^2}\le c \frac{k}{F_2(a)}$ holds for a sufficiently small $c > 0$ so that $\gamma \ge \frac{1}{2}\log(1/c) - C \ge 5$. In the following, we prove 
\begin{align}
\Pr_{\mathbf{h}\sim \calH^{n,m,\gamma,\kappa},(\mathbf{x}_i)}[u, v\in \Out_{a,\mathbf{h}}(\{\mathbf{x}_i\})] \ge \Omega\left( \frac{k}{F_2(a)} \right). \label{eq:multi-lb-goal}
\end{align}
From now on, we always use $\mathbf{h}$ to denote $\mathbf{h}\sim \calH^{n,m,\gamma, \kappa}$. We start with \eqref{eq:multi-goal}, which asserts that
\[
\Pr_{\mathbf{h},(\mathbf{x}_i)}[u,v\in \Out_{a, \mathbf{h}}(\{ \mathbf{x}_i \})] \ge \sum_{\ell^1\ne \ell^2} \Pr_{\mathbf{h},\mathbf{x}_i}[\good^{\ell^1,\ell^2}(\mathbf{h}, \mathbf{x}_i)].
\]
We use Lemma~\ref{lemma:multi-surgical-upperbound} and Lemma~\ref{lemma:multi-extend-is-random}. Consider Equation~\eqref{eq:surgical-lowerbound} given by Lemma~\ref{lemma:multi-surgical-upperbound}. We take expectation over $\mathbf{h},(\mathbf{x}_i)$ and sum up all pairs $(\ell^1,\ell^2)$ on both sides of \eqref{eq:surgical-lowerbound}. After that, we consider the right hand side of \eqref{eq:surgical-lowerbound}. We calculate the contribution from each line separately. The first line is
\[
E_1 := \sum_{\ell^1\ne \ell^2} \Pr_{\mathbf{T}^S\sim \Mwalk^S(h, (x_i)_i)} \left[ \mathbf{T}^S(\ell^1) = u \land \mathbf{T}^S(\ell^2) = v \right] \ge \sum_{\ell^1\ne\ell^2} \frac{2^{-|\widetilde{P}(\ell^1,\ell^2)|}}{n^2}.
\]
The second line is
\[
\begin{aligned}
E_2 &:= \sum_{\substack{\ell^1\ne \ell^2 \\ r^1 < r^2 < \ell^1: r^2_{\gamma+1} = \ell^1_{\gamma+1} \\ A = \{r^1, r^2\}}} \Pr_{\mathbf{T}^{S\cup A}} \left[ \left(  \mathbf{T}^S(\ell^1) = u \land \mathbf{T}^S(\ell^2) = v \right) \land (\mathbf{T}^{S\cup A}(r^1) \simeq \mathbf{T}^{S\cup A}(r^2)) \right] \\
& \le \sum_{\substack{\ell^1\ne \ell^2 \\ r^1<r^2 < \ell^1:r^2_{\gamma+1}=\ell^{1}_{\gamma+1}}} \frac{2^{-|\widetilde{P}(\{\ell^1,\ell^2,r^1,r^2\})|} \cdot F_2(a) }{n^4} + \sum_{\substack{\ell^1\ne \ell^2\\ r<\ell^1,r_{\gamma+1}=\ell^1_{\gamma+1}}} \frac{2^{-|\widetilde{P}(\{\ell^1,\ell^2,r\})|}\cdot F_{\infty}(a) }{n^3}.
\end{aligned}
\]
Here, the second term appears because it might be possible that $\ell^2\in \{r^1,r^2\}$. In this case, conditioning on $\mathbf{T}^S(\ell^1) = u, \mathbf{T}^S(\ell^2)=v$ (which happens with probability $\frac{1}{n^2}$), the probability that $\mathbf{T}^S(r) \simeq \mathbf{T}^S(\ell^2)$ is bounded by $\frac{F_{\infty}(a)}{n}$. Similarly, the third line is
\[
E_3 \le \sum_{\substack{\ell^1\ne \ell^2\\ r^1<r^2 < \ell^2:r^2_{\gamma+1}=\ell^{2}_{\gamma+1}}} \frac{2^{-|\widetilde{P}(\{\ell^1,\ell^2,r^1,r^2\})|} \cdot F_2(a) }{n^4} + \sum_{\substack{\ell^1\ne \ell^2\\ r<\ell^2,r_{\gamma+1}=\ell^2_{\gamma+1}}} \frac{2^{-|\widetilde{P}(\{\ell^1,\ell^2,r\})|}\cdot F_{\infty}(a) }{n^3}.
\]
The fourth line is
\[
\begin{aligned}
E_4 &:= \sum_{\substack{\ell^1\ne \ell^2 \\ r < \max(\ell^1, \ell^2) \\ A = \{r\}} } \Pr_{\mathbf{T}^{S\cup A}} \left[ \left(  \mathbf{T}^S(\ell^1) = u \land \mathbf{T}^S(\ell^2) = v \right) \land \mathbf{T}^{S\cup A}(r) \in \{u, v\} \right] \\
&\le \sum_{\substack{\ell^1\ne \ell^2\\ r < \max(\ell^1,\ell^2)}} \frac{2^{-|\widetilde{P}(\{\ell^1,\ell^2,r\})|}\cdot 2}{n^3}.
\end{aligned}
\]
Finally, the last line is
\[
\begin{aligned}
E_5 &:= \sum_{\substack{\ell^1\ne \ell^2 \\ r < \max(\ell^1, \ell^2): \width(r) = \tau \\ A = \{r\}} } \Pr_{\mathbf{T}^{S\cup A}} \left[ \left(  \mathbf{T}^S(\ell^1) = u \land \mathbf{T}^S(\ell^2) = v \right) \land \mathbf{T}^{S\cup A}(r)\ne \star \right] \\
&\le  \sum_{\substack{\ell^1\ne \ell^2\\ r < \max(\ell^1,\ell^2) : \width(r) = \tau}} \frac{2^{-|\widetilde{P}(\{\ell^1,\ell^2,r\})|}}{n^2}.
\end{aligned}
\]
By Lemma~\ref{lemma:multi-surgical-upperbound}, we conclude that
\[
\sum_{\ell^1\ne \ell^2} \Pr_{\mathbf{h},\mathbf{x}_i}[\good^{\ell^1,\ell^2}(\mathbf{h}, \mathbf{x}_i)] \ge E_1 - E_2 - E_3 - E_4 - E_5.
\]
Now we bound $E_1,\dots, E_5$. First,
\[
\begin{aligned}
E_1 
& \ge \sum_{\ell^1\ne\ell^2} \frac{2^{-|\widetilde{P}(\ell^1,\ell^2)|}}{n^2} \\
& = \sum_{\ell^1,\ell^2} \frac{2^{-|\widetilde{P}(\ell^1,\ell^2)|}}{n^2} - \sum_{\ell} \frac{2^{-|\widetilde{P}(\ell^1)|}}{n^2} \\
& \ge \frac{k^2}{n^2} (2 - 2^{-\tau} )^{2\gamma} - \frac{k}{n^2} (2-2^{-\tau})^\gamma & \text{(Similar to \eqref{eq:cjww-first-trick})}\\
& \ge \frac{k^2}{2 n^2} \cdot 2^{2\gamma} \\
& \ge \Omega\left( \frac{k}{2^{2C} F_2(a)} \right).
\end{aligned}
\]
Here, the big-$\Omega$ hides an absolute constant independent of $n$ and $C$. The second-to-last inequality is valid as long as $\gamma \ge 5$ and $\gamma\ll 2^{\tau}$. Second,
\[
\begin{aligned}
E_2 &\le \sum_{\substack{\ell^1, \ell^2 \\ r^1, r^2 : r^2_{\gamma+1} = \ell^{1}_{\gamma+1}}} \frac{2^{-|\widetilde{P}(\{\ell^1,\ell^2,r^1,r^2\})|} \cdot F_2(a) }{n^4} + \sum_{\substack{\ell^1, \ell^2\\ r: r_{\gamma+1} = \ell^1_{\gamma+1}}} \frac{2^{-|\widetilde{P}(\{\ell^1,\ell^2,r\})|}\cdot F_{\infty}(a) }{n^3} \\
&\le \frac{24 k^3 2^{4(\gamma+1)} F_2(a)}{n^4} + \frac{6\cdot k^2 2^{3(\gamma+1)} F_{\infty}(a) }{n^3}. &\text{(Corollary~\ref{corol:counting})} \\
&\le O\left( \frac{k}{F_2(a) \cdot 2^{3C}} \right). & \text{($F_{\infty}(a)\le F_2(a)^{1/2}$)}
\end{aligned}
\]
Again the big-$O$ hides a constant independent of $C$. A similar bound holds for $E_3$. Then,
\[
\begin{aligned}
E_4 \le \sum_{\substack{\ell^1, \ell^2, r} } \frac{2^{-|\widetilde{P}(\{\ell^1,\ell^2,r\})|}\cdot 2}{n^3} \le \frac{6 k^3 2^{3(\gamma+1)}}{n^3} \le O\left( \frac{k}{F_2(a) \cdot 2^{3C}} \right).
\end{aligned}
\]
For $E_5$, we enumerate the index $j\in [\gamma]$ such that $r_j = \tau$. Similar to \eqref{eq:cjww-third-trick}, it follows that
\[
E_5 \le \sum_{\substack{\ell^1, \ell^2, r : \width(r) = \tau}} \frac{2^{-|\widetilde{P}(\{\ell^1,\ell^2,r\})|}}{n^2} \le k^3 \cdot \frac{2^{-\tau}}{n^2}\cdot \gamma \cdot (\tau+1)^2 \cdot 6\cdot 2^{3(\gamma+1)} \le O\left( \frac{1}{n^3 \cdot 2^{3C}}  \right).
\]
Finally, choosing a large enough $C \ge 1$ ensures that
\[
E_1 - E_2 - E_3 - E_4 - E_5 \ge \Omega\left( \frac{k}{F_2(a)\cdot 2^{2C}} \right) \ge \Omega\left( \frac{k}{F_2(a)} \right),
\]
which completes the proof.
\end{proof}

\section{The Connecting Property of the Pseudorandom Hash}\label{sec:connecting}

This section, we prove Theorem~\ref{theo:connecting-intro}. The precise statement we will prove is the following.

\begin{theorem}\label{theo:connectingtheo}
   For every $c\ge 1$, there is a constant $D_c\ge 1$ satisfying the following. For all sufficiently large $n\ge 1$ and $m\ge n$, let $t=\frac{1}{2}\log(n) - D_c$ and $\kappa = 5c\log n$. Consider $\mathbf{h}\sim \calH^{n,m,t,\kappa}$. For every integer sequence $a\in [m]^n$ that contains distinct elements\footnote{Namely, $a_i\ne a_j$ holds for every $1\le i < j \le n$ and every $k$ vertices $1\le u_1 < \dots < u_k\le n$.}, it holds that
   \[
   \Pr_{\mathbf{h}\sim \calH^{n,m,t,\kappa},\mathbf{x}}[ u_1,\dots, u_c \in \Out_{a,\mathbf{h}}(\mathbf{x})] \ge \Omega_c(n^{-c/2}).
   \]
   % Let $\mathbf{h}_1,\dots, \mathbf{h}_t$ be $(c+2)\tau$-wise independent and the starting vertex $\mathbf{x}\sim [n]$ be uniformly chosen. For any $c$ fixed integers $u_1,\ldots, u_c\in [n]$, it holds that $\Pr[u_i\in \Out_{a,\mathbf{h}}(\mathbf{x}), \forall i\in [c]] \ge \Omega_c(n^{-c/2})$.
\end{theorem}

\paragraph*{Proving Theorem~\ref{theo:connecting-intro}.} Note that the seed length for $\calH^{n,m,t,\kappa}$ is bounded by $O(t\kappa\log(n+m))$. Although every hash $h$ in $ \supp(\calH^{n,m,t,\kappa})$ has co-domain $[n]\cup \{-1\}$, we can slightly modify the hash family by replacing every ``$-1$'' value in $h$ with $1\in [n]$. This modification does not degrade the connectivity of the graph $G_{a,h}$. Therefore, Theorem~\ref{theo:connectingtheo} implies Theorem~\ref{theo:connecting-intro}. 

Before proving Theorem~\ref{theo:connectingtheo}, we need another corollary, which is for later use. The proof is in Appendix~\ref{app:technical}.

\begin{corollary}\label{corol:counting2}
For any fixed positive integer $c$ and $t$, it holds that
\[
\sum_{\substack{\ell^1\le\ldots\le\ell^c\\ \exists i\in [c-1], \ell^i=\ell^{i+1}}}2^{1-|P^{(t)}(\{\ell^1,\ldots,\ell^c\})|}\le 2^{(c-1)\cdot (t+1)}.
\]
\end{corollary}

Next, we prove Theorem~\ref{theo:connectingtheo}.

\begin{proof}
Since $a_i\ne a_j$ holds for every $1\le i<j\le n$, we have $F_2(a)=n$. For every $h_1,\dots, h_t$ and $x$, we consider the standard walk $T = \walk^{\std}(h,x)$. Since $T$ contains at most one pair of duplicate elements, we have
\[
\mathbbm{1}[u_1,\dots, u_c\in \Out_{a,h}(x)] \ge \frac{1}{2} \cdot \sum_{\ell^1 < \dots < \ell^c} \mathbbm{1}[\{T(\ell^1),\dots, T(\ell^c)\} = \{u_1,\dots, u_c\}].
\]
Taking an expectation over $\mathbf{h}\sim \calH^{n,m,t,\kappa}$ and $\mathbf{x}\sim [n]$, we have
\[
\begin{aligned}
& ~~~~ \Pr_{\mathbf{h_1},\dots, \mathbf{h}_t, \mathbf{x}}[u_1,\ldots,u_c\in \Out_{a,\mathbf{h}}(\mathbf{x})] 
 \notag \\ 
 & \ge \frac{1}{2}\sum_{\ell^1 < \ldots < \ell^c}  \Pr_{\mathbf{h}_1,\dots,\mathbf{h}_t, \mathbf{x}}\left[  \{ \mathbf{T}(\ell^1), \ldots, \mathbf{T}(\ell^c)\} = \{u_1, \ldots, u_c\} \right]. \label{eq:translate-to-tree-new}
\end{aligned}
\]

Fix one $\tau$-bounded tuple to the right hand side of \eqref{eq:translate-to-tree-new}. According to Lemma~\ref{lemma:extend-is-random} and \ref{lemma:extend-lb-standard}, it holds that
\[
\begin{aligned}
& ~~~~ \Pr_{\mathbf{h}, \mathbf{x}}\left[  \{ \mathbf{T}(\ell^1), \ldots, \mathbf{T}(\ell^c)\} = \{u_1, \ldots, u_c\} \land \left( \forall \omega^1 < \omega^2 < \ell^c, \mathbf{T}(\omega^1) \not\simeq \mathbf{T}(\omega^2) \right)\right]  \\
& \ge \Pr_{\mathbf{T} \sim \mathbf{T}^{\{\ell^1,\ldots,\ell^c\}}}[ \{ \mathbf{T}(\ell^1),\ldots, \mathbf{T}(\ell^c)\} = \{u_1, \ldots, u_c\} ] - \\
& ~~~~ \sum_{\omega^1 < \omega^2 < \ell^c} \Pr_{\mathbf{T} \sim \mathbf{T}^{\{\ell^1,\ldots,\ell^c,\omega^1,\omega^2\}}}[ \{ \mathbf{T}(\ell^1), \ldots,\mathbf{T}(\ell^c)\} = \{u_1,\ldots, u_c\} \land \mathbf{T}(\omega^1) \simeq \mathbf{T}(\omega^2) ] - \\
& ~~~~ \sum_{\omega^3 < \ell^c: \width(\omega^3) = \tau} \Pr_{\mathbf{T} \sim \mathbf{T}^{\{\ell^1,\ldots,\ell^c,  \omega^3\}}}[ \{ \mathbf{T}(\ell^1), \ldots,\mathbf{T}(\ell^c)\} = \{u_1,\ldots, u_c\} \land \mathbf{T}(\omega^3) \ne \star ] \\
& \ge \frac{2^{1-|P(\{\ell^1,\ldots,\ell^c\})|}}{n^c} - \sum_{\omega^1,\omega^2}\frac{2^{1-|P(\{\ell^1,\ldots,\ell^c,\omega^1,\omega^2\})|}}{n^{c+1}}-\sum_{\omega^3:\width(\omega^3) = \tau}\frac{2^{1-|P(\{\ell^1,\ldots,\ell^c,\omega^3\})|}}{n^c}.
\end{aligned}
\]
Similarly, we sum up all $\tau$-bounded tuples.
\begin{align}
 & ~~~~ \Pr_{\mathbf{h_1},\dots, \mathbf{h}_t, \mathbf{x}}[u_1,\ldots,u_c\in \Out_{a,\mathbf{h}}(\mathbf{x})] \notag\\ 
 & = \sum_{\ell^1 < \ldots < \ell^c}  \Pr_{\mathbf{h}_1,\dots,\mathbf{h}_t, \mathbf{x}}\left[  \{ \mathbf{T}(\ell^1), \ldots, \mathbf{T}(\ell^c)\} = \{u_1, \ldots, u_c\} \land \left( \forall \omega^1 < \omega^2 < \ell^c, \mathbf{T}(\omega^1) \not\simeq \mathbf{T}(\omega^2) \right)\right] \notag\\ 
 & \ge \sum_{\ell^1<\ldots<\ell^c}\frac{2^{1-|P(\{\ell^1,\ldots,\ell^c\})|}}{n^c} - \sum_{\ell^1,\ldots,\ell^c,\omega^1,\omega^2}\frac{2^{1-|P(\{\ell^1,\ldots,\ell^c,\omega^1,\omega^2\})|}}{n^{c+1}}-\sum_{\ell^1,\ldots,\ell^c,\omega^3:\width(\omega^3) = \tau}\frac{2^{1-|P(\{\ell^1,\ldots,\ell^c,\omega^3\})|}}{n^c}. \label{eq:bonus}
\end{align}
Now, we bound $3$ terms above separately. For the first one, it holds that
\[
\begin{aligned}
\sum_{\ell^1<\ldots<\ell^c}\frac{2^{1-|P(\{\ell^1,\ldots,\ell^c\})|}}{n^c}
\ge \frac{1}{n^c}\cdot \left(\frac{1}{c!}\cdot \sum_{\ell^1,\ldots,\ell^c}2^{1-|P(\{\ell^1,\ldots,\ell^c\})|}-\sum_{\substack{\ell^1\le\ldots\le\ell^c\\ \exists i\in [c-1], \ell^i=\ell^{i+1}}}2^{1-|P(\{\ell^1,\ldots,\ell^c\})|}\right).
\end{aligned}
\]
Based on Corollary~\ref{corol:counting2}, we have
\[
\begin{aligned}
\sum_{\substack{\ell^1\le\ldots\le\ell^c\\ \exists i\in [c-1], \ell^i=\ell^{i+1}}}2^{1-|P(\{\ell^1,\ldots,\ell^c\})|}\le 2^{(c-1)\cdot (t+1)}.
\end{aligned}
\]
Note that $|P(\{\ell^1,\ldots,\ell^c\})|\le \sum_{i=1}^{c}|P(\ell^i)|-(c-1)$. Therefore,
\[
\begin{aligned}
\sum_{\ell^1,\ldots,\ell^c}2^{1-|P(\{\ell^1,\ldots,\ell^c\})|} & \ge \left(\sum_{\ell}2^{1-|P(\ell)|}\right)^c \\
&= \left(\prod_{i=1}^{t}\sum_{\ell_i=0}^{\tau}2^{-\ell_i}\right)^c \\
&= (2-2^{-\tau})^{c\cdot t} \\
& \ge 2^{c \cdot t-1}.
\end{aligned}
\]
Here, the last inequality holds since $c\cdot t \ll 2^{\tau}$. For the first term in \eqref{eq:bonus}, we conclude that
\begin{align}
\sum_{\ell^1<\ldots<\ell^c}\frac{2^{1-|P(\{\ell^1,\ldots,\ell^c\})|}}{n^c} \ge \frac{2^{c\cdot t-1}}{c!\cdot n^c}-\frac{2^{(c-1)\cdot (t+1)}}{n^c} \ge \frac{2^{c\cdot t-2}}{c!\cdot n^c}. \label{bonus-first}
\end{align}
The last inequality holds because we assumed that $n$ is sufficiently large, which implies that $t$ is large enough so that $\frac{2^{ct-1}}{c!} > 2\cdot 2^{(c-1)(t+1)}$.

Using Corollary \ref{corol:counting}, we deduce that
\begin{align}
    \sum_{\ell^1,\ldots,\ell^c,\omega^1,\omega^2}\frac{2^{1-|P(\{\ell^1,\ldots,\ell^c,\omega^1,\omega^2\})|}}{n^{c+1}} \le \frac{(c+2)!\cdot 2^{(c+2)\cdot (t+1)}}{n^{c+1}}. \label{bonus-second}
\end{align}
For the third term, we use the same method as \eqref{eq:cjww-third-trick} to obtain a good upper bound. Since $\width(\omega^3) = \tau$, we enumerate $i\in [t]$ and $\omega^3$ such that $\omega^3_i = \tau$. Recall that for every $r\in \mathbb{N}^{t}$ and $i\in [t]$, we defined a $(t-1)$-dimensional index $r_{-i} := (r_1,\dots, r_{i-1},r_{i+1},\dots, r_t)$. Also, for every $t$-dimensional indices $\ell^1,\ldots,\ell^c$, it holds that
\[
\begin{aligned}
    |P^{(t)}(\{\ell^1,\ldots,\ell^c,\omega^3\})|\ge\tau+|P^{(t-1)}(\{\ell^1_{-i},\ldots,\ell^c_{-i},\omega^3_{-i}\})|.
\end{aligned}
\]
As a result, we have

\begin{align}
& ~~~~ \sum_{\ell^1,\ldots,\ell^c,\omega^3:\width(\omega^3) = \tau}\frac{2^{1-|P(\{\ell^1,\ldots,\ell^c,\omega^3\})|}}{n^c} \notag\\
& \le \frac{1}{n^c}\sum_{i=1}^{t}\sum_{\ell^1,\ldots,\ell^c,\omega^3:\omega^3_i = \tau}2^{1-|P^{(t)}(\{\ell^1,\ldots,\ell^c,\omega^3\})|} \notag\\
& \le \frac{1}{n^c}\sum_{i=1}^{t}\sum_{\ell^1,\ldots,\ell^c,\omega^3:\omega^3_i = \tau}2^{1-\tau-|P^{(t-1)}(\{\ell^1_{-i},\ldots,\ell^c_{-i}, \omega^3_{-i}\})|} \notag\\
& = \frac{1}{n^c}\sum_{i=1}^{t}2^{-\tau}\cdot (\tau+1)^c\cdot \sum_{\ell^1_{-i},\ldots,\ell^c_{-i},\omega^3_{-i}}2^{1-|P^{(t-1)}(\{\ell^1_{-i},\ldots,\ell^c_{-i},\omega^3_{-i}\})|} \notag\\
& \le \frac{1}{n^c}\sum_{i=1}^{t}2^{-\tau}\cdot (\tau+1)^c\cdot (c+1)!\cdot 2^{(c+1)\cdot t} \notag\\
& = \frac{(c+1)!\cdot t\cdot (\tau+1)^c}{n^c}\cdot 2^{(c+1)\cdot t-\tau}. \label{bonus-third}
\end{align}
Plugging \eqref{bonus-first}, \eqref{bonus-second} and \eqref{bonus-third} back in \eqref{eq:bonus}, we have
\[
\begin{aligned}
 & ~~~~ \Pr_{\mathbf{h}, \mathbf{x}}[u_1,\ldots,u_c\in \Out_{a,\mathbf{h}}(\mathbf{x})] \ge \frac{2^{c\cdot t-2}}{c!\cdot n^c}- \frac{(c+2)!\cdot 2^{(c+2)\cdot (t+1)}}{n^{c+1}}-\frac{(c+1)!\cdot t\cdot (\tau+1)^c}{n^c}\cdot 2^{(c+1)\cdot t-\tau}. \\
\end{aligned}
\]
Recall the threshold $\tau = 5\log n$ and $t =\frac{1}{2}\log(n)-D_c$. Now, we set $D_c = D$ to be a sufficiently large constant such that 
\[
\begin{aligned}
    \frac{1}{c!\cdot 2^{c\cdot D+2}}>\frac{(c+2)!}{2^{(c+2)(D-1)}} + \frac{1}{D}.
\end{aligned}
\]
Finally,
\[
\begin{aligned}
 \Pr_{\mathbf{h}, \mathbf{x}}[u_1,\ldots,u_c\in \Out_{a,\mathbf{h}}(\mathbf{x})] & \ge \frac{1}{c!\cdot 2^{c\cdot D+2}\cdot n^{c/2}}-\frac{(c+2)!}{2^{(c+2)(D-1)}\cdot n^{c/2}}- \frac{(c+1)!\cdot t \cdot (\tau+1)^c}{2^{(c+1)\cdot D}\cdot  n^{(c+9)/2}} \\
 & \ge \Omega(n^{-c/2}), \\
\end{aligned}
\]
which completes the proof.

\end{proof}

\section*{Acknowledgements}

We would like to thank Lijie Chen, Ce Jin, Ryan Williams and Hongxun Wu for insightful discussions about their work \cite{ChenJWW22}. In particular, we thank Ce Jin for pointing us to Dinur's lower bound \cite{DBLP:conf/eurocrypt/Dinur20-lowerbound}.

X. Lyu was supported by ONR DORECG award N00014-17-1-2127.

\nocite{DBLP:journals/joc/OorschotW99}

\bibliographystyle{alpha}
\bibliography{mybib}

\appendix

\section{The Proof of Lemma~\ref{theo:counting}}\label{app:technical}

In this section, we provide a proof of Lemma~\ref{theo:counting}. Recall its statement. 

\begin{reminder}{Lemma~\ref{theo:counting}}
For any fixed positive integer $c$ and $t$, let $S^t$ be the set of all $t$-dimensional indices. Denote $f(c,t)$ as $\sum_{S\subseteq S^t, |S|=c}2^{1-|P^{(t)}(S)|}$. Then, we have
\[
f(c,t) \le 2^{c\cdot t}.
\]
\end{reminder}

We first need to prove the following lemma, which is for later use.

\begin{lemma}\label{lemma:combinatorics}
    For every non-negative integer $r \ge 0$, it holds that
    
\[
\sum_{k=0}^{+\infty}2^{-k}\cdot \binom{k}{r}=2.
\]
\end{lemma}

\begin{proof}
    Consider tossing an unbiased coin for infinite times. The probability that we get at least $r+1$ heads is $1$. From another perspective, suppose the $(r+1)$-th head appears at the $(k+1)$-th round. This implies that we get a head in the $(k+1)$-th round and exactly $r$ heads in the first $k$ rounds. In this case, there are $\binom{k}{r}$ possible results for the first $k$ rounds. The probability of getting each possible result is $2^{-k-1}$. Therefore, we have
    \[
    \sum_{k=0}^{+\infty}2^{-k-1}\cdot \binom{k}{r}=1,
    \]
    which is equivalent to the lemma.
\end{proof}

Now, we can start the proof of Theorem \ref{theo:counting}. We finish it by induction on $c$.

\begin{proof}
    \item \underline{\textsl{Base Case:}} When $c=1$, it holds that
    \[
    \begin{aligned}
    f(1,t)&=\sum_{S\subseteq S^t, |S|=1}2^{1-|P^{(t)}(S)|} \\
    &=\sum_{\ell\in S^t}2^{1-|P^{(t)}(\ell)|} \\
    &=\sum_{\ell\in S^t}2^{-\sum_{i=1}^{t}\ell_i} \\
    &=\prod_{i=1}^{t}\sum_{\ell_i=0}^{\tau}2^{-\ell_i} \\
    &= (2-2^{-\tau})^t \\
    &\le 2^t,\\
    \end{aligned}
    \]
    which implies $f(c,t)\le 2^{c\cdot t}$ when $c=1$.
    
    \item \underline{\textsl{Inductive Step:}}
    Assume for any positive integer $c'<c$ and $t$, it holds $f(c',t)\le 2^{c'\cdot t}$. Now, let us prove that $f(c,t)\le 2^{c\cdot t}$ for any positive integer $t$.
    
    When $t=1$, suppose the $c$ indices in $S$ are $\ell^1<\ell^2<\ldots<\ell^c$. As $t=1$, we have $0\le \ell^1_1<\ell^2_1<\ldots<\ell^c_1\le \tau$. For every $\ell^c_1$ from $0$ to $\tau$, there are $\binom{\ell^c_1}{c-1}$ ways to choose $(\ell_1^i)_{i=1}^{c-1}$. Each of them contributes $2^{-\ell^c_1}$ to $f(c,1)$. By Lemma \ref{lemma:combinatorics}, we have
    \[
    \begin{aligned}
    f(c,1) &= \sum_{\ell^c_1=0}^{\tau}2^{-\ell^c_1}\cdot \binom{\ell^c_1}{c-1} \\
    &\le \sum_{\ell^c_1=0}^{+\infty}2^{-\ell^c_1}\cdot \binom{\ell^c_1}{c-1} \\
    &= 2 < 2^c.
    \end{aligned}
    \]
    
    When $t>1$, suppose the $c$ indices in $S$ are $\ell^1<\ell^2<\ldots<\ell^c$, which implies $\ell^1_t\le \ell^2_t\le\ldots\le\ell^c_t$. For $0\le j \le \tau$, denote $S_{j}:=\{\ell:\ell\in S \land \ell_t=j\}$. Also, let $S_{j}':=\{\omega:\exists\ell \in S_j, \forall i\in[t-1], \omega_i=\ell_i\}$, which is a subset of $S^{t-1}$. Obviously, $S_0, S_1, \ldots, S_\tau$ form a partition of $S$. Let $k\in\{0,\ldots,\tau\}$ be the largest integer such that $S_{k}\neq\emptyset$. According to Definition \ref{def:walk-tree}, it holds that 
    \[
    |P^{(t)}(S)|-1=k+\sum_{i: S_i \neq\emptyset}(|P^{(t-1)}(S_i')|-1).
    \]
    Therefore, for set $S$, its contribution to $f(c,t)$ is
    \[
    2^{1-|P^{(t)}(S)|}=2^{-k}\cdot\prod_{i: S_i \neq\emptyset}2^{1-|P^{(t-1)}(S_i')|}.
    \]
    Now, let us get an upper bound for $f(c,t)$ when $t>1$. 
    
    \[
    \begin{aligned}
    f(c,t) &= \sum_{S\subseteq S^t, |S|=c}2^{1-|P^{(t)}(S)|} \\
    &= \sum_{k=0}^{\tau}~~ \sum_{S\subseteq S^t, |S|=c, \ell^c_t=k}2^{1-|P^{(t)}(S)|} \\
    &= \sum_{k=0}^{\tau}~~ \sum_{S\subseteq S^t, |S|=c, \ell^c_t=k}2^{-k}\cdot\prod_{i: S_i \neq\emptyset}2^{1-|P^{(t-1)}(S_i')|} \\
    &= \sum_{k=0}^{\tau}2^{-k}\cdot ~~ \sum_{\substack{S_0',S_1',\ldots,S_{k}'\subseteq S^{t-1} \\ |S_0'|+|S_1'|+\ldots+|S_{k}'|=c \\ |S_{k}'|>0}}\prod_{i: S_i \neq\emptyset}2^{1-|P^{(t-1)}(S_i')|} \\
    &= \sum_{k=0}^{\tau}2^{-k}\cdot ~~ \left(\sum_{\substack{S_{k}'\subseteq S^{t-1}, |S_{k}'|=c}}2^{1-|P^{(t-1)}(S_k')|}+\sum_{\substack{S_0',S_1',\ldots,S_{k}'\subseteq S^{t-1} \\ |S_0'|+|S_1'|+\ldots+|S_{k}'|=c \\ c>|S_{k}'|>0}}\prod_{i: S_i \neq\emptyset}2^{1-|P^{(t-1)}(S_i')|}\right) \\
    &= \sum_{k=0}^{\tau}2^{-k}\cdot ~~ \left(f(c,t-1)+\sum_{\substack{n_0+n_1+\ldots+n_k=c \\ n_i\ge 0 \\ c>n_k>0}}~~\prod_{i: n_i>0}\sum_{\substack{S_i'\subseteq S^{t-1} \\ |S_i'|=n_i}}2^{1-|P^{(t-1)}(S_i')|}\right) \\
    &= \sum_{k=0}^{\tau}2^{-k}\cdot ~~ \left(f(c,t-1)+\sum_{\substack{n_0+n_1+\ldots+n_k=c \\ n_i\ge 0 \\ c>n_k>0}}~~\prod_{i: n_i>0}f(n_i,t-1)\right). \\
    \end{aligned}
    \]
    
    For non-negative integers $n_0, n_1, \ldots, n_k$, if $c>n_k>0$ holds when $\sum_{i=0}^{k}=c$, we have $n_i<c$ for all $0\le i \le k$. According to our assumption of the inductive step, we have $f(n_i,t-1)\le 2^{n_i\cdot (t-1)}$. Therefore,
    
    \[
    \begin{aligned}
    f(c,t) &= \sum_{k=0}^{\tau}2^{-k}\cdot  \left(f(c,t-1)+\sum_{\substack{n_0+n_1+\ldots+n_k=c \\ n_i\ge 0 \\ c>n_k>0}}~~\prod_{i: n_i>0}f(n_i,t-1)\right) \\
    &\le \sum_{k=0}^{\tau}2^{-k}\cdot  \left(f(c,t-1)+\sum_{\substack{n_0+n_1+\ldots+n_k=c \\ n_i\ge 0 \\ c>n_k>0}}~~\prod_{i: n_i>0}2^{n_i\cdot (t-1)}\right) \\
    &= \sum_{k=0}^{\tau}2^{-k}\cdot  \left(f(c,t-1)+\sum_{\substack{n_0+n_1+\ldots+n_k=c \\ n_i\ge 0 \\ c>n_k>0}}~~2^{c\cdot (t-1)}\right) \\
    &= \sum_{k=0}^{\tau}2^{-k}\cdot  \left(f(c,t-1)+\left(\binom{c+k-1}{c-1}-1\right)\cdot 2^{c\cdot (t-1)}\right) \\
    &\le \sum_{k=0}^{+\infty}2^{-k}\cdot  \left(f(c,t-1)+\left(\binom{c+k-1}{c-1}-1\right)\cdot 2^{c\cdot (t-1)}\right) \\
    &\le \sum_{k=0}^{+\infty}2^{-k}\cdot \left(f(c,t-1)-2^{c\cdot (t-1)}\right)+2^{c\cdot t-1}\cdot \sum_{k=0}^{+\infty}2^{-(c+k-1)}\binom{c+k-1}{c-1} \\
    &\le 2\cdot \left(f(c,t-1)-2^{c\cdot (t-1)}\right)+2^{c\cdot t}. \\
    \end{aligned}
    \]
    The last step is based on Lemma \ref{lemma:combinatorics}. In this way, we have
    \[
    \begin{aligned}
    f(c,t)-2^{c\cdot t} &\le 2\cdot \left(f(c,t-1)-2^{c\cdot (t-1)}\right) \\
    &\le 2^2\cdot \left(f(c,t-2)-2^{c\cdot (t-2)}\right) \\
    &\le \ldots \\
    &\le 2^{t-1}\cdot \left(f(c,1)-2^c\right) \\
    &\le 0,
    \end{aligned}
    \]
    which implies $f(c,t)\le 2^{c\cdot t}$. This completes the proof.
    
\end{proof}

Now, we prove Corollary~\ref{corol:counting}. Recall its statement as follows.

\begin{reminder}{Corollary~\ref{corol:counting}}
For any fixed positive integer $c$ and $t$, it holds that
\[
\sum_{\ell^1,\ldots,\ell^c\in S^t}2^{1-|P^{(t)}(\{\ell^1,\ldots,\ell^c\})|}\le c!\cdot 2^{c\cdot (t+1)}.
\]
\end{reminder}

\begin{proof}
It is obvious that

\begin{align}
\sum_{\ell^1,\ldots,\ell^c\in S^t}2^{1-|P^{(t)}(\{\ell^1,\ldots,\ell^c\})|}\le c!\cdot \sum_{\ell^1\le\ldots\le\ell^c\in S^t}2^{1-|P^{(t)}(\{\ell^1,\ldots,\ell^c\})|}.
\label{eq:counting}
\end{align}

For $1\le i < c$, let $d_i$ be an integer which is $0$ or $1$. If $\ell^i=\ell^{i+1}$, set $d_i$ as $0$. Otherwise, set $d_i$ as $1$. There are totally $2^{c-1}$ possibilities for sequence $(d_i)_{i=1}^{c-1}$. For each possibility, if there are exactly $k$ zeros in the sequence, its contribution to the right side of \eqref{eq:counting} is $c!\cdot f(c-k,t)$, which is no larger than $c!\cdot 2^{(c-k)\cdot t}\le c!\cdot 2^{c\cdot t}$. Therefore, it holds that
\[
\begin{aligned}
    \sum_{\ell^1,\ldots,\ell^c\in S^t}2^{1-|P^{(t)}(\{\ell^1,\ldots,\ell^c\})|} &\le c!\cdot \sum_{\ell^1\le\ldots\le\ell^c\in S^t}2^{1-|P^{(t)}(\{\ell^1,\ldots,\ell^c\})|}  \\
    & \le 2^{c-1}\cdot c!\cdot 2^{c\cdot t}  \\
    & \le c!\cdot 2^{c\cdot (t+1)}, 
\end{aligned}
\]
which finishes the proof.

\end{proof}

Here we give the proof of Corollary~\ref{corol:counting2}. Recall its statement as follows.

\begin{reminder}{Corollary~\ref{corol:counting2}}
For any fixed positive integer $c$ and $t$, it holds that
\[
\sum_{\substack{\ell^1\le\ldots\le\ell^c\\ \exists i\in [c-1], \ell^i=\ell^{i+1}}}2^{1-|P^{(t)}(\{\ell^1,\ldots,\ell^c\})|}\le 2^{(c-1)\cdot (t+1)}.
\]
\end{reminder}
\begin{proof}
We define sequence $(d_i)_{i=1}^{c-1}$ the same way as that in the proof of Corollary~\ref{corol:counting}. There are totally $2^{c-1}-1$ possibilities, because there must be some $i\in [c-1]$ such that $d_i=0$. For each possibility, if there are exactly $k$ zeros in the sequence, its contribution is $f(c-k,t)$, which is no larger than $2^{(c-k)\cdot t}\le 2^{(c-1)\cdot t}$. Therefore, it holds that
\[
\begin{aligned}
    \sum_{\substack{\ell^1\le\ldots\le\ell^c\\ \exists i\in [c-1], \ell^i=\ell^{i+1}}}2^{1-|P^{(t)}(\{\ell^1,\ldots,\ell^c\})|} & \le (2^{c-1}-1)\cdot 2^{(c-1)\cdot t} \le 2^{(c-1)\cdot (t+1)},
\end{aligned}
\]
which finishes the proof.
\end{proof}

\section{$c$-Connecting is Asymptotically Optimal}\label{sec:connecting-ub}

In this section, we show that the probability lower bound in Definition~\ref{def:connecting} is asymptotically optimal by proving the following claim.

\begin{claim}
For every constant $c\ge 2$, there is a constant $C_c \ge 1$ satisfying the following. For all large enough $n,m\ge 1$ such that $m\ge n^2$, suppose $\calH$ is a distribution over hash functions $h:[m]\to [n]$. Then there is an injective mapping $a:[n]\to [m]$ and $c$ vertices $1 \le u_1 < \dots < u_c \le [n]$ such that 
\begin{align}
\Pr_{\mathbf{h}\sim \calH, \mathbf{x}\sim [n]}[u_i\in \Out_{a,\mathbf{h}}(\mathbf{x}), \forall i\in [c]] \le C_c\cdot n^{-c/2}. \label{eq:connecting-ub}
\end{align}
\end{claim}

\begin{proofsketch}
Let $\calH$ be an arbitrary distribution over hash functions $h:[m]\to [n]$. Let $\calA$ denote the uniform distribution over all mappings from $[n]$ to $[m]$. Consider sampling $\mathbf{a}\in \calA$ (Note that $\mathbf{a}$ is not necessarily injective). For every fixed $h\in \supp(\calH)$, $(h\circ \mathbf{a})$ is a random mapping from $[n]$ to $[n]$ where $\Pr_{\mathbf{a}}[(h\circ \mathbf{a})(x) = y] = \frac{|h^{-1}(y)|}{m}$. By a birthday-paradox style argument, we have
\[
\Ex_{\mathbf{a}\sim \calA,\mathbf{h}\sim \calH,\mathbf{x}\sim [n]} \left[ | \Out_{\mathbf{a},\mathbf{h}}(\mathbf{x})| \right] \le O(\sqrt{n}).
\]
In fact, with some extra effort, one can show that the distribution of $|\Out_{\mathbf{a},\mathbf{h}}(\mathbf{x})|$ decays to $0$ exponentially fast, and there is a constant $D_c\ge 1$ ($D_c$ depends on $c$) such that
\[
\Ex_{\mathbf{a}\sim \calA,\mathbf{h}\sim \calH,\mathbf{x}\sim [n]} \left[ \binom{|\Out_{\mathbf{a},\mathbf{h}}(\mathbf{x})|}{c} \right] \le D_c n^{c/2}.
\]
Let $A^I = \{a | a:[n]\to [m] \text{ is injective}\}$. By a simple union bound, we have $\Pr_{\mathbf{a}\sim \calA}[\mathbf{a}\in A^I] \ge 1 - \frac{1}{m}\binom{n}{2} \ge \frac{1}{2}$. Moreover, conditioning on $\mathbf{a}\in A^I$, $\mathbf{a}$ is uniformly distributed in $A^I$. Hence,
\[
\Ex_{\mathbf{a}\sim A^I,\mathbf{h}\sim \calH,\mathbf{x}\sim [n]} \left[ \binom{|\Out_{\mathbf{a},\mathbf{h}}(\mathbf{x})|}{c} \right] \le \frac{1}{\Pr_{\mathbf{a}\sim \calA}[\mathbf{a}\in A^I]}\cdot  \Ex_{\mathbf{a}\sim \calA,\mathbf{h}\sim \calH,\mathbf{x}\sim [n]} \left[ \binom{|\Out_{\mathbf{a},\mathbf{h}}(\mathbf{x})|}{c} \right] \le 2 D_c n^{c/2}.
\]
We choose $C_c = c!\cdot 2^{c+5} \cdot D_c$. Suppose the theorem statement does not hold. Then:
\[
\begin{aligned}
\Ex_{\mathbf{a}\sim A^I,\mathbf{h}\sim \calH,\mathbf{x}\sim [n]} \left[ \binom{|\Out_{\mathbf{a},\mathbf{h}}(\mathbf{x})|}{c} \right]
&= \Ex_{\mathbf{a}\sim A^I}  \Ex_{\mathbf{h}\sim \calH,\mathbf{x}\sim [n]} \left[ \binom{|\Out_{\mathbf{a},\mathbf{h}}(\mathbf{x})|}{c} \right] \\
&\ge \Ex_{\mathbf{a}\sim A^I} \left( \binom{n}{c} \cdot C_c \cdot n^{-c/2} \right) \\
&\ge \frac{C_c}{c!\cdot 2^{c}} \cdot n^{c/2} & \text{(Assume $n\ge 2c$)}\\
&> 2 D_cn^{c/2}.
\end{aligned}
\]
This leads to a contradiction. Therefore, there must be an injective mapping $a\in A^I$ and $c$ vertices $1\le u_1 < \dots < u_c\le n$ for which \eqref{eq:connecting-ub} holds.
\end{proofsketch}

\end{document}